\documentclass[acmsmall,screen,nonacm]{acmart}
\usepackage{cleveref,fontawesome5} 
\usepackage{mathpartir,verbatimbox,pgfplots} 

\DeclareMathOperator{\assgn}{{:\!=}}
\DeclareMathOperator{\iassgn}{{:\!\backsim}}
\DeclareMathOperator{\uassgn}{{\bumpeq}}
\DeclareMathOperator{\hash}{{\diamond}}
\DeclareMathOperator{\apply}{{@}}
\DeclareMathOperator{\bdot}{{\triangleleft}}

\newcommand{\stringlit}[1]{\textrm{\textquotesingle{#1}\textquotesingle}}

\newcommand{\refval}[2]{\ensuremath{{#1}{\circ}{\rec{{#2}}}}}
\newcommand{\refxval}[2]{\ensuremath{{#1}{\circ}{{#2}}}}

\newcommand{\tbn}{\ensuremath{{\color{brown}N}}}
\newcommand{\tbnprime}{\ensuremath{{\color{brown}N'}}}

\newcommand{\tbname}[1]{\ensuremath{\mathsf{\color{brown}{#1}}}}

\newcommand{\lbn}{\ensuremath{{\color{teal}L}}}
\newcommand{\lbnprimex}[1]{\ensuremath{{\color{teal}L'_{\mathit{#1}}}}}
\newcommand{\lbnx}[1]{\ensuremath{{\color{teal}L_{\mathit{#1}}}}}
\newcommand{\lbname}[1]{\ensuremath{\mathsf{\color{teal}{#1}}}}

\newcommand{\idn}{\ensuremath{{\color{olive}\mathit{id}}}}
\newcommand{\idnprime}{\ensuremath{{\color{olive}\mathit{id'}}}}

\newcommand{\idname}[1]{\ensuremath{{\color{olive}{#1}}}}

\newcommand{\dmeta}{W}
\newcommand{\dsetmeta}{\ensuremath{\overline{W}}}
\newcommand{\dsetmetax}[1]{\ensuremath{\overline{W_{#1}}}}
\newcommand{\dsetmetaprime}{\ensuremath{\overline{W'}}}
\newcommand{\dsetmetaprimex}[1]{\ensuremath{\overline{W'_{#1}}}}

\newcommand{\dbn}{\ensuremath{\mu}}
\newcommand{\dbnx}[1]{\ensuremath{\mu_{#1}}}

\newcommand{\many}{\ensuremath{\infty}}

\newcommand{\lrec}{\langle}
\newcommand{\rrec}{\rangle}
\newcommand{\rec}[1]{\ensuremath{\left\langle{#1}\right\rangle}}
\newcommand{\maprec}[2]{\ensuremath{#1} \assgn {#2}}
\newcommand{\imaprec}[2]{\ensuremath{#1} \iassgn {#2}}
\newcommand{\umaprec}[2]{\ensuremath{#1} \uassgn {#2}}
\newcommand{\recsep}{,}

\newcommand{\tmap}[3]{\ensuremath{{#1}:{#2}\hash{#3}}}
\newcommand{\tstring}{\ensuremath{\textit{str}}}
\newcommand{\tint}{\ensuremath{\textit{int}}}
\newcommand{\tbool}{\ensuremath{\textit{bool}}}
\newcommand{\reft}[2]{\ensuremath{{#1}{\circ}{\rec{{#2}}}}}

\newcommand{\schema}{\ensuremath{\Delta}}
\newcommand{\cardmeta}{\ensuremath{m}}
\newcommand{\tmeta}{\ensuremath{T}}
\newcommand{\tsmeta}{\ensuremath{t}}
\newcommand{\vmeta}{\ensuremath{V}}
\newcommand{\vsmeta}{\ensuremath{v}}
\newcommand{\vsetmeta}{\ensuremath{\overline{V}}}
\newcommand{\vssetmeta}{\ensuremath{\overline{v}}}

\newcommand{\vsetmetax}[1]{\ensuremath{\overline{V_{#1}}}}
\newcommand{\vsetmetaprimex}[1]{\ensuremath{\overline{V'_{#1}}}}
\newcommand{\vssetmetax}[1]{\ensuremath{\overline{v_{#1}}}}
\newcommand{\ctt}{\ensuremath{\mathsf{tt}}}
\newcommand{\cff}{\ensuremath{\mathsf{ff}}}
\newcommand{\lkn}{\ensuremath{\kappa}}
\newcommand{\lkopen}{\text{\footnotesize\faLockOpen}}
\newcommand{\lkclosed}{\text{\footnotesize\faLock}}

\newcommand{\card}[2]{\ensuremath{[{{#1},{#2}}]}}

\newcommand{\vctxn}{\ensuremath{\Omega}}
\newcommand{\tctxn}{\ensuremath{\Gamma}}
\newcommand{\emeta}{\ensuremath{e}}
\newcommand{\emp}{\ensuremath{\varnothing}}
\newcommand{\typmeta}{\ensuremath{\tau}}

\newcommand{\eval}[7]{\ensuremath{{#1} \vdash_{{#2}, {#3}} {#6} \parallel {#4} \Rightarrow {#7} \parallel {#5}}}
\newcommand{\seval}[5]{\ensuremath{{#1} \vdash {#4} \parallel {#2} \Rightarrow {#5} \parallel {#3}}}
\newcommand{\deval}[5]{\eval{#1}{\schema}{\dbnx{\textit{init}}}{#2}{#3}{#4}{#5}}

\newcommand{\shapevalmeta}{\ensuremath{R}}

\newcommand{\pone}{\ensuremath{\texttt{\underline{{1}}}}}
\newcommand{\popt}{\ensuremath{\texttt{\underline{{?}}}}}
\newcommand{\pany}{\ensuremath{\texttt{\underline{{*}}}}}

\setcopyright{rightsretained}
\acmDOI{XXXXXXX.XXXXXXX}
\authorsaddresses{}

\setcopyright{none}
\acmDOI{}
\acmISBN{}
\acmYear{Submitted for review to POPL 2026}
\acmConference{Submission Copy}{Anonymous Authors}{\today}

\begin{document}

\title{Querying Graph-Relational Data}

\author{Michael J. Sullivan}
\orcid{0009-0005-3127-317X}
\affiliation{
  \institution{Gel Data}
  \state{California}
  \country{USA}
}
\author{Zhibo Chen}
\affiliation{
  \institution{Carnegie Mellon University}
  \state{Pennsylvania}
  \country{USA}
}
\orcid{0000-0003-0045-5024}
\email{zhiboc@andrew.cmu.edu}

\author{Elvis Pranskevichus}
\affiliation{
  \institution{Gel Data}
  \state{California}
  \country{USA}
}
\author{Robert J. Simmons}
\orcid{0000-0003-2420-3067}
\affiliation{
  \institution{Owl and Crow Productions}
  \state{Massachusetts}
  \country{USA}
}

\author{Victor Petrovykh}
\affiliation{
  \institution{Gel Data}
  \country{Canada}
}
\author{Aljaž Mur Eržen}
\affiliation{
  \institution{Gel Data}
  \state{California}
  \country{Slovenia}
}
\author{Yury Selivanov}
\affiliation{
  \institution{Gel Data}
  \state{California}
  \country{USA}
}
\renewcommand{\shortauthors}{Sullivan, Chen, Pranskevichus et al.}

\begin{abstract}

For applications that store structured data in relational databases, there is an impedance mismatch between the flat representations encouraged by relational data models and the deeply nested information that applications expect to receive. 
In this work, we present the \textit{graph-relational} database model, which  provides a flexible, compositional, and strongly-typed solution to this ``object-relational mismatch.'' We formally define the graph-relational database model and present a static and dynamic semantics for queries.
In addition, we discuss the realization of the graph-relational database model in EdgeQL, a general-purpose SQL-style query language, and the Gel system, which compiles EdgeQL schemas and queries into PostgreSQL queries.
Gel facilitates the kind of object-shaped data manipulation that is frequently provided inefficiently by object-relational mapping (ORM) technologies, while achieving most of the efficiency that comes from writing complex PostgreSQL queries directly.

\end{abstract}

\maketitle

\section{Introduction}

When information needs to be stored reliably and manipulated in efficient and flexible ways, that information frequently gets stored in a relational databases that is queried with SQL \cite{10.1145/362384.362685,6359709}. SQL-powered databases have proven enormously popular, but they have some significant mismatches with how many applications present and manipulate data. As a simple example, consider an application that treats movies as entities that have directors and actors. These directors and actors are people, and people are another kind of entity.

\begin{verbbox}[\footnotesize]
[{ "title": "Transistors",
   "year": 2007,
   "directors": [
     { "name": "Michael Cove",
       "age": 53 }
   ],
   "actors": [
     { "name": "Megan Wolf",
       "@character": "Meg Tech" },
     { "name": "Shy Andbuff",
       "@character": "Sam Man" }
   ]
 }, 
 ...

\end{verbbox}
\begin{figure}[ht]
  \centering
  \includegraphics[width=8.6cm]{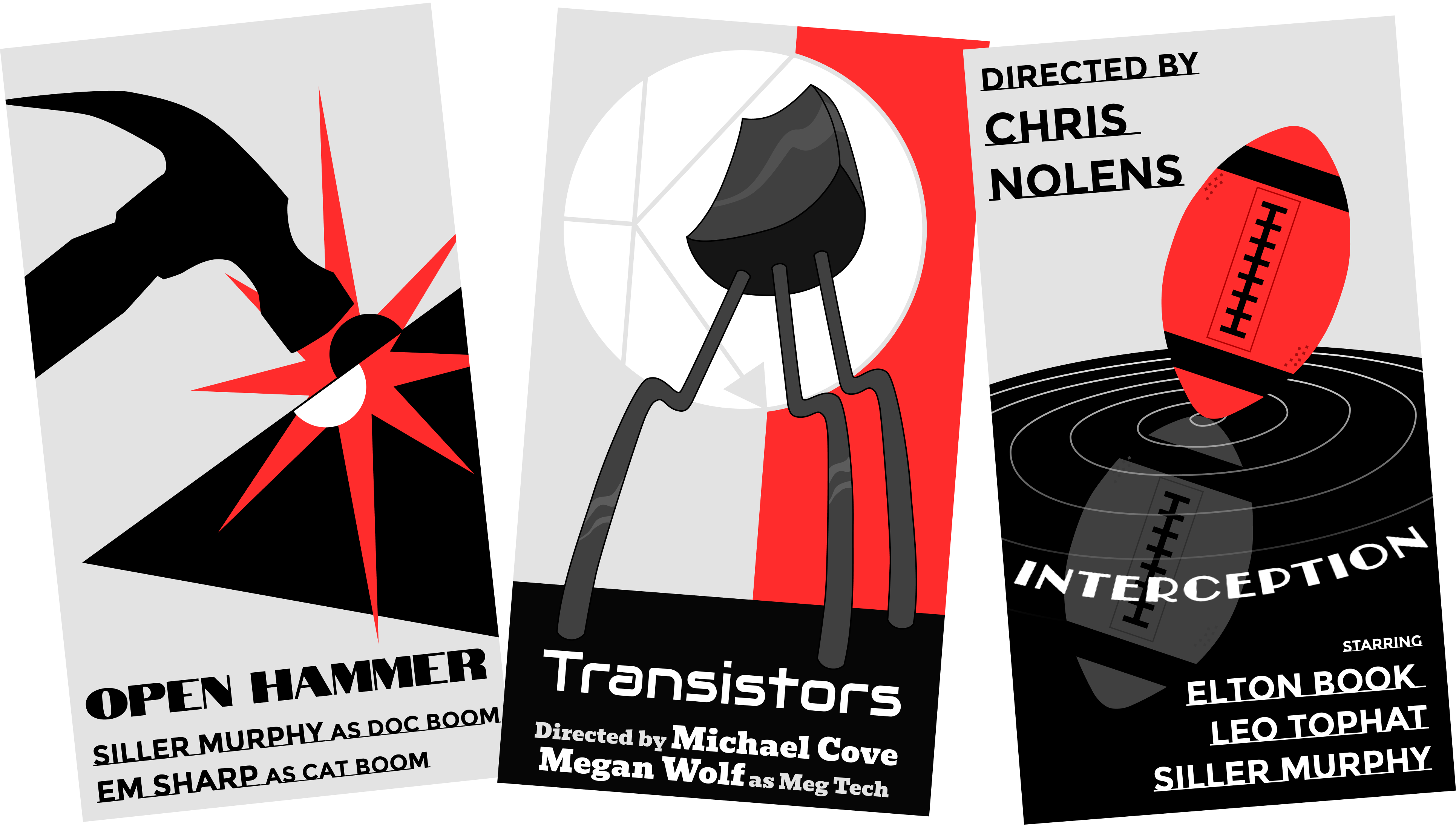}
  \theverbbox
  \caption{Some movie posters and a snippet of a possible JSON representation of those movies}
  \label{figmov}
\end{figure}

\Cref{figmov} invents some movies and presents a JSON structure describing those movies that a web application might expect to receive. This tree-structured JSON representation is generally referred to as an ``object-oriented'' perspective on data, though perhaps ``object-shaped'' is more evocative. Movie entities are presented to the application as the parent object, and they contain child objects which include both information about the entities (a person's name and age) and about the relationship between the entities (the name of the character a person plays in a movie).

In many document-oriented ``NoSQL'' databases like MongoDB \cite{MongoDBwebpage}, information might be stored in precisely the form we see in \cref{figmov}. However, because directors and actors are going to be present across multiple movies, and because biographical details may change over time, application developers often desire to keep a single canonical set of biographical information for each person. Relational databases capture a logical end-point of this desire; \cref{figtable} shows how a relational database might capture information about the movies in \cref{figmov}. Data about entities like movies and people can captured in tables where each row has a \textit{primary key} that is unique among all the rows in that table. Data about relationships between entities are stored in tables like the \texttt{Movie.actor} table, where ``source'' and ``target'' columns contain \textit{foreign keys} that refer to values in the ``id'' columns of the \texttt{Movie} and \texttt{Person} tables, respectively.

\begin{figure}
\centering \small
\begin{tabular}{|l|l|l|}
\multicolumn{3}{l}{\texttt{Movie}}
\\\hline
id & title & year \\\hline
7 & Transistors & 2007 \\
8 & Interception & 2010 \\
9 & Open Hammer & 2024 \\\hline
\multicolumn{3}{l}{} \\
\multicolumn{3}{l}{} \\
\multicolumn{3}{l}{} \\
\multicolumn{3}{l}{}
\end{tabular}
\,
\begin{tabular}{|l|l|l|l|}
\multicolumn{4}{l}{\texttt{Person}}
\\\hline
id & name & age & born \\\hline
  1 & Michael Cove & 60 & The moon \\
  2 & Megan Wolf & 38 & California \\
  3 & Leo Tophat & 50 & New York \\
  4 & Elton Book & 38 & Ottawa \\
  5 & Shy Andbuff & 38 & Kansas \\
  6 & Sillier Murphy & 49 & Ireland \\
\multicolumn{4}{|l|}{\ensuremath{\cdots}}
\end{tabular}
\,
\begin{tabular}{|l|l|l|}
\multicolumn{3}{l}{\texttt{Movie.actor}}
\\\hline
source & target & character \\\hline
  7 & 2 & Meg Tech \\
  7 & 5 & Sam Man \\
  8 & 3 & Corn Cobb \\
  8 & 4 & Spiderface \\
  8 & 6 & Fissure \\
  9 & 6 & Doc Boom \\
\multicolumn{3}{|l|}{\ensuremath{\cdots}}
\end{tabular}
\caption{A partial encoding of movie and actor information from \cref{figmov} into a relational database. The \texttt{Movie.actor} table annotates edges from movies (the source) to people (the target) with the name of the character played by that actor.}
\label{figtable}
\end{figure}

\subsection{An object-relational mismatch}\label{objectrelationalmismatch}

\begin{figure}
\centering \small
\begin{tabular}{|l|l|l|l|l|}
\hline
Movie.id & title & character & Person.id & name \\\hline
7 & Transistors & Meg Tech & 2 & Megan Wolf \\
7 & Transistors & Sam Man & 5 & Shy Andbuff \\
8 & Interception & Corn Cobb & 3 & Leo Tophat \\
8 & Interception & Spiderface & 4 & Elton Book \\
\multicolumn{5}{|l|}{\ensuremath{\cdots}}
\end{tabular}
\caption{A typical representation of the result of querying information about the character-in-a-movie relationship in a relational database by suitably joining the tables in \cref{figtable}.}
\label{figjoin}
\end{figure}

Even in this simple example, we can already see an instance of a well-known mismatch between the way information is stored and the way we want to retrieve it, because relational databases return information in tabular form: a single, homogeneous list of records where each record has same fields with the same types. With SQL, we can simultaneously fetch information about movies and their actors by performing a left join of the \texttt{Movie} and \texttt{Movie.actor} tables and then inner joining the result with the \texttt{Person} table, resulting in a table like the one shown in \cref{figjoin}.\footnote{For a precise PostgreSQL query, see the supplemental materials (\Cref{sec-compilation}).} This approach already lacks a certain parsimony, given that information about the movie is repeated multiple times, but it quickly becomes infeasible if we must simultaneously query information about the directors in order to reproduce all the information in the object-shaped representation from \cref{figmov}.

There are essentially three solutions for this problem. The first is to query the necessary information about movies, and then, for each movie, query all the directors and query all the actors. This is simple but can be wildly inefficient in practice, because it requires an additional query for every occurrence of a person in the output data. Many modern database implementations are tuned towards efficiently returning larger batches of information; an avalanche of small queries will degrade performance.\footnote{SQLite is a notable exception that is designed to efficiently support many small queries \cite{manysmall}.}

The second option uses three queries. The first query gets necessary information about the movies, as before. The second query includes the set of all required movie ids and gets all actors: the result is effectively what was shown in \cref{figjoin} but without the redundant ``title'' information. The third query again includes the set of all required movie ids and gets all directors (using a \texttt{Movie.directors} table similar to the \texttt{Movie.actors} table in \cref{figtable}). This is reasonably efficient in practice, but it forces the application to do a lot of the heavy lifting to correlate actors and directors with movies and create the object-shaped result. The complexity of this correlation increases as we move beyond the simple and shallow relationships from our example.

A third approach is less traditional: some relational databases, like PostgreSQL, allow queries to return nested data by introducing a data type of \textit{arrays of records} in addition to the numbers, dates, and strings implied by \cref{figtable,figjoin}. This has the potential to be quite efficient, but ``arrays of record values'' are not a natural fit for relational databases and the SQL query language: the resulting queries are boilerplate-heavy, unidiomatic, and not standardized across implementations.

Unfortunately, the simple and inefficient behavior we presented as the first option is often the default behavior for Object-Relational Mapper (ORM) libraries, which are the industry-standard way of mediating between object-shaped representations of information and relational databases used to store that information. Most developers who work with ORM libraries need to be aware of this ``n+1 query problem,'' as it can have significant impact on the performance of applications, an impact that often goes unnoticed until the system has grown to hold a significant amount of data. Most ORM libraries provide ways of punching through abstraction layers and providing more direct access to the underlying SQL tools in order to recover acceptable performance.

\subsection{The Gel System and EdgeQL Query Language}
\label{sec-intro-schema-type}

The Gel system was designed as a way of handling the object-relational mismatch while maintaining the benefits of using modern relational databases, including strong type-consistency guarantees and data consistency guarantees. Gel requires that data adhere to a strict schema; our running movie database example uses this schema:
\begin{quote}\small
\begin{verbatim}
type Person { required name: str; required age: int64; born: str; };
type Movie {
  required title: str;
  required year: int64;
  required multi directors: Person;
  multi actors: Person { character: str; }; };
\end{verbatim}
\end{quote}
Gel is implemented on top of the PostgreSQL database to take advantage of its specific nonstandard extensions, and this schema will induce Postgres to create tables like the ones shown in \cref{figtable}. The data shown in  \cref{figmov} can be requested in the EdgeQL query language used by Gel like this:
\begin{quote} \small
\begin{verbatim}
select Movie { title, year, directors: { name, age }, actors: { name, @character }};
\end{verbatim}
\end{quote}
In effect, Gel automates the third approach described above, providing EdgeQL as a convenient, compositional query language for interfacing with the ad-hoc SQL extensions that facilitate working with nested data in PostgreSQL.
It is an absolute requirement of the Gel implementation that all EdgeQL queries be able to be compiled into a single query in PostgreSQL; this has numerous practical benefits, including ensuring that all queries execute as a single database transaction, and has also proven sufficient to ensure that Gel provides acceptable performance relative to ORM approaches and hand-authored PostgreSQL queries.

\subsection{The Graph-Relational Database Model}

This paper presents the \textit{graph-relational database model} as a formal foundation for EdgeQL and Gel.
We call the model ``graph-relational'' because it manipulates graphs made up of nodes (entities, like movies and people) connected by directed edges (relationships between entities, like a person acting in a movie). This graph model maps naturally both onto the object-shaped structures we use to interact with the database and also onto the relational database we use to store data at rest, allowing for a straightforward and compositional expression of object-shaped queries.

In \cref{sec-data-at-rest}, we present a formal model of graph-relational databases, and in \cref{sec-expressions} we describe the syntax and semantics of a query language for graph-relational databases. The query language combines the nested shaping shown in the query above, SQL-style constructs for filtering and ordering, and also functional-programming-style constructs for computation and iteration. The query language enjoys a compositional static and big-step dynamic semantics; we discuss the metatheory of language (totality and type soundness) in \cref{sec-metatheory}.

While we present the graph-relational database model as a standalone programming model, the existence and design of this model is inextricably tied to the realization of this design in the Gel database system, to the design of Gel's EdgeQL query language \cite{GelDocs}, and to the requirement of compiling EdgeQL expressions to single PostgreSQL queries. We discuss the Gel system in \cref{sec-gel}, and present some benchmarking results in \cref{sec-eval}.

\section{Representing graph-relational data}\label{sec-data-at-rest}

To formalize the graph-relational model, we start by describing how we can store information about individual entities in records that map \textit{labels} $\lbn$ to \textit{stored values} $\vmeta$.
Each stored value $\vmeta_i$ is either a \textit{scalar value} (like a string, integer, or boolean) or a \textit{reference value} (which we will explain shortly). 

In the introduction, we saw that relationships between entities can be treated as directional: the ``actor'' and ``director'' relationships have a source (the movie) and a target or destination (the person). The relational database presentation sketched in the introduction treats source and target somewhat symmetrically, whereas our graph-relational model associates relationships more closely with their source entity. (In \cref{secrefvalues} we will see how this symmetry-breaking does not cost us expressive power: it remains possible to query relationships in both directions, and the Gel implementation allows this to be done efficiently.)

In our running example, ``person'' entities aren't the source of any relationships. As a result, their representation is relatively straightforward:
\begin{align*}
R_\textit{mc} & = 
    \lrec && 
     \maprec{\lbname{name}}{[\stringlit{Michael Cove}]}
     \recsep\;
     \maprec{\lbname{age}}{[60]}
     \recsep\;
     \maprec{\lbname{born}}{[\stringlit{The moon}]}
    \qquad \rrec
\\
R_\textit{mw} & = 
    \lrec && 
     \maprec{\lbname{name}}{[\stringlit{Megan Wolf}]}
     \recsep\;
     \maprec{\lbname{age}}{[38]}
     \recsep\;
     \maprec{\lbname{born}}{[\stringlit{California}]}
    \qquad \rrec
\\
R_\textit{lt} & = 
    \lrec && 
     \maprec{\lbname{name}}{[\stringlit{Leo Tophat}]}
     \recsep\;
     \maprec{\lbname{age}}{[50]}
     \recsep\;
     \maprec{\lbname{born}}{[\stringlit{New York}]}
    \qquad \rrec
\intertext{
Movies are the source of the ``actor'' and ``director'' relationships, so the information about these relationships gets stored with the description of movie:
}
R_\textit{tr} & = 
    \lrec &&
      \maprec{\lbname{title}}{[\stringlit{Transistors}]}
      \recsep \;
      \maprec{\lbname{year}}{[2007]}
      \recsep \;
      \maprec{\lbname{directors}}{[\refval{\idname{1}}{}]}
      \recsep \\
    &&&
      \maprec{\lbname{actors}}{[
        \refval{\idname{2}}{\imaprec{\lbname{character}}{[\stringlit{Meg Tech}]}}
        \recsep
        \refval{\idname{5}}{\imaprec{\lbname{character}}{[\stringlit{Sam Man}]}}
      ]}
    \qquad \rrec
\end{align*}
The reference value $\refval{\idname{2}}{\imaprec{\lbname{character}}{[\stringlit{Meg Tech}]}}$ captures that the character named ``Meg Tech'' in the movie \textit{Transistors} is played by the person with the unique identifier $\idname{2}$. (We will call {\lbname{character}} a \textit{link property}; it is a property of the {\lbname{actors}} link between two entities, not a property of either entity.) In order for a reference value to make sense, it must exist in a context of a \textit{database store} {\dbn}, a set of tuples that uniquely map identifiers to entity records. For our running example, that would look something like this:
\[ \dbnx{\textit{example}} = \{\;( \idname{7}, \tbname{Movie}, R_\textit{tr}), \, 
      ( \idname{1}, \tbname{Person}, R_\textit{mb}), \, 
      ( \idname{2}, \tbname{Person},  R_\textit{mw}), \,
      ( \idname{3},\tbname{Person},  R_\textit{lt}), \, \cdots\; \} \]
Unlike the relational database modeled in \cref{figtable}, we will require \textit{global} uniqueness of identifiers, a requirement handled in practice by generating all identifiers as UUIDs with an astronomically low probability of collision.

\subsection{Cardinality modes}

Perhaps the biggest difference between the traditional and graph-relational presentation is that, in the graph-relational presentation, labels are associated with \textit{sequences} of values, not individual values. This means that we can easily give a person multiple names, multiple ages, or zero names or ages. This may not always be desirable!

The novel inclusion of \textit{cardinality modes} in our type system tames this generality, allowing a user to specify properties like ``an actor has at least one name and exactly one age.''  This is crucial for the usability of our system: it connects the uniform model of graph-relational data (where things like movie titles are \emph{sequences} of strings) to the object-shaped representations that consumers of graph-relational data expect. (We will return to the topic of serialization in \cref{sec-type-directed-serial}.)

A cardinality mode takes the form $\card{i_1}{i_2}$ where $i_1 \in \{0, 1\}$ is a lower bound, $i_2 \in \{0, 1,\many\}$ is an upper bound, and $i_1 \leq i_2$. A set or sequence with cardinality mode $[i_1,i_2]$ contains at least $i_1$ elements and at most $i_2$ elements (inclusive). The five possible cardinality modes are:
\begin{enumerate}
\item The \textit{required single} mode $\card{1}{1}$ could be used to ensure that every movie has one title. 
\item The \textit{optional single} mode $\card{0}{1}$ represents the optional presence of an element, and could be used for a birthplace field in a setting where someone's birthplace might not be known.
\item The \textit{required multi} mode $\card{1}{\many}$ represents the certain presence of at least one element, and could be used to force every movie to have at least one director. 
\item The \textit{optional multi} mode $\card{0}{\many}$ is the totally unconstrained cardinality, and could be used for actors, since some movies do not have actors.
\item The \textit{empty} mode $\card{0}{0}$ has limited utility but provides an additive unit that makes cardinality modes a semiring with the ordering $\card{i_1}{i_2} \leq \card{i_3}{i_4}$ when $i_3 \leq i_1$ and $i_2 \leq i_4$, multiplication $\card{i_1}{i_2} \times \card{i_3}{i_4} = [i_1 \times i_3, i_2 \times i_4]$ with the unit $\card{1}{1}$, and addition $\card{i_1}{i_2} + \card{i_3}{i_4} = [\lfloor i_1 + i_3\rfloor, \lceil i_2 + i_4 \rceil]$, where the rounding is to the next element in $\{ 0, 1, \many \}$. 
\end{enumerate}
Gel's concrete syntax for the first four modes was shown in the schema from \cref{sec-intro-schema-type}. The optional single mode is the unannotated default, and the empty mode does not have concrete syntax in Gel.

The utility of ``null'' values in relational databases has been an topic of study for over four decades \cite{ZANIOLO1984142}, yet existing approaches continue to vex users in the present day \cite{10.14778/3551793.3551818}. The graph-relational database model suggests a different perspective: we can treat all fields uniformly as multiset-valued and recover the traditional cases of ``required'' and ``nullable'' as two of five possible constraints on the cardinality of a field. Relationships described as \textit{many-to-many} are generally treated as a different special case in relational databases; in contrast, the graph-relational database model treats these as instances of the cardinality mode $[0,\many]$. We are optimistic that this shift in perspective can improve clarity and user outcomes, though a proper investigation of user outcomes is left for future work.

\begin{figure}
\begin{center}
\begin{tabular}{lrll}
$\lbn$ & $::=$ & 
  $\ldots$ & 
  labels\\
$\tbn$ & $::=$ &
  $\ldots$ &
  type names\\
$\idn$ & $::=$ & 
  $\ldots$ &
  unique identifiers\\
$\cardmeta$ & $::=$ & 
  $[0,0] \mid [0,1] \mid [0, \many] \mid [1,1] \mid [1, \many]$ & 
  cardinality modes\\
$\tsmeta$ & $::=$ &
  $\tint \mid \tstring \mid \tbool \mid \ldots$ &
  scalar types\\
$\vsmeta$ & $::=$ &
  $\ctt \mid \cff \mid \stringlit{hi} \mid 3 \mid  \ldots$ &
  scalar values\\ 
$\vssetmeta$ & $::=$ &
  $[\vsmeta_1, \ldots, \vsmeta_n]$ &
  scalar value sequences\\ 
$\tmeta$ & $::=$ &
  $\tsmeta \mid \reft{\tbn}{\tmap{\lbnx{1}}{\tsmeta_1}{\cardmeta_1}\recsep\ldots\recsep\tmap{\lbnx{n}}{\tsmeta_n}{\cardmeta_n}}$ &
  stored value types\\
$\vmeta$ & $::=$ &
  $\vsmeta \mid \refval{\idn}{\imaprec{\lbnx{1}}{\vssetmetax{1}}\recsep\ldots\recsep\imaprec{\lbnx{n}}{\vssetmetax{n}}}$ & 
  stored values\\
$\vsetmeta$ & $::=$ &
  $[\vmeta_1, \ldots, \vmeta_n]$ &
  stored value sequences\\
$\lkn$ & $::=$ & 
  $\lkopen \mid \lkclosed$ &
  edit marks\\
\end{tabular}
\end{center}
\caption{Syntax of data at rest in the graph-relational database model}
\label{fig-syn-rest}
\end{figure}

\subsection{Types and schemas}\label{sec-schemas}

A database schema maps \textit{type names} {\tbn} to \textit{object types}, and 
object types assign types and cardinality modes to a set of labels. For object types that are not the source entity for any relationships, all labels are {properties} that map to scalar types: in our running example, if we want each person to have exactly one name, one age, and at most one birthplace, $\tbname{Person}$ will be associated with this object type:
\begin{align*}
\lrec\quad & \tmap{\lbname{name}}{\tstring}{[1,1]}
\recsep\;
\tmap{\lbname{age}}{\tint}{[1,1]}
\recsep\; 
\tmap{\lbname{born}}{\tstring}{[0,1]}
\quad \rrec
\intertext{If movies have one title and year, at least one director, and any number of actors, $\tbname{Movie}$ will be associated with this object type:}
\lrec \quad &
  \tmap{\lbname{title}}{\tstring}{[1,1]}
\recsep \\ &
  \tmap{\lbname{year}}{\tint}{[1,1]}
\recsep \\ &
  \tmap{\lbname{directors}}{\reft{\tbname{Person}}{}}{[1,\many]},
\recsep \\ &
  \tmap{\lbname{actors}}
    {\reft{\tbname{Person}}{\tmap{\lbname{character}}{\tstring}{[1,1]}}}
    {[0,\many]}
\quad \rrec
\end{align*}
Labels like {\lbname{directors}} that map to reference types are \textit{links}, {\lbname{age}} and {\lbname{title}} are \textit{properties}.

Formally, we will start with the definitions in \cref{fig-syn-rest} and define database schemas {\schema} as sets of tuples $\left(\tbn, \rec{\tmap{\lbnx{1}}{\tmeta_1}{\cardmeta_1}\recsep\ldots\recsep\tmap{\lbnx{n}}{\tmeta_n}{\cardmeta_n}}\right)$  where any given $\tbn$ only appears in a single tuple; this allows us to treat a schema as a partial map and write $\schema(\tbn) = \rec{\tmap{\lbnx{1}}{\tmeta_1}{\cardmeta_1}\recsep\ldots\recsep\tmap{\lbnx{n}}{\tmeta_n}{\cardmeta_n}}$ when convenient. We will similarly define define database stores {\dbn} as sets of tuples $\left(\idn, \tbn, \lkn, \rec{\maprec{\lbnx{1}}{\vsetmetax{1}}\recsep\ldots\recsep\maprec{\lbnx{n}}{\vsetmetax{n}}}\right)$ where tuples are uniquely identified by their $\idn$ component. Whenever we use the angle bracket notation for record-like syntax, we will treat the uniqueness of labels $\lbnx{i}$ within the record as a syntactic requirement.


\paragraph{Edit marks} The edit mark $\lkn$ is the only element of the formal syntax that we did not encounter in our motivating discussion, and it is the first element of the graph-relational model we encounter that is unambiguously influenced by the constraints of the PostgreSQL implementation. It's possible for an EdgeQL query to attempt to update the same entity multiple times, but in our compiled queries, only the first update to an entity will succeed, and subsequent updates will be ignored. When we present our operational semantics we can treat the database store as starting with all records unlocked ($\lkopen$), and these edit marks will keep track of edits over the course of a single query to capture the behavior of PostgreSQL by preventing multiple updates of the same record.

\paragraph{Typing scalar values}

We leave scalar types and values open-ended. Integers and strings are included for the sake of examples, and a Boolean type $\tbool$ with the two values
$\ctt$ and $\cff$ is needed for the later introduction of conditionals. We only need to require that the primitive typing judgment $\vsmeta : \tsmeta$ identifies each scalar value with at most one type.

\paragraph{Typing database stores}

\begin{figure}
\begin{flushleft}\fbox{$\strut\vdash_{\schema, \dbn}\strut \vsetmeta : \tmeta \hash m$}\end{flushleft}
\[
\inferrule[(ST-set)]
{i_\mathit{lo} \leq n \leq i_\mathit{hi} \\ 
  \vdash_{\schema, \dbn} \vmeta_1 : \tmeta \quad 
  \cdots \quad 
  \vdash_{\schema, \dbn} \vmeta_n : \tmeta}
{\vdash_{\schema, \dbn} \tmap{[\vmeta_1, \ldots, \vmeta_n]}{\tmeta}{[i_\mathit{lo}, i_\mathit{hi}]}}
\]

\begin{flushleft}\fbox{$\strut\vdash_{\schema, \dbn}\strut \vmeta : \tmeta$}\end{flushleft}
\[
\inferrule[(ST-prim)]
{\vsmeta : \tsmeta}
{\vdash_{\schema, \dbn} \vsmeta : \tsmeta}
\qquad
\inferrule[(ST-ref)]
{\left(\idn, \tbn, \lkn, \rec{\ldots}\right) \in \dbn
\\
\vdash_{\schema, \dbn} \vssetmetax{1} : \tsmeta_1 \hash \cardmeta_1
 \quad \cdots
 \quad \vdash_{\schema, \dbn} \vssetmetax{n} : \tsmeta_n \hash \cardmeta_n}
{
\vdash_{\schema, \dbn} 
  \refval{\idn}{\imaprec{\lbnx{1}}{\vssetmetax{1}}\recsep\ldots\recsep\imaprec{\lbnx{n}}{\vssetmetax{n}}}
  : \reft{\tbn}{\tmap{\lbnx{1}}{\tsmeta_1}{\cardmeta_1}\recsep\ldots\recsep\tmap{\lbnx{n}}{\tsmeta_n}{\cardmeta_n}}}
\]
\caption{Typing rule for values and value sets: the premises of \textsc{ST-ref} take advantage of the fact that scalar value sets $\vssetmeta$ are a syntactic refinement of value sets $\vsetmeta$, and likewise for $\tsmeta$ and $\tmeta$. (Alternatively, we could have copied the first judgment to create a similarly-defined judgment $\vdash_{\schema, \dbn}\strut \vssetmeta : \tsmeta \hash m$.)}
\label{figvaluetypes}
\end{figure}
The judgment $\vdash \schema \,\mathtt{ok}$ holds when all names $\tbn$ referenced in the range of $\schema$ are defined by $\schema$, and when the set of labels can be partitioned into \textit{object labels}, the properties and links that appear in the top-level records in the schema, and \textit{link property labels} which appear in the reference types $\tmeta$ within a schema.\footnote{In the EdgeQL implementation, this partitioning of labels is ensured by forcing link property labels to be prefixed with an `@` symbol, but it is simpler in our formal presentation to treat labels uniformly.}
The judgment $\vdash_\schema \dbn \,\mathtt{ok}$ holds when $\vdash \schema \,\mathtt{ok}$ and when, for each $\left(\idn, \tbn, \lkn, \rec{\maprec{\lbnx{1}}{\vsetmetax{1}}\recsep\ldots\recsep\maprec{\lbnx{n}}{\vsetmetax{n}}}\right) \in \dbn$, we have $\schema(\tbn) = \rec{\tmap{\lbnx{1}}{\tmeta_1}{\cardmeta_n}\recsep\ldots\recsep\tmap{\lbnx{n}}{\tmeta_n}{\cardmeta_n}}$ and for all $1 \leq i \leq n$ we have $\vdash_{\schema, \dbn}{\vsetmetax{i} : \tmeta_i \hash \cardmeta_i}$ as defined in \cref{figvaluetypes}.

\subsection{Computed Values}
\label{sec-general-values}

The previous sections described the representation of graph-relational data ``at rest'' in storage, and this data has a limited hierarchy: a database entry maps labels to stored values $\vmeta$, one kind of stored value is a reference value that both identifies another database entry and maps link properties $\lbn$ to scalar values $\vsmeta$, and scalar values cannot reference the store. Stored value types $\tmeta$ and scalar types $\tsmeta$ have an analogous relationship.

The (sequences of) values returned by queries need to have a more general recursive structure.
Consider a schema for a social-networking-like application where $\tbname{User}$ has the object type $\rec{\tmap{\lbname{name}}{\tstring}{[1,1]}\recsep \tmap{\lbname{following}}{\reft{\tbname{User}}{}}{[0,\many]}}$.
Under such a schema, the database store might contain an element like $\left(\idname{16}, \tbname{User}, \lkn, \rec{\maprec{\lbname{name}}{[\stringlit{Alice}]}
, \maprec{\lbname{following}}{[\refval{\idname{12}}{}, \refval{\idname{19}}{}]}}\right)$. We definitely want it to be possible for object-shaped queries to allow access to information about Alice that isn't directly present in this record, for example the names of the users Alice is following (and potentially the names of users Alice's followers are following).

The solution is to introduce types $\typmeta$ and computed values $\dmeta$ as straightforward recursive generalizations of property types $\tmeta$ and stored values $\vmeta$, as shown in \cref{fig-recursive-types-and=data}. Shaped expressions introduced in \cref{sec-exp-shapes} will allow additional information to be ``loaded'' into computed values, and only the information that is present in a computed value will be serialized in the output of a query. The only twist is the presence of \textit{record entry visibility marks}, which exclude certain property-to-value mappings from serialization; these markings only influence serialization and have very little effect on the semantics, so will be mostly ignored until \cref{sec-gel}.

\begin{figure}
\begin{tabular}{lrll}
$\typmeta$ & $::=$ &
  $\tsmeta \mid \reft{\tbn}{\tmap{\lbnx{1}}{\typmeta_1}{\cardmeta_1}\recsep\ldots\recsep\tmap{\lbnx{n}}{\typmeta_n}{\cardmeta_n}}$ &
  types\\
$\dmeta$ & $::=$ &
  $\vsmeta \mid \refxval{\idn}{\shapevalmeta}$ & 
  computed values\\
$\dsetmeta$ & $::=$ &
  $[\dmeta_1, \ldots, \dmeta_n]$ &
  computed value sequences \\
$\shapevalmeta$ & $::=$ &
  $\rec{\umaprec{\lbnx{1}}{\dsetmetax{1}}
  \recsep
  \ldots
  \recsep
  \umaprec{\lbnx{n}}{\dsetmetax{n}}
  }$ &
  shape values \\
$\uassgn$ & $::=$ &
  ${\assgn} \mid {\iassgn}$&
  record entry visibility mark\\
\end{tabular}
\caption{Syntax of computed types and computed values}
\label{fig-recursive-types-and=data}
\end{figure}

\section{Expressions as queries on graph-relational data}
\label{sec-expressions}

EdgeQL is a query language, but in this paper we will call graph-relational queries \textit{expressions} in order to differentiate them from SQL queries. A client retrieves and modifies data in the graph-relational database model by evaluating  expressions {\emeta}  to sequences of computed values $\dsetmeta = [\dmeta_1, \ldots, \dmeta_n]$. In this section, we introduce the EdgeQL expression language in stages: basic expressions in \cref{secebasic,secrefvalues} will be extended with \textit{shapes} in \cref{sec-exp-shapes}, built-in functions in \cref{sec-exp-builtin}, control features in \cref{seceprog}, and expressions that manipulate the database in \cref{secemutate}.

EdgeQL expressions are strongly typed, and the core judgment describing the static semantics of EdgeQL expressions expresses that, given a schema {\schema}, an expression {\emeta} may be given the type {\typmeta} and the cardinality mode {\cardmeta} in a context {\tctxn} that assigns types and cardinality modes to the free variables in {\emeta}:
\[ \tctxn \vdash_\schema \tmap{\emeta}{\typmeta}{\cardmeta} \]

Evaluation can change the database store, so in addition to the starting expression and the computed sequence of values, the big-step operational semantics needs to carry a ``before'' and ``after'' database store. This suggests a judgment that looks like $\emeta \parallel \dbnx{1} \Rightarrow \dsetmeta \parallel \dbnx{2}$, but there are three additional considerations.
\begin{enumerate}
    \item The language has variable binding, necessitating environments {\vctxn} mapping expression variables to computed value sequences.
    \item It is convenient for our definitions if we have access to the schema {\schema} during evaluation.
    \item EdgeQL expressions access the database store that existed \textit{prior} to expression evaluation, not the database store being modified as part of evaluation, so the evaluation judgment needs to carry the initial database state $\dbnx{init}$ as a parameter. This is another decision forced by the PostgreSQL implementation and the goal of compiling EdgeQL expressions to single PostgreSQL queries: PostgreSQL queries can only access the database state that exists before a query's evaluation.
\end{enumerate}
These considerations give us the following evaluation judgment:
\[ \deval{\vctxn}{\dbnx{1}}{\dbnx{2}}{\emeta}{\dsetmeta} \]
When we write rules that do not access the schema {\schema} or initial database {\dbnx{\textit{init}}, we elide  both.

\subsection{Basic Expressions}\label{secebasic}

\begin{figure}
\begin{tabular}{lrll}
$\emeta$ & $::=$ &
  $x \mid \vsmeta \mid \varnothing_\typmeta \mid \emeta_1 \cup \emeta_2  \mid \tbn \mid \ldots$ &
  expressions \\
$\tctxn$ & $::=$ &
  $\tmap{x_1}{\typmeta_1}{\cardmeta_1}, \ldots, \tmap{x_n}{\typmeta_n}{\cardmeta_n}$
  & type contexts (variables must be distinct) \\
$\vctxn$ & $::=$ &
  ${x_1} \mapsto \dsetmetax{1}, \ldots, {x_n} \mapsto \dsetmetax{n}$
  & environments (variables must be distinct)
\end{tabular}

\begin{mathpar}

\inferrule[(T-var)]
{ \tmap{x}{\typmeta}{\cardmeta} \in \tctxn}
{\tctxn \vdash_{\schema} \tmap{x}{\typmeta}{\cardmeta}}

\inferrule[(T-prim)]
{\vsmeta : \tsmeta}
{\tctxn \vdash_{\schema} \vsmeta : \tsmeta \hash [1,1]}

\inferrule[(T-emp)]
{ }
{\tctxn \vdash_{\schema} \varnothing_\typmeta : \typmeta \hash [0,0]}

\\\\

\inferrule[(T-union)]
{
\tctxn \vdash_{\schema} \tmap{\emeta_1}{\typmeta}{\cardmeta_1}
\\
\tctxn \vdash_{\schema} \emeta_2 : \typmeta \hash \cardmeta_2
}
{\tctxn \vdash_{\schema} \emeta_1 \cup \emeta_2 : \typmeta \hash \cardmeta_1 + \cardmeta_2}

\inferrule[(T-name)]
{\tbn \in \mathrm{dom}(\schema)}
{\tctxn \vdash_{\schema} \tmap{\tbn}{\reft{\tbn}{}}{[0,\many]}}
\\\\

\inferrule[(E-var)]
{x \mapsto \dsetmeta \in \vctxn}
{\seval{\vctxn}{\dbn}{\dbn}{x}{\dsetmeta}}

\inferrule[(E-prim)]
{ }
{\seval{\vctxn}{\dbn}{\dbn}{\vsmeta}{[\vsmeta]}}

\inferrule[(E-emp)]
{ }
{\seval{\vctxn}{\dbn}{\dbn}{\varnothing_\typmeta}{[]}}

\\\\

\inferrule[(E-union)]
{
\seval{\vctxn}{\dbn_a}{\dbn_b}{\emeta_1}{\dsetmetax{1}}
\\\\
\seval{\vctxn}{\dbn_b}{\dbn_c}{\emeta_2}{\dsetmetax{2}}
\\
\left[ \dsetmetax{1} , \dsetmetax{2} \right] \approx \dsetmeta
}
{\seval{\vctxn}{\dbn_a}{\dbn_c}{\emeta_1 \cup \emeta_2}{\dsetmeta}}

\inferrule[(E-name)]
{\bigl[\{\refval{\idn}{}\} \mid (\idn, \tbn, \lkn, \rec{\ldots}) \in \dbnx{init}\bigr] \approx \dsetmeta}
{\deval{\vctxn}{\dbn}{\dbn}{\tbn}{\dsetmeta}}

\end{mathpar}

\caption{Syntax and semantics of basic expressions}
\label{figsembasic}
\end{figure}

In a database query formalism where expressions evaluate to sequences, the most fundamental operations are the ones that build sequences and the ones that access database records.

The empty set expression $\emp_\typmeta$ evaluates to an empty sequence $[]$, a literal mention of a scalar value $\vsmeta$ is an expression that evaluates to a singleton sequence $[\vsmeta]$, and union $\emeta_1 \cup \emeta_2$ evaluates its sub-expressions to two sequences and then evaluates nondeterministically to some permutation of the concatenation of those sequences. The expression $3 \cup 4 \cup 4 \cup \varnothing_\tint$ therefore evaluates to the sequence $[3,4,4]$, the sequence $[4,3,4]$, or the sequence $[4,4,3]$.

A named type {\tbn} evaluates to a sequence with one element $\refval{\idn}{}$ for each tuple of that type in the database. In our running example, $\tbname{Movie}$ would evaluate to some permutation of $[\refval{\idname{7}}{}, \refval{\idname{8}}{}, \refval{\idname{9}}{}]$, where the three identifiers reference the three movies in our database.

The syntax and static and dynamic semantics of these basic expressions is given in  \cref{figsembasic}. Many evaluation rules boil down to underlying operations on sequences: the judgment $\left[ \dsetmetax{1} , \dsetmetax{2} \right] \approx \dsetmeta$ in \textsc{E-union} holds whenever $\dsetmeta$ is some permutation of $\dsetmetax{1}$ concatenated with $\dsetmetax{2}$. (In general, we will use $\approx$ to describe sequences equivalent up to permutation.) This high degree of nondeterminism will be a common pattern: when we don't care about the order of results, this nondeterminism allows the underlying database engine to return results in whatever order is most efficient.

The judgment form $\left[ F_P \mid P \in F \right] \approx \dsetmeta$ introduced in \textsc{E-name} is a judgment describing filtered sequence (or set) comprehension: for each element in a finite set or sequence $F$ that matches the pattern $P$, we can define $F_P$, which is either a sequence $\dsetmetax{P}$ or a set that is then coerced into a sequence. The resulting $\dsetmeta$ is some permutation of all those $F_P$ sequences concatenated. In functional programming terms, the judgment represents a filter, then a map, then a concatenation, and then a nondeterministic permutation.

\subsection{Projection}\label{secrefvalues}

\begin{figure}
\begin{tabular}{lrll}
$\emeta$ & $::=$ & $\ldots \mid e \cdot \lbn \mid \emeta \bdot \lbn [\tbn] \mid \ldots$ &
  expressions (continued)
\end{tabular}

\begin{mathpar}

\inferrule[(T-proj-db)]
{
\tctxn \vdash_{\schema} \tmap{e}{\reft{\tbn}{\tmap{\lbnx{1}}{\typmeta_1}{\cardmeta_1},\ldots,\tmap{\lbnx{n}}{\typmeta_n}{\cardmeta_n}}}{\cardmeta'}
\\\\
\lbn \notin \{ \lbnx{1}, \ldots, \lbnx{n} \}
\\
\schema(\tbn) = \rec{\ldots\recsep\tmap{\lbn}{\tmeta}{\cardmeta}\recsep\ldots}
}
{\tctxn \vdash_{\schema} \tmap{e \cdot \lbn}{\tmeta}{\cardmeta \times \cardmeta'}}

\inferrule[(T-proj-ext)]
{
\tctxn \vdash_{\schema} \tmap{e}{\reft{\tbn}{\ldots\recsep\tmap{\lbn}{\typmeta}{\cardmeta}\recsep\ldots}}{\cardmeta'}
}
{\tctxn \vdash_{\schema} \tmap{e \cdot \lbn}{\typmeta}{\cardmeta \times \cardmeta'}}
\\\\
\inferrule[(T-backlink)]
{
\schema(\tbn) = \rec{\ldots, \tmap{\lbn}{\tmeta}{\cardmeta}\recsep\ldots}
\\
\tmeta = \reft{\tbnprime}{\tmap{\lbnx{1}}{\tsmeta_1}{\cardmeta_1}\recsep\ldots\recsep\tmap{\lbnx{n}}{\tsmeta_n}{\cardmeta_n}}
\\
\tctxn \vdash_{\schema} \tmap{e}{\reft{\tbnprime}{\ldots}}{\cardmeta'}
}
{\tctxn \vdash_{\schema} \tmap{\emeta \bdot \lbn [\tbn]}{
\reft{\tbn}{\tmap{\lbnx{1}}{\tsmeta_1}{\cardmeta_1}\recsep\ldots\recsep\tmap{\lbnx{n}}{\tsmeta_n}{\cardmeta_n}}
}{\card{0}{\infty}}}
\\\\

\inferrule[(E-proj)]
{
\deval{\vctxn}{\dbn_a}{\dbn_b}{\emeta}{\dsetmetax{e}}
\\\\
 \left[ \mathsf{project}(\dbn_\textit{init}, \lbn, \dmeta) \mid \dmeta \in \dsetmetax{e} \right] \approx \dsetmetax{b}
}
{\deval{\vctxn}{\dbn_a}{\dbn_b}{\emeta \cdot \lbn}{\dsetmetax{b}}}

\inferrule[(E-backlink)]
{
\deval{\vctxn}{\dbn_a}{\dbn_b}{\emeta}{\dsetmetax{e}}
\\\\
 \left[ \textsf{seek}(\dbn_\textit{init}, \tbn, \lbn, \idn) \mid \refval{\idn}{\ldots} \in \dsetmetax{e} \right] \approx \dsetmetax{b}
}
{\deval{\vctxn}{\dbn_a}{\dbn_b}{\emeta \bdot \lbn [\tbn]}{\dsetmetax{b}}}

\end{mathpar}

\begin{align*}
 \mathsf{project}(\dbn, \lbn, \dmeta) & = \begin{cases}
  \dsetmetax{L} & \text{if } W = \refval{\idn}{\ldots\recsep \umaprec{\lbn}{\dsetmetax{L}}\recsep\ldots} \\
  \vsetmeta & \text{otherwise, if } W = \refval{\idn}{\ldots} \text{ and } \left(\idn, \tbn, \lkn, \rec{{\ldots\recsep\maprec{\lbn}{\vsetmeta}},\ldots}\right) \in \dbn 
    \end{cases}
\\
\mathsf{seek}(\dbn, \tbn, \lbn, \idn) &=
\left\{
\refxval{\idnprime}{\shapevalmeta}
\mid
\left(
\idnprime, \tbn, \lkn, \rec{\ldots\recsep\maprec{\lbn}{\vsetmeta}\recsep\ldots}
\right)
\in \dbn
\wedge
\refxval{\idn}{\shapevalmeta} \in \vsetmeta
\right\}
\end{align*}

\caption{Syntax and semantics of projection}
\label{figsemproj}
\end{figure}

The rule \textsc{E-name} in \cref{figsembasic} returns sequences of object values, but does not describe any elimination forms for these values. The simplest way of using object values is by projecting labels from them. In our running example, if $e_\mathit{tr}$ is an expression that evaluates to the singleton object value $\refval{\idname{7}}{}$ referencing the movie \textit{Transistors}, we have the following:
\begin{align*}
 \tbname{Movie} \cdot \lbname{title} & \approx \left[ \stringlit{Transistors},  \stringlit{Interception}, \stringlit{Open Hammer} \right] \\
 e_\textit{tr} \cdot \lbname{title} & = \left[ \stringlit{Transistors}\right] \\
 e_\textit{tr} \cdot \lbname{actors} \cdot \lbname{character} & \approx \left[ \stringlit{Meg Tech}, \stringlit{Sam Man} \right] \\
 e_\textit{tr} \cdot \lbname{directors} & = \left[\refval{\idname{1}}{}\right]
 \\
 e_\textit{tr} \cdot \lbname{directors} \cdot \lbname{name} & = \left[{\stringlit{Michael Cove}}\right]
\intertext{From these examples, it should be apparent that projecting a property like {\lbname{title}} from a reference $\refval{\idname{1}}{}$ or a link like $\lbname{directors}$ from a reference $\refval{\idname{3}}{}$ will necessarily involve looking up the corresponding records in the database store (written respectively as $R_\mathit{tr}$ and $R_\mathit{mc}$ when we set up our running example in the introduction to \Cref{sec-data-at-rest}). The link property {\lbname{character}} is not stored in the database store entry for a person, so when we evaluate $e_\textit{tr} \cdot \lbname{actors}$, the information about which character an actor played must be stored in the projected object value in order for us to later access it:}
 e_\textit{tr} \cdot \lbname{actors} & \approx \left[
 \refval{\idname{2}}{\imaprec{\lbname{character}}{[\stringlit{Meg Tech}]}}
        ,
        \refval{\idname{5}}{\imaprec{\lbname{character}}{[\stringlit{Sam Man}}]}
 \right]
\end{align*}
This behavior of object values --- storing values internally to the object value and to preferring to project the internal value over any value in the database store --- is captured in \cref{figsemproj} by the rule \textsc{E-proj} and the two cases of the helper function $\mathsf{project}(\dbn, \lbn, \dmeta)$.

A key advantage of the graph-relational model over NoSQL models is that it's possible to cleanly and (in the context of the Gel implementation) efficiently query the reverse of relationships; after all, these relationships are ultimately captured in PostgreSQL link tables with multiple indexed columns containing foreign keys, as illustrated in \cref{figtable}. The backlink expression $\emeta \bdot \lbn [\tbn]$ captures these reverse lookups; if
$\emeta_\textit{\textit{cn}}$ evaluates to the singleton object value $\refval{\idname{11}}{}$ for the person named ``Christopher Nolens,'' a backwards projection of {\lbname{directors}} finds all the movies where that person was the director:
\begin{align*}
 \emeta_\textit{cn} \bdot \lbname{director} [\tbname{Movie}]
 & \approx \left[ \refval{\idname{8}}{}, \refval{\idname{9}}{} \right] \\
 \emeta_\textit{cn} \bdot \lbname{director} [\tbname{Movie}] \cdot \lbname{title}
 & \approx \left[ \stringlit{Interception}, \stringlit{Open Hammer} \right]  \\
  \emeta_\textit{cn} \bdot \lbname{actor} [\tbname{Movie}] & = []
 \intertext{If $\emeta_{\textit{sm}}$ evaluates to the singleton object value $\refval{\idname{6}}{}$ for the person ``Sillier Murphy,'' then a backwards projection of {\lbname{actors}} finds all the movies where that person played a character, and attaches to those values the link property of the character that person played:}
 \emeta_\textit{\textit{sm}} \bdot \lbname{actors}  [\tbname{Movie}]
 & \approx \left[ \refval{\idname{8}}{ \imaprec{\lbname{character}}{[\stringlit{Fissure}]} }, \refval{\idname{9}}{ \imaprec{\lbname{character}}{[\stringlit{Doc Boom}]}} \right] \\
 \emeta_\textit{\textit{sm}} \bdot \lbname{actors}  [\tbname{Movie}] \cdot \lbname{title}
 & \approx \left[ \stringlit{Interception}, \stringlit{Open Hammer} \right] \\
 \emeta_\textit{\textit{sm}} \bdot \lbname{actors}  [\tbname{Movie}] \cdot \lbname{character}
 & \approx \left[ \stringlit{Fissure}, \stringlit{Doc Boom} \right]
\end{align*}
This behavior is captured in the rule \textsc{E-backlink} and the helper function $\mathsf{seek}(\dbn, \tbn, \lbn, \idn)$.

Backwards links remove all information about cardinality: the rule \textsc{T-backlink} results in the cardinality constraint $\card{0}{\infty}$. In EdgeQL, a schema can say a property is exclusive (or unique) in the same way that a relational database might say that users have to have an exclusive (or unique) email address. If we backwards-project such a constrained property from an expression with cardinality $\card{i}{1}$, then the resulting expression has cardinality $\card{0}{1}$. Formalizing uniqueness constraints, and the way they can make cardinality inference more specific, is left for future work.

\subsection{Shapes}
\label{sec-exp-shapes}

\begin{figure}
\begin{tabular}{lrll}
$\emeta$ & $::=$ & $\ldots \mid x \mid e \apply x.S \mid \ldots$ &
  expressions (continued)\\
$S$ & $::=$ & $\rec{\maprec{\lbnx{1}}{\emeta_1}\recsep\ldots\recsep\maprec{\lbnx{n}}{\emeta_n}}$ & shapes
\end{tabular}

\begin{mathpar}

\inferrule[(T-shaping)]
{
\tctxn \vdash_{\schema} \tmap{e}{\typmeta}{\cardmeta}
\\
\typmeta = \reft{\tbn}{\ldots}
\\\\
\tctxn, \tmap{x}{\typmeta}{\card{1}{1}} \vdash_{\schema} \tmap{\emeta_1}{\typmeta_1}{\cardmeta_1}
\\
\cdots
\\
\tctxn, \tmap{x}{\typmeta}{\card{1}{1}} \vdash_{\schema} \tmap{\emeta_n}{\typmeta_n}{\cardmeta_n}
\\\\
\typmeta' = \typmeta \uplus \rec{
 \tmap{\lbnx{1}}{\typmeta_1}{\cardmeta_1}
 \recsep
 \ldots
 \recsep
 \tmap{\lbnx{n}}{\typmeta_n}{\cardmeta_n}
 }
}
{\tctxn \vdash_{\schema} \tmap{e \apply x.\rec{\maprec{\lbnx{1}}{\emeta_1}\recsep\ldots\recsep\maprec{\lbnx{n}}{\emeta_n}}}{\typmeta'}{\cardmeta }}

\\\\

\inferrule[(E-shaping)]
{
\seval{\vctxn}{\dbn_0}{\dbn_1}{\emeta}{[\dmeta_{1}, \ldots, \dmeta_{n}]}
\\\\
\seval{\vctxn, x \mapsto [\dmeta_1]}{\dbn_1}{\dbn_2}{S}{\shapevalmeta_1}
\\
\cdots
\\
\seval{\vctxn, x \mapsto [\dmeta_n]}{\dbn_n}{\dbn_{n+1}}{S}{\shapevalmeta_n}
}
{\seval{\vctxn}{\dbn_0}{\dbn_{n+1}}{\emeta \apply x.S}{[\dmeta_{1} \uplus \shapevalmeta_1, \ldots, \dmeta_{n} \uplus \shapevalmeta_n]}}

\inferrule[(E-shape)]
{
\seval{\vctxn}{\dbn_1}{\dbn_2}{\emeta_1}{\dmeta_1}
\\
\cdots
\\
\seval{\vctxn}{\dbn_{n}}{\dbn_{n+1}}{\emeta_n}{\dmeta_n}
}
{\seval{\vctxn}{\dbn_1}{\dbn_{n+1}}
{\rec{\maprec{\lbnx{1}}{\emeta_1}\recsep\ldots\recsep\maprec{\lbnx{n}}{\emeta_n}}}
{\rec{\maprec{\lbnx{1}}{\dmeta_1}\recsep\ldots\recsep\maprec{\lbnx{n}}{\dmeta_{n}}}}
}

\end{mathpar}

\caption{Syntax and semantics of shapes: the rules \textsc{E-shaping} and \textsc{E-shape} introduce $\seval{\vctxn}{\dbn}{\dbn'}{S}{\shapevalmeta}$ as a new judgment.}
\label{figsemshape}
\end{figure}

The syntax forms we have presented thus far will only ever attach link properties like $\lbname{character}$ to object values, but as we alluded to in \cref{sec-general-values}, we want to be able to extend object values by attaching arbitrary nested information. This is achieved with the \textit{shaping expression} $\emeta \apply x.S$, where the shape $S$ has the form $\rec{
\maprec{\lbnx{1}}{\emeta_1}
\recsep
\ldots
\recsep
\maprec{\lbnx{n}}{\emeta_{n}}
}$. The mappings $\maprec{\lbnx{i}}{\emeta_i}$ explain how to derive new property values from a single object value represented by the variable $x$. As before, we start with a few illustrative cases based on our running example:
\begin{align*}
\tbname{Movie} \apply {x.\rec{\maprec{\lbname{rating}}{4}}}
& \approx
\left[
\refval{\idname{7}}{\maprec{\lbname{rating}}{[4]}}, \refval{\idname{8}}{\maprec{\lbname{rating}}{[4]}}, \refval{\idname{9}}{\maprec{\lbname{rating}}{[4]}}
\right]
\\
\tbname{Movie} \apply {x.\rec{\maprec{\lbname{year}}{x \cdot \lbname{year}}}}
& \approx
\left[
\refval{\idname{7}}{\maprec{\lbname{year}}{[2007]}}, \refval{\idname{8}}{\maprec{\lbname{rating}}{[2010]}}, \refval{\idname{9}}{\maprec{\lbname{rating}}{[2024]}}
\right]
\\
\tbname{Movie} \apply {x.\rec{}}
& \approx
[\refval{\idname{7}}{}, \refval{\idname{8}}{}, \refval{\idname{9}}{}]
\\
\tbname{Movie} \apply {x.\rec{\maprec{\lbname{rating}}{4}}} \apply {y.\rec{}}
& \approx
\left[
\refval{\idname{7}}{\imaprec{\lbname{rating}}{[4]}}, \refval{\idname{8}}{\imaprec{\lbname{rating}}{[4]}}, \refval{\idname{9}}{\imaprec{\lbname{rating}}{[4]}}
\right]
\end{align*}

The first example adds a new property, {\lbname{rating}}, to the movie references that are returned, whereas the second example loads in an existing property of movies, {\lbname{year}}, that is already present in the database store. The third example may give the impression that applying the trivial shape, $x.{\rec{}}$, doesn't do anything, but that isn't quite true: when object values that contain computed properties are re-shaped, any record entries that aren't explicitly mentioned in the final shape are preserved but marked as non-visible. In the last example, the ratings are rendered invisible by reshaping; the mappings of the form $\maprec{\lbname{rating}}{[4]}$ are replaced by mappings $\imaprec{\lbname{rating}}{[4]}$. Record entry visibility marks were mentioned in \cref{sec-general-values}, but this is the first time we are seeing them used; they do not affect the evaluation of expressions or projection, but record entries of the form $\imaprec{\lbn}{\dsetmeta}$ are intended to be omitted when results are displayed or serialized.

Nested shaping is the essential tool for querying object-shaped data in the graph-relational database model; evaluating this expression would get us the structure in \cref{figmov}:
\begin{align*}
\tbname{Movie} \apply x.\lrec \;\; &  \lbname{title} := x\cdot\lbname{title}\recsep\\
& \lbname{year}  := x\cdot\lbname{year}\recsep\\
& \lbname{directors}  := (x\cdot\lbname{directors}) \apply y.\rec{
  \; \lbname{name} := y\cdot\lbname{name}\recsep
  \; \lbname{age} := y\cdot\lbname{age}
  \;}
\recsep\\
& \lbname{actors}  := (x\cdot\lbname{actors}) \apply y.\rec{
  \; \lbname{name} := y\cdot\lbname{name}\recsep
  \; \lbname{character} := y\cdot\lbname{character}
  \; }
    \;\; \rrec
\intertext{In the context of our running example, this will evaluate to a thee values, one of which is the following:}
\lrec \;\;
& \lbname{title} := [\;\stringlit{Transistors}\;]\recsep\\
& \lbname{year}  := [\;2007\;]\recsep\\
& \lbname{directors} := [\;\refval{\idname{1}}{\;
    \maprec{\lbname{name}}{[\,\stringlit{Michael Cove}\,]}
    \recsep\;
    \maprec{\lbname{age}}{[\,53\,]}\;
    }\;]
\recsep\\
& \lbname{actors}  := [\;
  \refval{\idname{2}}{\;
    \maprec{\lbname{name}}{[\,\stringlit{Megan Wolf}\,]}
    \recsep\;
    \maprec{\lbname{character}}{[\,\stringlit{Meg Tech}\,]}\;
  }
  \recsep\\
&\qquad\qquad\;\,\,
  \refval{\idname{5}}{\;
    \maprec{\lbname{name}}{[\,\stringlit{Shy Andbuff}\,]}
    \recsep\;
    \maprec{\lbname{character}}{[\,\stringlit{Sam Man}\,]}\;
  }
\;]
    \;\; \rrec
\end{align*}
Unlike the stored value $R_\mathit{tr}$ presented at the beginning of \cref{sec-data-at-rest}, this value explicitly contains all the information from the JSON object that we saw in \cref{figmov} as visible record entries. (We will return to JSON serialization in \cref{sec-type-directed-serial}.)

This example makes it apparent that record entries of the form $\maprec{\lbn}{x\cdot \lbn}$ in shapes are extremely common in many use cases. The concrete syntax for this query shown in \cref{sec-intro-schema-type} demonstrates EdgeQL's syntactic shorthand for this pattern, analogous to shorthand property names in JavaScript.
That shorthand syntax bears a striking resemblance to a queries in some API configuration languages, particularly GraphQL \cite{Hartig18www,GraphQLwebpage}. There is substantial overlap between the goals of GraphQL and EdgeQL, and the language of GraphQL queries bears a strong resemblance to a limited subset of EdgeQL that is primarily concerned with shaping. However, GraphQL's perspective on the storage and mutation of data is fundamentally different: GraphQL mutations don't have a defined semantics, and the relationship between GraphQL types and the representation of data at rest is entirely up to the programmer, rendering may graph-relational database concepts, like backwards projection, meaningless.

Shapes are defined in \cref{figsemshape} with the help of a few additional definitions. One is $\seval{\vctxn}{\dbn}{\dbn'}{S}{\shapevalmeta}$, which manages the threading of state through the evaluation of the many parts of a shape record. The other is the right-biased \textit{record extension}, which is defined both on object types as ($\tau \uplus \rec{
 \tmap{\lbnx{1}}{\typmeta_1}{\cardmeta_1}
 \recsep
 \ldots
 \recsep
 \tmap{\lbnx{n}}{\typmeta_n}{\cardmeta_n}
 }$) and on object values as
($\dmeta \uplus \rec{
 \maprec{\lbnx{1}}{\dsetmeta_1}
 \recsep
 \ldots
 \recsep
 \maprec{\lbnx{n}}{\dsetmeta_n}
 }$). Both operators include all the mappings on the right-hand side and then add any mappings from the original object type or reference value that are not being shadowed by the new additions. Formally:
\begin{itemize}
    \item
$\reft{\tbn}{
 \tmap{\lbnx{1}}{\typmeta_1}{\cardmeta_1}
 \recsep
 \ldots
 \recsep
 \tmap{\lbnx{n}}{\typmeta_n}{\cardmeta_n}
} \uplus \rec{
 \tmap{\lbnprimex{1}}{\typmeta'_1}{\cardmeta'_1}
 \recsep
 \ldots
 \recsep
 \tmap{\lbnprimex{n}}{\typmeta'_n}{\cardmeta'_n}
}$
\\
$ ~ = \reft{\tbn}{\left\{
\tmap{\lbnx{i}}{\typmeta_i}{\cardmeta_i}
\mid
\lbnx{i} \notin \{ \lbnprimex{1}, \ldots, \lbnprimex{n} \}
\right\} \cup \left\{
 \tmap{\lbnprimex{1}}{\typmeta'_1}{\cardmeta'_1}
 \recsep
 \ldots
 \recsep
 \tmap{\lbnprimex{n}}{\typmeta'_n}{\cardmeta'_n}
\right\} }$

\item
$\refval{\idn}{
 \umaprec{\lbnx{1}}{\dsetmetax{1}}
 \recsep
 \ldots
 \recsep
 \umaprec{\lbnx{n}}{\dsetmetax{n}}
} \uplus \rec{
 \maprec{\lbnprimex{1}}{\dsetmetaprimex{1}}
 \recsep
 \ldots
 \recsep
 \maprec{\lbnprimex{n}}{\dsetmetaprimex{n}}
}$
\\
$ ~ = \refval{\idn}{\left\{
 \imaprec{\lbnx{i}}{\dsetmetax{i}}
 \mid
 \lbnx{i} \notin \{ \lbnprimex{1}, \ldots, \lbnprimex{n} \}
\right\} \cup \left\{
 \maprec{\lbnprimex{1}}{\dsetmetaprimex{1}}
 \recsep
 \ldots
 \recsep
 \maprec{\lbnprimex{n}}{\dsetmetaprimex{n}}
\right\}
}$
\end{itemize}
Record extension for values $\dmeta \uplus \shapevalmeta$ is also responsible for rendering all mappings that are carried over from the original object value invisible by changing every $\umaprec{\lbnx{i}}{\dsetmetax{i}}$ mapping to $\imaprec{\lbnx{i}}{\dsetmetax{i}}$.



\subsection{Built-in functions}\label{sec-exp-builtin}

\begin{figure}
\begin{tabular}{lrll}
$f$ & $::=$ &
$\mathtt{count} \mid \mathtt{eq} \mid \mathtt{append} \mid \mathtt{coalesce} \mid \mathtt{any} \mid \ldots$
 & built-in functions
\\
$p$ & $::=$ & $\pone \mid \popt \mid \pany$
& parameter modifiers \\
$\emeta$ & $::=$ & $
\ldots \mid
f^\dagger(\emeta_1, \ldots, \emeta_n) \mid \ldots
$ &
 expressions (continued)
\end{tabular}

\begin{mathpar}
\inferrule[(T-function)]
{
f : \left( \typmeta_1 \hash p_1, \ldots \typmeta_n \hash p_n \right) \rightarrow \typmeta' \hash \cardmeta'
\\\\
\Gamma \vdash_\schema \tmap{\emeta_1}{\typmeta_1}{\cardmeta_1}
\quad
\cdots
\quad
\Gamma \vdash_\schema \tmap{\emeta_n}{\typmeta_n}{\cardmeta_n}
\\
\cardmeta_1 \leq \llbracket p_1 \rrbracket
\quad
\cdots
\quad
\cardmeta_n \leq \llbracket p_n \rrbracket
}
{\tctxn \vdash_{\schema} \tmap{f^\dagger(\emeta_1, \ldots, \emeta_n)}{\typmeta'}{\cardmeta'}}

\inferrule[(E-function)]
{
\seval{\vctxn}{\dbn_1}{\dbn_2}{\emeta_1}{\dsetmetax{1}}
\quad
\cdots
\quad
\seval{\vctxn}{\dbn_n}{\dbn_{n+1}}{\emeta_n}{\dsetmetax{n}}
\\
\llbracket f \rrbracket (\dsetmetax{1},\ldots,\dsetmetax{n}) \approx \dsetmeta
}
{\seval{\vctxn}{\dbn_1}{\dbn_{n+1}}{f^\dagger(\emeta_1, \ldots, \emeta_n)}{\dsetmeta}}

\end{mathpar}

\caption{Syntax and semantics of built-in functions: the rule \textsc{T-function} introduces $f : \left( \typmeta_1 \hash p_1, \ldots \typmeta_n \hash p_n \right) \rightarrow \typmeta'$ as a new judgment.}
\label{figsemfunction}
\end{figure}

The ability to support nontrivial transformations on data within database queries is a significant part of what distinguishes Gel from systems like GraphQL that are primarily concerned with the shape of data, and one of the ways this is managed while keeping the language total is to have an extensive standard library of functions. \Cref{figsemfunction} introduces the structure of built-in functions and gives a few examples.

Built-in functions are typed according to a typing judgment $f : \left( \typmeta_1 \hash p_1, \ldots \typmeta_n \hash p_n \right) \rightarrow \typmeta' \hash \cardmeta'$, where the parameter modifiers $p$ correspond to three of the cardinality modes: $\llbracket \pone \rrbracket = \card{1}{1}$, $\llbracket \popt \rrbracket = \card{0}{1}$, and $\llbracket \pany \rrbracket = \card{0}{\infty}$. The theorems we prove about the language in the \Cref{sec-metatheory} place a number of constraints on built-in functions:
\begin{itemize}
    \item The parameter modifiers and return type and cardinality must be uniquely defined by the function $f$ and the sequence of types $\typmeta_1, \ldots, \typmeta_n$. This requirement allows significant ``overloading'' of functions, works perfectly well with variadic functions that take any number of arguments, and allows functions like equality to be polymorphic by letting $\mathtt{eq} : (\typmeta \hash \pone, \typmeta \hash \pone) \rightarrow \tbool \hash \card{1}{1}$ for every type $\typmeta$.
    \item If arguments $\dsetmetax{1},\ldots,\dsetmetax{n}$ match the type and cardinality indicated by the  typing judgment for built-ins, then the interpretation of that built-in function (written as $\llbracket f \rrbracket$ in \textsc{E-function}) must be defined and must return a sequence of computed values with the type and cardinality described by the typing judgment for built-in functions.
\end{itemize}
Here are a few illustrative examples of built-in functions:
\begin{itemize}
    \item $\mathtt{count}(\typmeta \hash \pany) \rightarrow \tint \hash \card{1}{1}$ returns the number of values in its argument.
    \item $\mathtt{eq}(\typmeta \hash \pone, \typmeta \hash \pone) \rightarrow \tbool \hash \card{1}{1}$ returns {\ctt} if both its arguments are the same, and otherwise returns {\cff}.
    \item $\mathtt{append}(\tstring \hash \pone, \tstring \hash \pone) \rightarrow \tstring \hash \card{1}{1}$ appends two strings.
    \item $\mathtt{coalesce}(\typmeta \hash \popt, \typmeta \hash \pany) \rightarrow \typmeta \hash [0, \infty]$ returns the first argument if it's nonempty, and otherwise returns the second argument.
    \item $\mathtt{any}(\tbool \hash \pany) \rightarrow \tbool \hash [1, 1]$ returns {\ctt} if its argument contains any {\ctt} values, and otherwise returns {\cff}.
\end{itemize}
Many built-in functions have special syntax in EdgeQL's concrete syntax: $\mathtt{eq}(x,y)$ is, for example, written \verb|x=y|, and $\mathtt{coalesce}(x,y)$ is written \verb|x??y|.

The rule \textsc{T-function} is the first time we have encountered an expression that might fail to have a type because the cardinality of a sub-expression is wrong: because $\card{1}{\infty} \not\leq \card{1}{1}$, you cannot give a type to the expression $\mathtt{append}^\dagger(\stringlit{Hello } \cup \stringlit{Bye }, \stringlit{Alice} \cup \stringlit{Bob})$. In some cases this is desirable, but in EdgeQL we would like all function calls to be appropriately lifted so that this append operation results in four strings which say hello and goodbye to Alice and Bob. This will be accomplished in the next section with the introduction of derived expressions. (We're adopting a convention of marking syntax forms that are only used via derived forms with a dagger symbol ($\dagger$).)

\subsection{Computation in expressions}\label{seceprog}

\begin{figure}
\begin{tabular}{lrll}
$\emeta$ & $::=$ & $
\ldots \mid \mathsf{if}^\dagger(\emeta; \emeta_t; \emeta_f)\mid \mathsf{with}(\emeta_1; x.\emeta_2)
\mid \mathsf{for}(\emeta_1; x.\emeta_2)
\mid \mathsf{orderby}(\emeta_1; x.\emeta_2) \mid \ldots$
\end{tabular}

\begin{mathpar}

\inferrule[(T-if)]
{
\tctxn \vdash_{\schema} \tmap{\emeta}{\tbool}{\card{1}{1}}
\\\\
\tctxn \vdash_{\schema} \tmap{\emeta_t}{\typmeta}{\card{i_1}{i_2}}
\\
\tctxn \vdash_{\schema} \tmap{\emeta_f}{\typmeta}{\card{i_3}{i_4}}
\\\\
\cardmeta' = [\mathrm{min}(i_1, i_3), \mathrm{max}(i_2,i_4)]
}
{
{\tctxn \vdash_{\schema} \tmap{\mathsf{if}^\dagger(\emeta; \emeta_t; x.\emeta_f)}{\typmeta}{ \cardmeta'}}
}

\inferrule[(T-with)]
{
\tctxn \vdash_{\schema} \tmap{\emeta_1}{\typmeta_1}{\cardmeta_1}
\\\\
\tctxn, \tmap{x}{\typmeta_1}{\cardmeta_1} \vdash_{\schema} \tmap{\emeta_2}{\typmeta_2}{\cardmeta_2}
}
{\tctxn \vdash_{\schema} \tmap{\mathsf{with}(\emeta_1; x.\emeta_2)}{\typmeta_2}{\cardmeta_2}}

\inferrule[(T-for)]
{
\tctxn \vdash_{\schema} \tmap{\emeta_1}{\typmeta_1}{\cardmeta_1}
\\\\
\tctxn, \tmap{x}{\typmeta_1}{[1,1]} \vdash_{\schema} \tmap{\emeta_2}{\typmeta_2}{\cardmeta_2}
}
{\tctxn \vdash_{\schema} \tmap{\mathsf{for}(\emeta_1; x.\emeta_2)}{\typmeta_2}{\cardmeta_1 \times \cardmeta_2}}

\inferrule[(T-orderby)]
{
\tctxn \vdash_{\schema} \tmap{\emeta_1}{\typmeta_1}{\cardmeta}
\\\\
\tctxn, \tmap{x}{\typmeta_1}{[1,1]} \vdash_{\schema} \tmap{\emeta_2}{\typmeta_2}{\card{0}{1}}
}
{\tctxn \vdash_{\schema} \tmap{\mathsf{orderby}(\emeta_1; x.\emeta_2)}{\typmeta_2}{\cardmeta}}

\inferrule[(E-if)]
{\seval{\vctxn}{\dbn_a}{\dbn_b}{\emeta}{[\dmeta_\textit{cond}]}
\\\\
{\begin{cases}
\seval{\vctxn}{\dbn_b}{\dbn_c}{\emeta_t}{\dsetmeta} & \text{if } \dmeta_\textit{cond} = \ctt\\
\seval{\vctxn}{\dbn_b}{\dbn_c}{\emeta_f}{\dsetmeta} & \text{if } \dmeta_\textit{cond} = \cff
\end{cases}}
}
{\seval{\vctxn}{\dbn_a}{\dbn_{c}}{\mathsf{if}^\dagger(\emeta; \emeta_t; \emeta_f)}{\dsetmeta}}

\inferrule[(E-with)]
{\seval{\vctxn}{\dbn_a}{\dbn_b}{\emeta_1}{\dsetmeta}
\\\\
\seval{\vctxn, x \mapsto \dsetmeta}{\dbn_b}{\dbn_c}{\emeta_2}{\dsetmetaprime}
}
{\seval{\vctxn}{\dbn_a}{\dbn_{c}}{\mathsf{with}(\emeta_1; x.\emeta_2)}{\dsetmetaprime}}

\inferrule[(E-for)]
{\seval{\vctxn}{\dbn_0}{\dbn_1}{\emeta}{[\dmeta_1, \ldots, \dmeta_n]}
\\\\
\seval{\vctxn, x \mapsto \dmeta_1}{\dbn_1}{\dbn_2}{\emeta'}{\dsetmetaprimex{1}}
\quad
\cdots
\quad
\seval{\vctxn, x \mapsto \dmeta_n}{\dbn_n}{\dbn_{n+1}}{\emeta'}{\dsetmetaprimex{n}}
\\
\left[\dsetmetaprimex{1}, \ldots ,\dsetmetaprimex{n}\right] \approx \dsetmetaprime
}
{\seval{\vctxn}{\dbn_0}{\dbn_{n+1}}{\mathsf{for}(\emeta; x.\emeta')}{\dsetmetaprime}}

\inferrule[(E-orderby)]
{\seval{\vctxn}{\dbn_0}{\dbn_1}{\emeta}{[\dmeta_1, \ldots, \dmeta_n]}
\\\\
\seval{\vctxn, x \mapsto \dmeta_1}{\dbn_1}{\dbn_2}{\emeta'}{\dsetmetaprime{1}}
\quad
\cdots
\quad
\seval{\vctxn, x \mapsto \dmeta_n}{\dbn_n}{\dbn_{n+1}}{\emeta'}{\dsetmetaprimex{n}}
\\
{\textsf{sort}}\left(
\left( \dmeta_1, \dsetmetaprimex{1} \right),
\ldots,
\left( \dmeta_n, \dsetmetaprimex{n}  \right)
\right) = \dsetmetax{\textit{ordered}}
}
{\seval{\vctxn}{\dbn_0}{\dbn_{n+1}}{\mathsf{orderby}(\emeta; x.\emeta')}{\dsetmetax{\textit{ordered}}}}

\end{mathpar}

\caption{Syntax and semantics of computations: the \textsf{sort(\ldots)} operation in \textsf{E-orderby} returns a sequence containing the first element of the tuples passed in, sorted by the second element of the tuples passed in. The second element must be a sequence of length zero or one; empty sequences are treated as smaller than nonempty sequences.}
\label{figsemprog}
\end{figure}

Alongside built-in functions, \cref{figsemprog} describes some expressions that allow programmers to more directly manipulate information in the graph-relational model.
With-binding\footnote{SQL programmers say ``with'' where functional programmers say ``let.''} and for-loops have the same syntactic structure, but the former evaluates the sub-expression once with the variable bound to the full result of evaluating the first argument, and the latter evaluates the sub-expression once for each value resulting from evaluating the first argument.

Conditional expressions provide a simple example of derived forms for syntax. The requirement that the scrutinee of a conditional expression have the cardinality mode $\card{1}{1}$ is more restrictive than the EdgeQL implementation, but we can derive the  if-then-else expression we actually want as
$\mathsf{if}(\emeta; \emeta_t; \emeta_f) \doteq \mathsf{for}(\emeta; x.\,\mathsf{if}^\dagger(x; \emeta_f; \emeta_t))$. We can similarly define the critical filter function $\mathsf{filter}(\emeta_1; x.\emeta_2)$, equivalent to SQL's \verb|WHERE|, as
$
\mathsf{for}(
e_1;
x.\;
\mathsf{if}(\mathtt{any}^\dagger(e_2);
x;
\emp
)
)
$.
Another useful derived form is $\mathsf{optional\_for}(e_1; x.e_2)$ which behaves like a for-loop unless the scrutinee evaluates to zero values, in which case it still runs the  sub-expression but with $x$ mapping to the empty sequence; it can be defined with existing machinery like this:
\[\mathsf{optional\_for} \doteq
\mathsf{with}(e_1; y.\;\mathsf{if}(\mathtt{eq}(0, \mathtt{count}^\dagger(y));\; \mathsf{with}(\emp; x.e_2);\; \mathsf{for}(y; x.e_2))\]
Derived forms in this presentation need not be derived forms in the actual implementation, and in both of the last two cases they are definitely not. In full EdgeQL, cardinality inference for filtering expressions interacts with the uniqueness constraints previously mentioned in \cref{secrefvalues}, and actually implementing $\mathsf{optional\_for}$ as described here could greatly increase code size.

Finally, using $\mathsf{with}$, $\mathsf{for}$, and $\mathsf{optional\_for}$, we can present a lifted version of function call, $f(e_1, \ldots, e_2)$, which is not restricted by cardinality, as follows:
\[
f(e_1, \ldots, e_2) \doteq \mathsf{bind}_1(e_1; x_1. \ldots \mathsf{bind}_n(e_n; x_n. \; f^\dagger(x_1, \ldots, x_n))\ldots)
\]
where $\mathsf{bind}_i$ is ``$\mathsf{with}$'' if the $i^\textit{th}$ argument has parameter mode $\pany$, ``$\mathsf{for}$'' if it has the parameter mode $\pone$, and the derived ``$\mathsf{optional\_for}$'' if it has the parameter mode $\popt$. The effect of this derived form is to ensure that function arguments can, in practice, have any cardinality. If required by the parameter mode, arguments will be split apart and their cross-product will be broadcast across multiple calls to the underlying function.

\subsection{Manipulating the database store}\label{secemutate}

\begin{figure}
\begin{tabular}{lrll}
$\emeta$ & $::=$ & $
\ldots \mid
\textsf{insert}\,\tbn\,S \mid \textsf{update}^\dagger\,\emeta\,\textsf{set}\,x.S
$ &
 expressions (continued)\\
\end{tabular}

\begin{mathpar}
\inferrule[(T-insert)]
{
\schema(\tbn) = \rec{\tmap{\lbnx{1}}{\tmeta_1}{\cardmeta_1},\ldots,\tmap{\lbnx{n}}{\tmeta_n}{\cardmeta_n}}
\\
\tctxn \vdash_{\schema} \emeta_1 \div \typmeta_1 / \tmeta_1 \hash \cardmeta_1
\quad
\cdots
\quad
\tctxn \vdash_{\schema} \emeta_n \div \typmeta_n / \tmeta_n \hash \cardmeta_n
}
{\tctxn \vdash_{\schema} \tmap{\mathsf{insert}\,\tbn\,
\rec{\maprec{\lbnx{1}}{\emeta_1}\recsep\ldots\recsep\maprec{\lbnx{n}}{\emeta_n}}
}{\reft{\tbn}{
\tmap{\lbnx{1}}{\typmeta_1}{\cardmeta_1},\ldots,\tmap{\lbnx{n}}{\typmeta_n}{\cardmeta_n}
}}{\card{1}{1}}}

\inferrule[(T-update)]
{
\tctxn \vdash_{\schema} \tmap{\emeta}{\typmeta}{\card{1}{1}}
\\
\typmeta = \reft{\tbn}{\ldots}
\\
\schema(\tbn) = \rec{\tmap{\lbnx{1}}{\tmeta_1}{\cardmeta_1},\ldots,\tmap{\lbnx{n}}{\tmeta_n}{\cardmeta_n}, \ldots}
\\
\tctxn, \tmap{x}{\typmeta}{\card{1}{1}} \vdash_{\schema} \emeta_1 \div \typmeta_1/\tmeta_1 \hash \cardmeta_1
\quad
\cdots
\quad
\tctxn, \tmap{x}{\typmeta}{\card{1}{1}} \vdash_{\schema} \emeta_n \div \typmeta_n/\tmeta_n \hash \cardmeta_n
}
{\tctxn \vdash_{\schema} \tmap{\mathsf{update}^\dagger\,\emeta\,\mathsf{set}\,
x.\rec{\maprec{\lbnx{1}}{\emeta_1}\recsep\ldots\recsep\maprec{\lbnx{n}}{\emeta_n}}
}{\reft{\tbn}{
\tmap{\lbnx{1}}{\typmeta_1}{\cardmeta_1},\ldots,\tmap{\lbnx{n}}{\typmeta_n}{\cardmeta_n}
}}{\card{0}{1}}}

\inferrule[(TS-simple)]
{
\tctxn \vdash_{\schema} \tmap{\emeta}{\tsmeta'}{\cardmeta'}
\\
\cardmeta' \leq \cardmeta
}
{\tctxn \vdash_{\schema} \emeta \div \tsmeta/\tsmeta \hash \cardmeta}

\inferrule[(TS-reference)]
{
\tctxn \vdash_{\schema} \tmap{\emeta}{\typmeta}{\cardmeta'}
\\
\typmeta = \reft{\tbn}{\tmap{\lbnx{1}}{\tsmeta_1}{\cardmeta_1'}, \ldots, \tmap{\lbnx{n}}{\tsmeta_n}{\cardmeta_n'}, \ldots}
\\\\
\cardmeta' \leq \cardmeta \\
\cardmeta'_1 \leq \cardmeta_1 \quad \cdots \quad \cardmeta'_n \subseteq \cardmeta_n
}
{\tctxn \vdash_{\schema} \emeta \div \typmeta/\reft{\tbn}{\tmap{\lbnx{1}}{\tsmeta_1}{\cardmeta_1}, \ldots, \tmap{\lbnx{n}}{\tsmeta_n}{\cardmeta_n}} \hash \cardmeta}

\inferrule[(E-insert)]
{
\schema(\tbn) = \rec{\tmap{\lbnx{1}}{\tmeta_1}{\cardmeta_1},\ldots,\tmap{\lbnx{n}}{\tmeta_n}{\cardmeta_n}}
\\
\deval{\vctxn}{\dbn_1}{\dbn_2}{\emeta_1}{\dsetmetax{1}}
\quad
\cdots
\quad
\deval{\vctxn}{\dbn_n}{\dbn_{n+1}}{\emeta_n}{\dsetmetax{n}}
\\
\dsetmetax{1} : \tmeta_1 \rhd \vsetmetax{1}
\quad
\cdots
\quad
\dsetmetax{n} : \tmeta_n \rhd \vsetmetax{n}
\\
\idn\#\dbn_\textit{init}
\\
\idn\#\dbn_{n+1}
\\
\dbn' = \dbn_{n+1} \cup \left\{\left(
  \idn, \tbn, \lkclosed\,,
   \rec{
     \maprec{\lbnx{1}}{\vsetmetax{1}},\ldots,
     \maprec{\lbnx{n}}{\vsetmetax{n}}
   }
\right)\right\}
}
{\deval{\vctxn}{\dbn_1}{\dbn'}{
\mathsf{insert}\,\tbn\,
\rec{\maprec{\lbnx{1}}{\emeta_1}\recsep\ldots\recsep\maprec{\lbnx{n}}{\emeta_n}}
}
{[
\refval{\idn}{
\imaprec{\lbnx{1}}{\dsetmetax{1}},\ldots,
\imaprec{\lbnx{n}}{\dsetmetax{n}}}
]}
}

\inferrule[(E-update)]
{
\seval{\vctxn}{\dbn_0}{\dbn_1}{\emeta}{[\dmeta]}
\\
\dmeta = \refval{\idn}{\ldots}
\\\\
\seval{\vctxn, x \mapsto [\dmeta]}{\dbn_1}{\dbn_2}{\emeta_1}{\dsetmetax{1}}
\quad
\cdots
\quad
\seval{\vctxn, x \mapsto [\dmeta]}{\dbn_n}{\dbn_{n+1}}{\emeta_n}{\dsetmetax{n}}
\\
\dbn_{n+1} - \left\{\left(\idn, \tbn, \lkn, \rec{
\maprec{\lbnprimex{1}}{\vsetmetaprimex{1}}\recsep\ldots\recsep\maprec{\lbnprimex{n}}{\vsetmetaprimex{n}}
}\right)\right\} = \dbn'
\\
\dbn'\#\idn
\\\\
\schema(\tbn) = \rec{\tmap{\lbnx{1}}{\tmeta_1}{\cardmeta_1},\ldots,\tmap{\lbnx{n}}{\tmeta_n}{\cardmeta_n},\ldots}
\\
\dsetmetax{1} : \tmeta_1 \rhd \vsetmetax{1}
\quad
\cdots
\quad
\dsetmetax{n} : \tmeta_n \rhd \vsetmetax{n}
\\
\dbn'' = \dbn' \cup \left\{\left(
\idn, \tbn, \lkclosed\,, \rec{
\left\{
\maprec{\lbnprimex{i}}{\vsetmetaprimex{i}}
\mid \lbnprimex{i} \notin \{ \lbnx{1}, \ldots, \lbnx{n} \}
\right\}
\cup
\left\{
\maprec{\lbnx{1}}{\vsetmetax{1}},\ldots,\maprec{\lbnx{n}}{\vsetmetax{n}}
\right\}
}
\right)\right\}
\\\\
{\begin{cases}
\dbn_\textit{result} = \dbn_{n+1}, \dsetmetaprime = [] & \text{if }  \idn\#\dbn_{n+1} \text{ or if } \lkn = \lkclosed\\
\dbn_\textit{result} = \dbn'', \dsetmetaprime = [\refval{\idn}{\imaprec{\lbnx{1}}{\dsetmetax{1}},\ldots,\imaprec{\lbnx{n}}{\dsetmetax{n}}}] & \textit{if } \dbn_{n+1} \neq \dbn' \text{ and } \lkn = \lkopen
\end{cases}}
}
{\seval{\vctxn}{\dbn_0}{\dbn_\textit{result}}
{
\mathsf{update}^\dagger\,\emeta\,\mathsf{set}\,x.\rec{\maprec{\lbnx{1}}{\emeta_1}\recsep\ldots\recsep\maprec{\lbnx{n}}{\emeta_n}}
}
{
\dsetmetaprime
}
}

\inferrule[(ES-sequence)]
{ \dmeta_1 : \tmeta \rhd \vmeta_1
\quad
\cdots
\quad
  \dmeta_n : \tmeta \rhd \vmeta_n
}
{[\dmeta_1, \ldots, \dmeta_n] : \tmeta \rhd [\vmeta_1, \ldots, \vmeta_n]}

\inferrule[(ES-simple)]
{ }
{\vsmeta : \tsmeta \rhd \tsmeta}

\inferrule[(ES-reference)]
{ }
{\refval{\idn}{\umaprec{\lbnx{1}}{\vsmeta_1}, \ldots, \umaprec{\lbnx{n}}{\vsmeta_n}, \ldots} : \reft{\tbn}{\tmap{\lbnx{1}}{\tsmeta_1}{\cardmeta_1'}, \ldots, \tmap{\lbnx{n}}{\tsmeta_n}{\cardmeta_n'}} \rhd \refval{\idn}{\imaprec{\lbnx{1}}{\vsmeta_1}, \ldots, \imaprec{\lbnx{n}}{\vsmeta_n}}}
\end{mathpar}

\caption{Syntax and semantics of insertion and update: the static rules introduce a judgment $\tctxn \vdash_{\schema} \emeta \div \typmeta/\tmeta \hash \cardmeta$ for checking expressions against the stored value type and schema present in a schema, and the dynamic rules introduce $\dsetmeta : \tmeta \rhd \vsetmeta$ and $\dmeta : \tmeta \rhd \vmeta$ for stripping everything but link properties off of a value for storage. The apartness judgment $\idn\#\dbn$ holds if no tuple in $\dbn$ has $\idn$ as its first element.}
\label{figseminsertupdate}
\end{figure}

The last stop of our tour of the core calculus for graph-relational programming are two expressions that actually manipulate the database store, insertion and update.
While the language specification could in some ways be dramatically specified if insertion and update were interpreted as top-level commands and not as potentially-deeply-nested expressions, it's incredibly useful to be able to compositionally add subterms ``where they belong'' in an object-shaped structure, as demonstrated by this example adding a movie along with its director:
\begin{align*}
\mathsf{insert}\;\tbname{Movie}\; &
    \lrec &&
      \maprec{\lbname{directors}}{
      \mathsf{insert}\;\tbname{Person}
      \rec{\;
        \maprec{\lbname{name}}{\stringlit{Paul Shiver}},
        \maprec{\lbname{age}}{37},
        \maprec{\lbname{born}}{\stringlit{Earth}}
      \;}
      }
      \recsep \\
      &&&
            \maprec{\lbname{title}}{\stringlit{Frozen Planet}}\recsep\;\;
      \maprec{\lbname{year}}{2011}\recsep\;\;
      \maprec{\lbname{actors}}{\emp}
    \qquad \rrec
\end{align*}

The definitions in \cref{figseminsertupdate} have a couple of features worth explicitly calling out. It is a bit of a subtle point that \textsc{T-update} is only required to capture a subset of the labels defined in the schema, and that the object types in \textsc{TA-reference} can mention a superset of the labels in the schema. The rule \textsc{E-update} includes a right-biased union like the ones we saw in \cref{sec-exp-shapes}.

We lift the update operation to work on multiple expressions the same way we lifted conditionals:
\[\textsf{update}\,\emeta\,\textsf{set}\,x.S
 \doteq \mathsf{for}(\emeta; y.\,\textsf{update}^\dagger\,y\,\textsf{set}\,x.S)\]
The shape parts of the update and insertion expressions do not get lifted; these join $\mathsf{orderby}$ in requiring subterms to have specific cardinality constraints rather than using a lifted or broadcasting interpretation.

\section{Metatheory}
\label{sec-metatheory}

Our static semantics can be read in the conventional way as a functional algorithm or logic program, which leads immediately to the following two theorems about the type system:

\begin{theorem}[Decidability of type synthesis]
    Given schema $\schema$ where $\vdash \schema \,\mathtt{ok}$, context $\tctxn$, and expression $\emeta$, it is decidable whether there exists some type $\typmeta$ and some cardinality mode $\cardmeta$ such that $\tctxn \vdash_\schema \tmap{\emeta}{\typmeta}{\cardmeta}$.
\end{theorem}

\begin{theorem}[Uniqueness of type synthesis]
If $\vdash \schema \,\mathtt{ok}$, $\tctxn \vdash_\schema \tmap{\emeta}{\typmeta}{\cardmeta}$, and $\tctxn \vdash_\schema \tmap{\emeta}{\typmeta'}{\cardmeta'}$, then $\typmeta = \typmeta'$ and $\cardmeta = \cardmeta'$.
\end{theorem}

\begin{proof}
    By induction on the expression $\emeta$ in both cases.
    Assuming the properties of built-in function typing required in \cref{sec-exp-builtin}, the typing judgments can be read algorithmically as a  deterministic syntax-directed logic program with $\schema$, $\tctxn$, and $\emeta$ as inputs and $\typmeta$ and $\cardmeta$ as outputs. (The auxiliary judgment $\tctxn \vdash_{\schema} \emeta \div \typmeta/\tmeta \hash \cardmeta$ in \cref{figseminsertupdate} treats all components except $\typmeta$ as input.)
\end{proof}

\begin{figure}
\begin{flushleft}\fbox{$\strut\vdash_{\schema, \dbn_\textit{init}, \dbn}\strut \dsetmeta : \typmeta \hash m$}\end{flushleft}
\[
\inferrule[(VT-set)]
{i_\mathit{lo} \leq n \leq i_\mathit{hi} \\
  \vdash_{\schema, \dbn} \dmeta_1 : \typmeta \quad
  \cdots \quad
  \vdash_{\schema, \dbn} \dmeta_n : \typmeta}
{\vdash_{\schema, \dbn_\textit{init}, \dbn} \tmap{[\dmeta_1, \ldots, \dmeta_n]}{\typmeta}{ [i_\mathit{lo}, i_\mathit{hi}]}}
\]

\begin{flushleft}\fbox{$\strut\vdash_{\schema, \dbn_\textit{init}, \dbn}\strut \dmeta : \typmeta$}\end{flushleft}
\begin{mathpar}
\inferrule[(VT-prim)]
{\vsmeta : \tsmeta}
{\vdash_{\schema, \dbn_\textit{init}, \dbn} \vsmeta : \tsmeta}

\inferrule[(VT-stored-ref)]
{\left(
\idn, \tbn, \lkn, \rec{\ldots}
\right) \in \dbn_\textit{init}
\\
\vdash_{\schema, \dbn_\textit{init}, \dbn} \dsetmetax{1} : \typmeta_1 \hash \cardmeta_1
 \quad \cdots
 \quad \vdash_{\schema, \dbn_\textit{init}, \dbn} \dsetmetax{n} : \typmeta_n \hash \cardmeta_n}
{\vdash_{\schema, \dbn_\textit{init}, \dbn}
  \refval{\idn}{\umaprec{\lbnx{1}}{\dsetmetax{1}}\recsep\ldots\recsep\umaprec{\lbnx{n}}{\dsetmetax{n}}}
  : \reft{\tbn}{\tmap{\lbnx{1}}{\typmeta_1}{\cardmeta_1}\recsep\ldots\recsep\tmap{\lbnx{n}}{\tsmeta_n}{\cardmeta_n}}}

\inferrule[(VT-new-ref)]{
\schema(\tbn) =
\rec{\tmap{\lbnprimex{1}}{\tmeta'_1}{\cardmeta'_1}\recsep\ldots\recsep\tmap{\lbnprimex{n}}{\tmeta'_n}{\cardmeta'_n}}
\\
\{ \lbnprimex{1}, \ldots, \lbnprimex{n} \} \subseteq \{ \lbnx{1}, \ldots, \lbnx{n} \}
\\\\
\left(
\idn, \tbn, \lkclosed\,, \rec{\ldots}
\right) \in \dbn
\\
\vdash_{\schema, \dbn_\textit{init}, \dbn} \dsetmetax{1} : \typmeta_1 \hash \cardmeta_1
 \quad \cdots
 \quad \vdash_{\schema, \dbn_\textit{init}, \dbn} \dsetmetax{n} : \typmeta_n \hash \cardmeta_n}
{\vdash_{\schema, \dbn_\textit{init}, \dbn}
  \refval{\idn}{\umaprec{\lbnx{1}}{\dsetmetax{1}}\recsep\ldots\recsep\umaprec{\lbnx{n}}{\dsetmetax{n}}}
  : \reft{\tbn}{\tmap{\lbnx{1}}{\typmeta_1}{\cardmeta_1}\recsep\ldots\recsep\tmap{\lbnx{n}}{\tsmeta_n}{\cardmeta_n}}}
\end{mathpar}

\caption{Typing rule for computed values}
\label{figcomputedvaluetypes}
\end{figure}

To talk about the type soundness of our model of graph-relational programming, it is first necessary to create an adaptation of typing rules for stored values from \cref{figvaluetypes} that targets computed values, shown in \cref{figcomputedvaluetypes}. These rules type values relative to two stores: an initial store and an extended store. Tuples that initially appeared in the database store can be typed with \textsc{VT-stored-ref}, which is the common case. Tuples that were introduced during the course of a query, on the other hand, will not appear in the initial store, only in the extended store. Projection in the dynamic semantics reference either the internally-stored values or the \textit{initial} store. The rule \textsc{VT-new-ref} encodes the way we avoid a situation where projecting from a newly-inserted object is undefined: any references to newly-introduced tuples carry mappings for every label in the schema, which ensures that projection will fall back to the initial database store. The rules \textsc{VT-stored-ref} and \textsc{VT-new-ref} are effectively the same rule with two different side conditions, and when both side conditions apply, the two rules will both apply but will give the same answer.

Our desired type soundness property is this: if we have a well-formed schema and database store ($\vdash \schema \,\mathtt{ok}$ and $\vdash_\schema \dbn \,\mathtt{ok}$), then we can take a closed, well-typed query
($\emptyset \vdash_\schema \tmap{\emeta}{\typmeta}{\cardmeta}$) and use the dynamic semantics to generate both a computed sequence $\dsetmeta$ and a new database store $\dbn'$ ($\eval{\emptyset}{\schema}{\dbn}{\dbn}{\dbn'}{\emeta }{\dsetmeta}$). Furthermore, whenever we run the dynamic semantics, the resulting $\dsetmeta$ has the expected type and cardinality ($\emptyset \vdash_{\schema, \dbn, \dbn'}\strut \tmap{\dsetmeta}{\typmeta}{\cardmeta}$) and the database store remains well-formed ($\vdash_\schema \dbn' \,\mathtt{ok}$).
Because evaluation is non-deterministic, this should be expressed with two separate theorems, one with evaluation as an output (the ``progress'' or ``totality'' theorem) and another with evaluation as an input (the ``preservation'' theorem).

As is usually the case, the statements of progress and preservation need to be strengthened with some additional machinery in order for to prove the result by induction over the appropriate derivation.
\begin{itemize}
\item We say that $\dbn'$ \textit{extends} $\dbn$ if $\dbn'$ has all the identifiers (with the same types) as $\dbn$, and if any tuple in $\dbn'$ with $\lkn = \lkopen$ also appears, unchanged, in $\dbn$.
\item \label{pos-envtypingdef} The judgment
$\strut\vdash_{\schema, \dbn_\textit{init}, \dbn} \vctxn : \tctxn$
holds if $\strut\vdash_{\schema} \dbn \,\mathtt{ok}$, if the variables mapped by $\vctxn$ and $\tctxn$ are identical, and if, whenever $\tmap{x}{\typmeta}{\cardmeta} \in \tctxn$ and $x \mapsto \dsetmeta \in \vctxn$, we have $\vdash_{\schema, \dbn} \tmap{\dsetmeta}{\typmeta}{\cardmeta}$.
\end{itemize}

These theorem statements also introduce a judgment for typing variable environments:
\label{pos-envtypingdef}
$\strut\vdash_{\schema, \dbn_\textit{init}, \dbn} \vctxn : \tctxn$
holds if the variables mapped by $\vctxn$ and $\tctxn$ are identical, and if, whenever $\tmap{x}{\typmeta}{\cardmeta} \in \tctxn$ and $x \mapsto \dsetmeta \in \vctxn$, we have $\vdash_{\schema, \dbn_\textit{init}, \dbn} \tmap{\dsetmeta}{\typmeta}{\cardmeta}$. 

\begin{theorem}[Preservation]
\label{theorem-preservation}
If $\dbn_1$ extends $\dbn_\textit{init}$ and we have $\tctxn \vdash_\schema \tmap{\emeta}{\typmeta}{\cardmeta}$ and $\mathstrut\vdash_\schema \dbn_\textit{init}\,\mathsf{ok}$ and $\mathstrut\vdash_\schema \dbn_1\,\mathsf{ok}$ and $\strut\vdash_{\schema, \dbn_\textit{init}, \dbn_1} \vctxn : \tctxn$ and
$\deval{\vctxn}{\dbn_1}{\dbn_2}{\emeta }{\dsetmeta}$,
then $\vdash_{\schema, \dbn_\textit{init}, \dbn_2}\strut \tmap{\dsetmeta}{\typmeta}{\cardmeta}$ and $\vdash_{\schema} \dbn_2\,\mathsf{ok}$
\end{theorem}

\begin{proof}
    By induction on the evaluation derivation.
\end{proof}

\begin{theorem}[Totality]
\label{theorem-progress}
If $\dbn_1$ extends $\dbn_\textit{init}$ and we have $\tctxn \vdash_\schema \tmap{\emeta}{\typmeta}{\cardmeta}$ and $\vdash_\schema \dbn_\textit{init} \,\mathtt{ok}$ and $\vdash_\schema \dbn_1 \,\mathtt{ok}$ and $\strut\vdash_{\schema, \dbn_\textit{init}, \dbn_1} \vctxn : \tctxn$,
then we can construct $\dsetmeta$ and $\dbn_2$ such that
$\deval{\vctxn}{\dbn_1}{\dbn_2}{\emeta }{\dsetmeta}$.
\end{theorem}

\begin{proof}
    By induction on the typing derivation and appeals to preservation for values.
\end{proof}

%
%

\noindent
More details of these proofs are available in supplementary materials (\Cref{sec-proofs}).
\section{Realizing Graph-Relational Programming in Gel}
\label{sec-gel}

The Gel system implements EdgeQL, a realization of the graph-relational programming language described in \cref{sec-data-at-rest,sec-expressions} \cite{GelDocs}. Gel is an open-source system that has been actively evolving since 2021. The core Gel implementation is around 150k lines of code, which includes the query compiler, the schema implementation and migration engine, and the database server. This count does \textit{not} include the PostgreSQL implementation, which is about an order of magnitude larger. Most of the code is mypy-checked Python, with the lexer, parser, and server written in a combination of Rust and Cython.

\subsection{Concrete Syntax}
\label{sec-concrete-nicities}

The SQL-adjacent concrete syntax of EdgeQL corresponds closely to the abstract syntax we have presented, with the exception of curly brackets replacing the angle-brackets used in our abstract syntax; we briefly discussed an example of some of EdgeQL's shorthand property notation in the discussion of shapes in \cref{sec-exp-shapes}.
In this section, we briefly cover some of the more interesting and pervasive differences between EdgeQL's concrete syntax and the abstract syntax for graph-relational programming that we have presented.

One important change is that some of the binding structures that our abstract syntax makes explicit are implicit in the concrete syntax of EdgeQL. We saw in \cref{sec-exp-shapes} that the shape expression $\tbname{Movie} \apply x.\rec{\maprec{\lbname{title}}{x\cdot\lbname{title}}}$ can be expressed with the syntactic sugar of \verb|select Movie { title }|, but in EdgeQL's concrete syntax even the fully expanded form does not include a binding site. Instead, EdgeQL automatically fills in the variable by \textit{shadowing the type name}, so the corresponding EdgeQL is actually \verb|select Movie { title := Movie.title }|, where the ``inner'' reference to \verb|Movie| is a variable referring to a single movie object. In addition, the most recently bound variable is implicitly accessed if a label projection is given with no argument, so \verb|select Movie { title := .title }| means the same thing as the previous two EdgeQL statements.

In EdgeQL's documentation and implementation, the term ``set'' describes what we have called \textit{sequences} consistently in this paper, in part to emphasize that the order of results is, in almost all cases, not at all guaranteed. This is reflected in the shorthand notation for unions: the expression \verb|{2,3,4}| corresponds to $2 \cup 3 \cup 4$ in our abstract syntax. These curly brackets do not form new or nested sets: \verb|{{{2}, {3, {4}}}}| is an equivalent expression.

Other conveniences of the concrete syntax could be defined via elaboration (sometimes type-directed elaboration), such as ``splat'' shapes (i.e. \verb|select Movie {*}|) that make all  properties visible, schema-defined computed values (these would allow our schema for movies to have an $\lbname{anniversary}$ property that is defined as the $\lbname{year}$ property plus 10), and schema-defined non-recursive user-defined functions (which can be defined via inlining). 

\subsection{Client Libraries and Type-directed Serialization}\label{sec-type-directed-serial}

\begin{figure}\footnotesize
\begin{center}
    \begin{tabular}{l|l||l|l}
        Computed value & Type and cardinality  & Derived JSON & Derived TypeScript type\\\hline
          $[]$ & 
          $\tstring \hash \card{0}{1}$ &  
          \verb|null| & 
          \verb`string | null`
        \\ 
          $[]$ & 
          $\tstring \hash \card{0}{\infty}$ & 
          \verb|[]| &
          \verb`string[]`
        \\ 
          $[\stringlit{Hi}]$ & 
          $\tstring \hash \card{0}{1}$ & 
          \verb|"Hi"| &
          \verb`string | null`
        \\ 
          $[\stringlit{Hi}]$ &
          $\tstring \hash \card{1}{1}$ &
          \verb|"Hi"|  &
          \verb|string| 
        \\
          $[\stringlit{Hi}]$ & 
          $\tstring \hash \card{0}{\infty}$ & 
          \verb|["Hi"]|  & 
          \verb|string[]| 
        \\
          $[\stringlit{Hi}, \stringlit{you}]$ & 
          $\tstring \hash \card{1}{\infty}$ & 
          \verb|["Hi","you"]|  &
          \verb|[string, ...string[]]| 
        \\
          $[\refval{\idname{7}}{}, \refval{\idname{8}}{}]$&
          $\reft{\tbn}{} \hash \card{0}{\infty}$ & 
          \verb|[{"id":"7"},|&
          \verb|{ id: string }[]| 
        \\ & & \verb| {"id":"8"}]| &
        \\
          $[\refval{\idname{7}}{\maprec{\lbname{foo}}{[4]}}]$&
          $\reft{\tbn}{\tmap{\lbname{foo}}{\mathit{int}}{\card{1}{1}}} \hash \card{1}{0}$ & 
          \verb|{"foo":4}| &
          \verb`{ foo: int } | null`
        \\
          $[\refval{\idname{7}}{\imaprec{\lbname{a}}{[4]}, \maprec{\lbname{b}}{[]}}]$&
          $\reft{\tbn}{\tmap{\lbname{a}}{\mathit{int}}{\card{1}{1}},\tmap{\lbname{b}}{\mathit{int}}{\card{0}{1}}} \hash \card{1}{1}$ &
          \verb|{"b":null}|&
          \verb`{ b: int | null }` 
    \end{tabular}
\end{center}
    \caption{Several examples of how graph-relational data, types, cardinality constraints (the leftmost two columns) are translated by Gel into TypeScript values and types (the rightmost two columns).
    }
    \label{figtypescript}
\end{figure}

Client libraries are an aspect of the Gel system that is outside the scope of our current formalization, but they are critical for the usability of Gel. Both types and cardinality modes contribute to providing simple, well-typed interfaces to EdgeQL queries.
Gel officially supports client libraries for TypeScript, Go, Python, and Rust; we will discuss the TypeScript library.

One way to use Gel in a TypeScript project is to compile queries into generated code and types. The query \verb|select "Hello" union <str>$name| in a file \verb|myQuery.edgeql| causes Gel to generate a TypeScript function with this signature:
\begin{quote}\small
\begin{verbatim}
function myQuery(
  client: Executor,
  args: { readonly name: string },
): Promise<[string, ...string[]]>
\end{verbatim}
\end{quote}
Rather than representing everything as a sequence, Gel generates the most appropriate type for a query based on information about cardinality. The TypeScript type \verb|[string, ...string[]]| describes nonempty arrays of strings, and so neatly captures the type and cardinality $\mathit{str} \hash \card{1}{\infty}$. \Cref{figtypescript} gives a set of examples about how both type and cardinality influence JSON-serializable JavaScript values returned from calls to Gel.\footnote{Gel also supports a typed binary serialization format that can avoid overheads connected to JSON serialization.} The last three examples demonstrate the convention that the \verb|id| field is included automatically if the output otherwise would be an property-less JavaScript object. This finally closes out our running example: the shaped query described in \cref{sec-exp-shapes} will return precisely the object-shaped data from \cref{figmov}, with all the benefits of storing that object-shaped data in a relational database.

\Cref{figtypescript} also illustrates a subtle shortcoming of our formalization. We can fully account for how a computed value and its type translate to a value in TypeScript; however, because visibility information is tracked in the dynamic semantics of our formal system and not the static semantics, we \textit{cannot} fully account for how a type and cardinality translate to the derived TypeScript type, since we don't have information about which properties are invisible and thus excluded from the type. We believe the appropriate fix is to track visibility in the static semantics, at which point it can be omitted in the dynamic semantics. Our relatively late detection of this deficiency illustrates the potential value, in future work, of also formalizing the relationship between the types and cardinalities used in EdgeQL and the types and values delivered by client libraries; such a effort would have likely highlighted this issue.

It is worth mentioning another way to use the TypeScript client library that avoids concrete syntax altogether in favor of a more DSL-like approach. Rather than generating code for individual queries, Gel can generate a schema-specific TypeScript library that allows queries to be built and, in most cases, statically checked by the TypeScript type system. This DSL actually follows our abstract syntax more closely than EdgeQL's concrete syntax, for instance by introducing a bound variable in shaping expressions. Compare the following to the abstract syntax in \cref{sec-exp-shapes}, for example:
\begin{quote} \small
\begin{verbatim}
const result = await e
  .select(e.Movie, x => ({
    title: x.title,
    year: x.year,
    directors: e.shape(x.directors, y => ({ name: y.name, age: y.age })),
    actors: e.shape(x.actors, y => ({
      name: y.name, "@character": y["@character"],
    })),
  }))
  .run(client);
\end{verbatim}
\end{quote}
The inferred TypeScript type of \verb|result| is a simple and specific description of the expected data following the pattern in \cref{figtypescript}.

\subsection{Gel and PostgreSQL}

Gel implements graph-relational programming on top of the PostgreSQL relational database; EdgeQL schemas are converted to PostgreSQL schemas in which each object
type has a dedicated table; links and properties with cardinality $\card{0}{1}$ or $\card{1}{1}$ are stored in columns of that table, and other links and properties are stored in separate tables. Links to other objects are represented by their primary key,  and link properties are stored in other columns of the link table. 

For simple queries on nested data, Gel's general strategy is to use the PostgreSQL-specific \verb|array_agg| aggregation function to represent nested data with arrays of tuples. Some illustrative examples of the queries that Gel generates are shown in supplemental material (\Cref{sec-compilation}).

\subsection{The Relationship Between EdgeQL and This Formalization}

The actual EdgeQL language implementation includes a number of features that do not exist in our formalization. A few notable examples:
\begin{itemize}
    \item \textbf{Exceptions.} Our formal system does not allow well-typed queries to fail, which allows us to give a crisp statement of what type correctness means. The actual implementation allows for checked runtime assertions in many situations. For example, there is an \verb|assert_single| operator which allows the programmer to ``downcast'' an expression of cardinality $\card{0}{\infty}$ to cardinality $\card{0}{1}$, failing at runtime (and reverting the database store to its state prior to the query's evaluation) if the computed sequence does not have exactly one element.
    \item \textbf{Deletion.} Another benefit of exceptions is permitting deletion, which necessarily creates the possibility of failure. In our running example,  attempting to delete the $\tbname{Person}$ record for either of our directors would have to fail because the schema requires that movies have at least one director.
    \item \textbf{Structured data.} In our formalization, schema-defined objects were the only structured data. In practice, it is quite useful to be able to create ``free'' objects on an ad-hoc basis even if there is no corresponding entry in the schema, as in this example:
    \begin{quote}\small
    \begin{verbatim}
for x in { 1, 2, 3 }
select { n := x };
\end{verbatim}
\end{quote}
    The EdgeQL implementation also has types for fixed-length tuples, variable-length arrays, and arbitrary nested and unstructured JSON data.
    \item \textbf{Additional types.} In addition to free objects, tuples, and arrays, EdgeQL has and multiple numeric types like \verb|int16|, \verb|float64|, and \verb|bigint|. Casts like \verb|<int16>55| allow for the introduction of basic scalar values with no dedicated syntax (numeric constants with no decimal are always inferred to be type \verb|int64|).
    \item \textbf{Inheritance.} Object types may be declared to inherit from other types. Selecting a parent type will also return all objects of its descendant types. 
    \item \textbf{Constraints.} In addition to exclusivity constraints mentioned in \cref{secrefvalues}, a schema can give an object type an arbitrary expression that must be true in order for that object to be stored or updated in the database.
\end{itemize}
All of these could potentially be in scope of a formalization effort like the one we have described here, but we have sought a middle-weight balance between complexity and completeness of this formalization. (Inheritance and subtyping, in particular, require extensive support in the implementation and so represent  a desirable target for future formalization.) Beyond these features, there are many technical aspects of creating a production-quality database system that are beyond the scope of the formalization effort described here, though they might benefit from future formalization: access policies, code that triggers on certain inserts, updates, or deletions, and data migrations.

Conversely, there are expressions that are well-typed in our formal system that are not accepted as EdgeQL queries in Gel. Some of these are best understood as outstanding bugs: Gel doesn't allow mutation in a number of places, for instance within filter statements, simply because this does not come up much in practical usage and so hasn't been worth the engineering effort. Other differences are better-motivated and unlikely to change. One example is that link properties in Gel are only allowed to have the cardinality $\card{0}{1}$ or $\card{1}{1}$, as the alternative was deemed complicated to implement and confusing to use. Another example is that Gel disallows the query \verb|select Person { name := 2 }| that changes the type of \lbname{name} ({\tstring} in our schema), and it disallows the query \verb|select Person { name := { "Alice", "Bob" }}| because it changes the cardinality ({\card{1}{1}} in our schema). This restriction has a desirable effect that, from the perspective of a TypeScript program using the result of a query, a the result of a shaped query is generally a subtype of the result of an unshaped query. However, over multiple attempts, we have found that modeling this restriction in the formal type system concealed more than it revealed. We believe that the correct statement of the current intended relationship between the systems is that \textbf{if the formal static semantics and the Gel implementation agree that an expression is well-formed, they should agree on the result of evaluating that query}.

There is one notable exception to this desirable criteria, and it's worth discussing as a case study on the ongoing relationship between formalization efforts and practical language design. In the actual EdgeQL system, certain projections de-duplicate their results when the result of the projection has object type. Our running example has three movies with one director each, and under the semantics we have presented in \cref{sec-expressions}, the expression $\tbname{Movie} \cdot \lbname{directors}$ evaluates to a sequence of three objects. The Gel system, on the other hand, will de-duplicate the result of projecting out $\lbname{directors}$, and so because two movies were directed by the same person, the EdgeQL query \verb|select Movie.directors| in Gel will return two objects representing the two different directors.

We could bring our formal semantics \textit{mostly} in line with the EdgeQL implementation by modifying the \textsc{E-proj} and \textsc{E-backlink} rules to de-duplicate their results when the result of projection is a sequence of object references $\refxval{\idn}{\shapevalmeta}$. De-duplication would mean that if the sequence contains two objects $\refval{\idn}{\shapevalmeta}$ and $\refval{\idn}{\shapevalmeta'}$ with the same $\idn$, one will nondeterministically be removed. However, this \textit{still} would not perfectly match EdgeQL's behavior, because not all projection operations trigger de-duplication: the EdgeQL query \verb|select Movie.actors@characters| does not de-duplicate actors; the result will include both of Sillier Murphy's characters ``Doc Boom'' and ``Fissure'' despite the fact that the EdgeQL query \verb|select Movie.actor| \textit{will} remove one of the two references to Sillier Murphy via de-duplication.

\subsection{An Ongoing Dance of Formalization, Implementation, and Revision}
\label{sec-dance}

What to make of this discrepancy between the formalization of graph-relational programming and the EdgeQL implementation? That is a ongoing conversation! From the perspective of the authors actively involved in formalization, this is a bug in the design of EdgeQL, an unambiguously non-compositional aspect of the query language. A reasonable counterpoint is offered by the authors who were involved in the original design of the language, who interpret projection as something that happens along a \textit{path} with multiple labels, not one label at a time. This would indicate a bug in the formalization: instead of the abstract syntax treating \verb|select Movie.actors@character| as $(\tbname{Movie} \cdot \lbname{actors}) \cdot \lbname{character}$, it is possible to parse the query as $\tbname{Movie} \cdot (\lbname{actors}, \lbname{character})$ at which point it is reasonable for a compositional semantics analyze the \textit{entirety} of the projection path to determine whether and where object values should be de-duplicated. Certainly, though, the discrepancy should be resolved one way or the other.

Work on the formalization presented here has been in progress since late 2022 while Gel 3.0 was in development, and has continued through to ongoing development of Gel 7.0. At least one significant change to the language was driven in part by the formalization effort: a change to the rules by which a table name like \verb|Movie| goes from being a reference to all movies to being a reference to a single movie. The discussion of binding and shadowing in \cref{sec-concrete-nicities} describes behavior introduced as Gel 6.0 to replace a now-deprecated set of rules that have a less compositional interpretation \cite{GelPathScoping}. These ``legacy path scoping'' rules were formalized in an earlier (unpublished) iteration of this paper; describing the system formally helped highlight the complexity of legacy path scoping relative to the actual utility of this feature for EdgeQL users. That analysis was part of the motivation for legacy path scoping's subsequent re-evaluation and eventual removal. 

These two vignettes on de-duplication and legacy path scoping are illustrative case studies in how formalization can contribute to language design and evolution. In the case of legacy path scoping, the formalization effort highlighted the complex and non-compositional nature of the feature. In the case of de-duplication, compositionality is in the eye of the beholder. Without question, EdgeQL's current de-duplication logic ``goes against the grain'' of modern programming language design,\footnote{~``[A] medium's grain comprises the tendencies, advantages, and constraints which emerge from its aggregate low-level properties, but which are conceptualized abstractly by skilled creators.'' \cite{10.1145/3544548.3581434}} and there is certainly positive correlation between how formally trained individual authors are in modern programming language design and how confusing they find the behavior of the current implementation. But ultimately, for software intended to address the needs of its users, the original design vision and the design vision implicit in the modern practice programming formalization are just two (important) inputs to a design process. Neither can provide a definitive answer without a broader understanding of the needs of a system's intended users.

\subsection{Future Use of Formalization Efforts}

Language formalization efforts such as this one, at their best, serve as a peculiar kind of fuzz testing; counterintuitive aspects of a formalization naturally prompt the question ``how does the implementation handle this?'' A number of bugs have been identified in this way. Going forward, we see significant potential in using a semantics-derived prototype implementation to systematically fuzz test the implementation. 
The formalization effort has spawned a number of candidate prototype implementations of the graph-relational model, including at least one intended to faithfully interpret the semantics at the cost of efficiency and one that attempts to use SQLite and its high tolerance for many small queries to avoid much of the complexity of the PostgreSQL-backed implementation.

Another potential win would come from using a direct implementation of the semantics to check the test suites for reliance on implementation-defined behavior: it is very easy to write a test that expects $1 \cup 2$ to evaluate to $[1,2]$ specifically, but $[2,1]$ is also allowed by the language definition. 
Our intended property of having the dynamic semantics agree on the evaluation of any expression that both static semantics treat as well-formed is sufficient to make these use cases viable.

Being able to provide mechanized guarantees about the relationship between formalization and implementation is a tantalizing, but likely distant, goal.

\section{Evaluation}
\label{sec-eval}

\begin{figure}
    \centering

\pgfplotsset{width=13cm,height=4cm,compat=1.18}
\begin{tikzpicture}
\begin{axis}[
ybar,
enlargelimits=0.1,
ymin=68,
ymax=750,
ylabel={iterations/s},
symbolic x coords={Django (Py), SQLAlchemy (Py), TypeORM (JS), Sequelize (JS), Prisma (JS), Drizzle (JS), Gel+Py,Gel+JS, asyncpg},
x tick label style={anchor=east,yshift=-0.5em,xshift=1em, rotate=30,align=left},
nodes near coords,
nodes near coords align={vertical},
]
\addplot coordinates {(Django (Py), 71) (SQLAlchemy (Py), 149) (TypeORM (JS), 182) (Sequelize (JS), 210) (Prisma (JS), 109) (Drizzle (JS), 284) (Gel+Py,544) (Gel+JS,530) (asyncpg,658)};
\end{axis}
\end{tikzpicture}
%
%
%
%
%
%
%
    \caption{Comparison of Gel's throughput (geometric mean across on three IMDBench benchmarks) against two Python ORMs, four Node.JS ORMs, and direct use of the \textbf{asyncpg} Python client library for PostgreSQL using hand-tuned queries.}
    \label{fig:perf}
\end{figure}
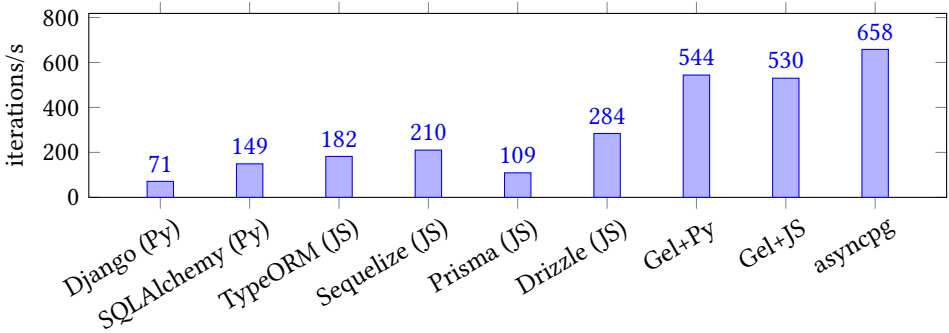

This section summarizes updated results from a separate effort aimed at creating a performance benchmark for the Gel system \cite{whyorms}. The benchmark queries object-shaped data from an IMDB-like database of movies, inserting a one-millisecond delay for all communication between the application and database to simulate a network delay. This delay has the effect of additionally punishing approaches that require multiple queries to fetch data.

A relatively straightforward success criteria for Gel is that the ``obvious way'' of manipulating object-shaped data with EdgeQL should be faster (lower latency and/or higher throughput) than the similarly-obvious way of performing the same manipulations in an ORM library. \Cref{fig:perf} shows that Gel clearly meets the first success criteria with regards to throughput, and additional benchmark results (including tests for latency) are included in supplemental materials (\Cref{sec-perf}).

A stronger success criteria is that, in a specific application language, the Gel client library should provide the same performance as running hand-tuned PostgreSQL queries with a language-specific Postgres client library. We believe the \texttt{asyncpg} library in \cref{fig:perf} represents the best throughput that is practically possible, and Gel comes within 20\% of the performance of hand-optimized queries using this approach. In practice, there are cases where Gel \textit{exceeds} the best possible performance achievable with popular PostgreSQL client libraries. This can happen if the client library has missing or inefficient support for PostgreSQL's arrays, necessitating the use of multiple queries (the second approach described in \cref{objectrelationalmismatch}).

\section{Conclusion}
\label{sec-conclusion}

We have described a formal model of graph-relational programming and the realization of this model in the EdgeQL query language for the Gel system. We have not intended to provide the final and definitive word on the details of graph-relational programming; rather, we have sought to present graph-relational programming as a previously-unexplored ``sweet spot'' of expressiveness and practicality for the extremely common use case of dealing with object-shaped data, while at the same time preserving many benefits that come with using relational databases. There is significant work adjacent to ours; we will conclude with an incomplete review of some connections that were not previously mentioned, with an emphasis on those that were a specific inspiration to the design of EdgeQL.

In \cref{sec-exp-shapes} we discussed the connections between and EdgeQL's shapes and GraphQL; similarly, UnQL's structurally recursive transformations on semistructured data were an inspiration for this work \cite{Buneman00vldb}. Other languages, particularly those like Lorel that query semistructured data, have previously introduced the ``dot notation'' for links that translates to database joins \cite{Abiteboul97ijdl}. It is common for languages with dot-projection to allow set-valued attributes \cite{Kim91csc,Zaniolo83sigmod}, which can be viewed as rejecting the ``first normal form'' assumption; however, these approaches tend to treat sets and singletons as different types, in contrast to our uniform treatment that distinguishes singletons from sets with cardinality constraints. This uniform treatment goes hand-in-hand with our mostly-uniform lifting of operations (such as built-in function calls) from singleton values to sets: this is a mechanism as old as APL \cite{apliverson} but one that still receives attention today \cite{10.1145/3315454.3329961}. Compared to this other work on lifted operations, the goals of Gel are modest: our aim was to capture a usable form of lifting that facilitates efficient implementation on top of the existing PostgreSQL implementation.

We believe cardinality constraints as we have presented them are relatively novel, though given that all expressions in our model evaluate to sequences of values, it is possible to view them as an extremely simplified form of something like liquid refinement types \cite{10.1145/1375581.1375602}. While we have found our minimal set of five cardinality suitable for database-query purposes, work on refinement types could potentially be adapted to give dramatically more specific statically checked cardinality constraints.

Graph-relational programming is certainly adjacent to the significant body of work on graph databases \cite{10.5555/2846287,10.1109/ICDEW.2012.31,10.1145/1322432.1322433} and graph query languages like the recent standardization efforts for SQL/PGQ and GQL \cite{francis:hal-04094449}. Despite the shared use of the word ``graph,'' we see these as largely orthogonal lines of work! Graph databases and graph query languages are intended as powerful tools for thinking about the graph structure of data, allowing properties like graph connectivity to be queried. Graph-relational programming, on the other hand, is intended to facilitate an application's connection to strongly-typed, object-shaped data. We believe techniques from graph databases could be productively applied to the graph-relational database model just as they have been applied to the relational database model.

The goals of the graph-relational database model are, in some ways, closer to the goals of object-relational mapping (ORM) systems  \cite{10.1145/1376616.1376773,TORRES20171,9481049} than to most query languages. Graph-relational programming improves upon work on ORM systems with a relatively high degree of language independence, a general methodology for avoiding the types of inefficiency endemic to ORM systems, and a powerful and compositional language for describing database-side computations.

\begin{acks}
Credit to \href{https://357meg.com/}{Megan Sullivan} for the illustration in \cref{figmov}.
Fantix King updated the 2022 benchmark and generated the results presented in \cref{sec-eval,sec-perf}.
\end{acks}

\bibliographystyle{ACM-Reference-Format}
\bibliography{main}

\appendix

\clearpage
\section{EdgeQL and PostgreSQL in More Detail}
\label{sec-compilation}

This appendix shows more of EdgeQL's concrete syntax, and discusses how EdgeQL queries actually relate to PostgresQL queries.
We will start by repeating the actual EdgeQL schema for the paper's running example presented in \cref{sec-intro-schema-type}; compare to the abstract syntax we gave in \cref{sec-schemas}:

\vspace{0.5em}
\begin{quote}\small
\begin{verbatim}
using future simple_scoping;
module default {
  type Person {
    required name: str;
    required age: int64;
    born: str;
  };
  type Movie {
    required title: str;
    required year: int64;
    required multi directors: Person;
    multi actors: Person {
      character: str;
    };
  };
}
\end{verbatim}
\end{quote}

\vspace{0.5em}
\noindent
All the cardinality modes have specific names --- $\card{0}{1}$ is the unannotated default, $\card{1}{\infty}$ is ``required multi'' --- and the ``simple\_scoping'' directive, as of Gel 6.0, tells the system to use the scoping rules described in \cref{sec-concrete-nicities} rather than the legacy ``path factoring'' rules mentioned in \cref{sec-dance}.

The actual PostgreSQL database schema used by Gel to describe our running example is not fundamentally more complicated than what was shown in \cref{figtable} aside from Gel's extensive use of UUIDs. The actual table names are also translated to UUIDs; these examples use a Postgres interface provided by Gel that renames meaningful table names to the correct UUIDs:

\vspace{0.5em}
\begin{quote}\small
\begin{verbatim}
main=# SELECT id, title, year FROM "Movie";
                  id                  |    title     | year
--------------------------------------+--------------+------
 396d42de-59d8-11f0-afe7-abd8660e07dc | Transistors  | 2007
 396da954-59d8-11f0-afe7-c71c5a3aed98 | Interception | 2010
 39a690de-59d8-11f0-afe7-dfc8a673e1d6 | Open Hammer  | 2024
\end{verbatim}
\end{quote}

\vspace{0.5em}
\noindent
We'll just show the last 5 characters of ids for the rest of the queries:

\vspace{0.5em}
\begin{quote}
    \small
    \begin{verbatim}
main=# SELECT right(id::text, 5) AS id, age, born, name FROM "Person";
  id   | age |    born     |      name
-------+-----+-------------+----------------
 03ed4 |  38 | California  | Megan Wolf
 a0467 |  38 | Los Angeles | Shy Andbuff
 61688 |  38 | Ottawa      | Elton Book
 8eb94 |  50 | New York    | Leo Tophat
 98bad |  49 | Ireland     | Sillier Murphy
 6021d |  60 | The moon    | Michael Cove
 8a4a6 |  50 | London      | Chris Nolens
 f41c1 |  41 | London      | Em Sharp
\end{verbatim}
\end{quote}

\clearpage
\begin{quote}
    \small
\begin{verbatim}
main=# SELECT right(id::text, 5) as id, title, year FROM "Movie";
  id   |    title     | year
-------+--------------+------
 e07dc | Transistors  | 2007
 aed98 | Interception | 2010
 3e1d6 | Open Hammer  | 2024
                                                             ^
main=# SELECT
  right(source::text, 5) AS source,
  right(target::text, 5) AS target
FROM "Movie.directors";
 source | target
--------+--------
 e07dc  | 6021d
 aed98  | 8a4a6
 3e1d6  | 8a4a6

main=# SELECT
  right(source::text, 5) AS source,
  right(target::text, 5) AS target,
  character
FROM "Movie.actors";
 source | target | character
--------+--------+------------
 e07dc  | 03ed4  | Meg Tech
 e07dc  | a0467  | Sam Man
 aed98  | 61688  | Spiderface
 aed98  | 8eb94  | Corn Cobb
 aed98  | 98bad  | Fissure
 3e1d6  | 98bad  | Doc Boom
 3e1d6  | f41c1  | Cat Boom
\end{verbatim}
\end{quote}

\vspace{0.5em}
\noindent
The concrete analogue of the table in \cref{figjoin} is generated like this:

\vspace{0.5em}
\begin{quote}\small
\begin{verbatim}
main=# SELECT
  right("Movie".id::text, 5) AS movie_id,
  title,
  character,
  right("Person".id::text, 5) AS person_id,
  name
FROM "Movie"
LEFT JOIN "Movie.actors" ON "Movie".id = "Movie.actors".source
INNER JOIN "Person" on "Movie.actors".target = "Person".id;
 movie_id |    title     | character  | person_id |      name
----------+--------------+------------+-----------+----------------
 e07dc    | Transistors  | Meg Tech   | 03ed4     | Megan Wolf
 e07dc    | Transistors  | Sam Man    | a0467     | Shy Andbuff
 aed98    | Interception | Spiderface | 61688     | Elton Book
 aed98    | Interception | Corn Cobb  | 8eb94     | Leo Tophat
 aed98    | Interception | Fissure    | 98bad     | Sillier Murphy
 3e1d6    | Open Hammer  | Doc Boom   | 98bad     | Sillier Murphy
 3e1d6    | Open Hammer  | Cat Boom   | f41c1     | Em Sharp
\end{verbatim}
\end{quote}

\clearpage

These examples show that the PostgreSQL \textit{schemas} produced by Gel are essentially idiomatic SQL tables. The \textit{queries} that Gel creates are \textit{not} idiomatic SQL queries for a variety of reasons, including features like support for reflection and subtyping that go beyond the scope of this paper. But the fundamental trick that Gel is using to handle nested queries is conceptually straightforward: Gel uses PostgreSQL's built-in array type and the nonstandard array aggregation function to collect an entire column or set of columns into a single value:

\vspace{0.5em}
\begin{quote}\small
\begin{verbatim}
main=# SELECT array_agg("Movie.actors".character) FROM "Movie.actors";
                                  array_agg
-----------------------------------------------------------------------------
 {"Meg Tech","Sam Man",Spiderface,"Corn Cobb",Fissure,"Doc Boom","Cat Boom"}
\end{verbatim}
\end{quote}

\vspace{0.5em}
\noindent
Array values and array aggregation are the nonstandard SQL feature that Gel relies on. Array aggregation was added to PostgreSQL in 2009,\footnote{\url{https://www.postgresql.org/docs/8.4/release-8-4.html\#AEN97667}}, but arrays and array aggregation are only available in the popular MySQL database system through a generic JSON type\footnote{\url{https://dev.mysql.com/doc/refman/8.4/en/aggregate-functions.html\#function_json-arrayagg}} that supports array aggregation as of 2018.\footnote{\url{https://dev.mysql.com/doc/relnotes/mysql/5.7/en/news-5-7-22.html}} Combined with standard SQL tuple type, array aggregation allows an entire query's tabular results to be captured in a single value:

\vspace{0.5em}
\begin{quote}\small
\begin{verbatim}
main=# SELECT array_agg(("Person".name, "Movie.actors".character))
FROM "Movie"
LEFT JOIN "Movie.actors" ON "Movie".id = "Movie.actors".source
INNER JOIN "Person" ON "Movie.actors".target = "Person".id
WHERE "Movie".title = 'Transistors';
                             array_agg
-------------------------------------------------------------------
 {"(\"Megan Wolf\",\"Meg Tech\")","(\"Shy Andbuff\",\"Sam Man\")"}
\end{verbatim}
\end{quote}

\vspace{0.5em}
\noindent
The only remaining trick is that PostgreSQL's \verb|array_agg| function returns no value, rather than a single empty array, if asked to aggregate zero elements, so when projecting a field like \verb|actors| that might have zero elements (due to its {\card{0}{\infty}} cardinality mode), a call to SQL's \verb|COALESCE()| function is required to insert an empty array as the value for movies with no actors. If we did not do that, movies with no actors be omitted from the results.

\vspace{0.5em}
\begin{quote}\small
\begin{verbatim}
main=# SELECT "Movie".title, "Movie".year,
  ( SELECT array_agg(("Person".name, "Person".age))
    FROM "Movie.directors"
    INNER JOIN "Person" on "Movie.directors".target = "Person".id
    WHERE "Movie".id = "Movie.directors".source
  ) AS directors,
  ( SELECT COALESCE(
      array_agg(("Person".name, "Movie.actors".character)),
      (ARRAY[])::record[]
    )
    FROM "Movie.actors"
    INNER JOIN "Person" on "Movie.actors".target = "Person".id
    WHERE "Movie".id = "Movie.actors".source
  ) AS actors
FROM "Movie";
\end{verbatim}
\end{quote}

\clearpage

\noindent
The results from this query look like this:

\vspace{0.5em}
\begin{quote}\small
\begin{verbatim}
    title     | year |         directors         |                      actors
--------------+------+---------------------------+----------------------------...
 Transistors  | 2007 | {"(\"Michael Cove\",60)"} | {"(\"Megan Wolf\",\"Meg Tec...
 Interception | 2010 | {"(\"Chris Nolens\",50)"} | {"(\"Elton Book\",Spiderfac...
 Open Hammer  | 2024 | {"(\"Chris Nolens\",50)"} | {"(\"Sillier Murphy\",\"Doc...
\end{verbatim}
\end{quote}

\vspace{0.5em}
\noindent
These results include all the relevant information in the correct nesting structure; Gel also provides a JSON output mode that directly aggregates to the structure from \cref{figmov} within the PostgreSQL query.
\section{Extended Proofs}
\label{sec-proofs}

We'll start with a more precise treatment of the notion of database extension that was introduced in \cref{sec-metatheory}, and a few properties about extension and its relationship to the static and dynamic semantics.

\begin{definition}[Database store extension]
    We say that $\dbn'$ \textit{extends} $\dbn$ if:
\begin{itemize}
    \item $\forall \idn. \;
    (\idn, \tbn, \lkn, \rec{\ldots}) \in \dbn
    \Rightarrow \exists \lkn, \ldots\, (\idn, \tbn, \lkn', \rec{\ldots}) \in \dbn'$
    \item $\forall \idn. \;
    (\idn, \tbn, \lkopen, \rec{\ldots}) \in \dbn'
    \Rightarrow (\idn, \tbn, \lkopen, \rec{\ldots}) \in \dbn$
\end{itemize}
\end{definition}
Extension captures a notion that $\dbn$ and $\dbn'$ agree on types and that ``unlocked'' tuples in the extended store are unchanged from the initial store.

\begin{lemma}\label{lem-extension-props}
    Extension is symmetric ($\dbn$ extends $\dbn$) and transitive (if $\dbn_2$ extends $\dbn_1$ and $\dbn_3$ extends $\dbn_2$, then $\dbn_3$ extends $\dbn_1$).
\end{lemma}
\begin{proof}
    Immediate from definitions.
\end{proof}

\begin{lemma}[Evaluation implies extension]\label{label-eval-pres-extension}
    If $\deval{\vctxn}{\dbn_1}{\dbn_2}{e}{\dsetmeta}$, then $\dbn_2$ extends $\dbn_1$.
\end{lemma}
\begin{proof}
    By induction on the evaluation derivation. \textsc{E-insert} adds new tuples, but they are locked, and \textsc{E-update} modifies tuples, but does not change the associated type, and always locks the modified tuple.
\end{proof}

Proving \Cref{label-eval-pres-extension} separately simplifies the statements of \cref{theorem-preservation} and \cref{theorem-progress}.

\begin{lemma}[Extension preserves value typing]\label{extension-typing}
If $\dbn_2$ extends $\dbn_1$,
\begin{itemize}
    \item If $\vdash_{\schema, \dbn_\textit{init}, \dbn_1} \dmeta : \typmeta$  then  $\vdash_{\schema, \dbn_\textit{init}, \dbn_2} \dmeta : \typmeta$.
    \item If $\vdash_{\schema, \dbn_\textit{init}, \dbn_1} \tmap{\dsetmeta}{\typmeta}{\cardmeta}$  then  $\vdash_{\schema, \dbn_\textit{init}, \dbn_2} \tmap{\dsetmeta}{\typmeta}{\cardmeta}$.
    \item If $\vdash_{\schema, \dbn_\textit{init}, \dbn_1} \vctxn : \tctxn$ then
    $\vdash_{\schema, \dbn_\textit{init}, \dbn_2} \vctxn : \tctxn$.
    \item If $\vdash_{\schema, \dbn_1} \tmap{\vsetmeta}{\tmeta}{\cardmeta}$
    then $\vdash_{\schema, \dbn_2} \tmap{\vsetmeta}{\tmeta}{\cardmeta}$.
\end{itemize}
\end{lemma}
\begin{proof}
    Straightforward induction over the relevant typing derivation, and the definition of extension.
\end{proof}

\subsection{Proof of Theorem 4.3, Preservation}

To restate, if we have:
\begin{enumerate}
    \item $\dbn_1$ \textit{extends} $\dbn_\textit{init}$
    \item $\tctxn \vdash_\schema \tmap{\emeta}{\typmeta}{\cardmeta}$,
    \item $\strut\vdash_{\schema} \dbn_\textit{init}\,\mathsf{ok}$,
    \item $\strut\vdash_{\schema} \dbn_\textit{before}\,\mathsf{ok}$,
    \item $\strut\vdash_{\schema, \dbn_\textit{init}, \dbn_\textit{before}} \vctxn : \tctxn$, and
    \item $\deval{\vctxn}{\dbn_\textit{before}}{\dbn_\textit{after}}{\emeta }{\dsetmeta}$,
\end{enumerate}
then $\vdash_{\schema, \dbn_\textit{init}, \dbn_\textit{after}}\strut \tmap{\dsetmeta}{\typmeta}{\cardmeta}$ and $\strut\vdash_{\schema} \dbn_\textit{after}\,\mathsf{ok}$

\medskip
\noindent
The proof is by induction on the evaluation derivation. Representative cases are shown:

\[
\inferrule[(E-var)]
{x \mapsto \dsetmeta \in \vctxn}
{\seval{\vctxn}{\dbn}{\dbn}{x}{\dsetmeta}}
\]
To show: $\mathstrut \vdash_{\schema, \dbn_\textit{init}, \dbn}\strut \tmap{\dsetmeta}{\typmeta}{\cardmeta}$ and $\mathstrut \vdash_\schema \dbn\,\textsf{ok}$. The latter is immediate from (4).

By inversion on the typing judgment (2) the first derivation must end with \textsc{T-var}, and we have $\tmap{x}{\typmeta}{\cardmeta} \in \tctxn$. From (5) and the definition of environment typing (page~\pageref{pos-envtypingdef} in \cref{sec-metatheory}) we have
$\vdash_{\schema, \dbn_\textit{init},\dbn}\strut \tmap{\dsetmeta}{\typmeta}{\cardmeta}$ as required.

\[
\inferrule[(E-emp)]
{ }
{\seval{\vctxn}{\dbn}{\dbn}{\varnothing_\typmeta}{[]}}
\]
To show: $\mathstrut \vdash_{\schema, \dbn_\textit{init}, \dbn}\strut \tmap{[]}{\typmeta}{\cardmeta}$ and $\mathstrut \vdash_\schema \dbn\,\textsf{ok}$. The latter is immediate from (4).

By inversion on the typing judgment, the first derivation must end with \textsc{T-emp}, $\cardmeta = \card{0}{0}$. We derive $\vdash_{\schema, \dbn_\textit{init}, \dbn} \tmap{[]}{\typmeta}{\card{0}{0}}$ as required with \textsc{VT-set}.

\[
\inferrule[(E-union)]
{
\deval{\vctxn}{\dbn_a}{\dbn_b}{\emeta_1}{\dsetmetax{1}}
\\\\
\deval{\vctxn}{\dbn_b}{\dbn_c}{\emeta_2}{\dsetmetax{2}}
\\
\left[ \dsetmetax{1} , \dsetmetax{2} \right] \approx \dsetmeta
}
{\deval{\vctxn}{\dbn_a}{\dbn_c}{\emeta_1 \cup \emeta_2}{\dsetmeta}}
\]
To show: $\mathstrut \vdash_{\schema, \dbn_\textit{init}, \dbn_c} \tmap{\dsetmeta}{\typmeta}{\cardmeta}$ and $\mathstrut \vdash_\schema \dbn_c\,\textsf{ok}$.

By inversion on the typing judgment, the derivation must end with \textsc{T-union} and we have $\tctxn \vdash_\schema \tmap{\emeta_1}{\typmeta}{\cardmeta_1}$ and $\tctxn \vdash_\schema \tmap{\emeta_2}{\typmeta}{\cardmeta_2}$ and $\cardmeta = \cardmeta_1 + \cardmeta_2$.

Because evaluation preserves extension and extension is transitive (\cref{label-eval-pres-extension,lem-extension-props}), $\dbn_b$ extends $\dbn_\textit{init}$, $\dbn_b$ extends $\dbn_a$, and  $\dbn_c$ extends $\dbn_b$. Because extension preserves value typing (\cref{extension-typing}), $\strut\vdash_{\schema, \dbn_\textit{init}, \dbn_b} \vctxn : \tctxn$.

By the induction hypothesis, $\mathstrut\vdash_{\schema, \dbn_\textit{init}, \dbn_b} \dsetmetax{1} : \cardmeta_1$
and $\mathstrut\vdash_{\schema}\dbn_b\,\textsf{ok}$, and again by the induction hypothesis
$\mathstrut\vdash_{\schema, \dbn_\textit{init}, \dbn_c} \dsetmetax{2} : \cardmeta_2$
and $\mathstrut\vdash_{\schema}\dbn_c\,\textsf{ok}$.

By the third premise, $\dsetmeta$ is some permutation of the concatenation of $\dsetmetax{1}$ and $\dsetmetax{2}$, and all the elements $\dmeta$ in those sets have
$\vdash_{\sigma, \dbn_\textit{init}, \dbn_c} \tmap{\dmeta}{\typmeta}{\cardmeta}$. (This involves a use of \cref{extension-typing} for the elements in $\dsetmetax{1}$.)

The result follows from \textsc{VT-set} and the following lemma, which ensures that the combined list has a length in $\cardmeta_1 + \cardmeta_2$:

\begin{lemma}
    If $l_i \in \cardmeta_i$ for $i = 1\ldots n$, then $l_1 + \ldots + l_n \in \cardmeta_1 + \ldots + \cardmeta_n$.
\end{lemma}
\begin{proof}
    Follows from definition of cardinality mode.
\end{proof}

\[
\inferrule[(E-name)]
{\bigl[\{\refval{\idn}{}\} \mid (\idn, \tbn, \lkn, \rec{\ldots}) \in \dbnx{init}\bigr] \approx \dsetmeta}
{\deval{\vctxn}{\dbn}{\dbn}{\tbn}{\dsetmeta}}
\]
To show: $\mathstrut \vdash_{\schema, \dbn_\textit{init}, \dbn}\strut \tmap{\dsetmeta}{\typmeta}{\cardmeta}$ and $\mathstrut \vdash_\schema \dbn\,\textsf{ok}$. The latter is immediate from (4).

By inversion on the typing judgment, the derivation must end with \textsc{T-name}, $\typmeta = \reft{\tbn}{}$, and $\cardmeta = \card{0}{\infty}$. The result follows immediately from \textsc{VT-stored-ref} and \textsc{VT-set}.

\[\inferrule[(E-proj)]
{
\deval{\vctxn}{\dbn_a}{\dbn_b}{\emeta}{\dsetmetax{e}}
\\\\
 \left[ \mathsf{project}(\dbn_\textit{init}, \lbn, \dmeta) \mid \dmeta \in \dsetmetax{e} \right] \approx \dsetmetax{b}
}
{\deval{\vctxn}{\dbn_a}{\dbn_b}{\emeta \cdot \lbn}{\dsetmetax{b}}}
\]
To show: $\mathstrut \vdash_{\schema, \dbn_\textit{init}, \dbn_b}\strut \tmap{\dsetmetax{b}}{\typmeta_\textit{after}}{\cardmeta}$ and $\mathstrut \vdash_\schema \dbn_b\,\textsf{ok}$. The latter follows straightforwardly from the induction hypothesis and the extension lemmas, the former is more interesting.

By inversion on the typing judgment, the derivation must end in one of two ways, \textsc{T-proj-db} or \textsc{T-proj-ext}.

In the first case, $\typmeta_\textit{after} = \tmeta$, $\cardmeta'' = \cardmeta \times \cardmeta'$,
and we have $\tctxn \vdash_{\schema} \tmap{e}{\reft{\tbn}{\tmap{\lbnx{1}}{\typmeta_1}{\cardmeta_1},\ldots,\tmap{\lbnx{n}}{\typmeta_n}{\cardmeta_n}}}{\cardmeta'}$
with $\lbn$ not among the $\lbn_i$ and $\tmeta$ and $\cardmeta$ as the type that the schema assigns to $\lbn$.  By the induction hypothesis $\vdash_{\schema,\dbn_\textit{init}, \dbn_b} \tmap{\dsetmetax{e}}{\reft{\tbn}{\tmap{\lbnx{1}}{\typmeta_1}{\cardmeta_1},\ldots,\tmap{\lbnx{n}}{\typmeta_n}{\cardmeta_n}}}{\cardmeta'}$. For each $\dmeta \in \dsetmetax{e}$, projection will fall through to the second case, resulting in a $\vsetmeta$ such that $\vdash_{\schema, \dbn_\textit{init}} \tmap{\vsetmeta}{\tmeta}{\cardmeta}$. This is the stored value typing judgment from \cref{figvaluetypes} rather than the computed value typing judgment from \cref{figcomputedvaluetypes}, so we need an extra lemma (and an appeal to extension lemmas) to finish this case.

\begin{lemma}
If $\vdash_{\schema, \dbn} \tmap{\vsetmeta}{\tmeta}{\cardmeta}$ (stored value typing) then $\vdash_{\schema, \dbn, \dbn'} \tmap{\vsetmeta}{\tmeta}{\cardmeta}$ (computed value typing).
\end{lemma}
\begin{proof}
    Proof by cases, uses of \textsc{ST-ref} become uses of \textsc{VT-stored-ref}.
\end{proof}

If the derivation ends with \textsc{T-proj-ext}, then each $\dmeta \in \dsetmetax{e}$ has an associated sequence of values for the label $\lbn$, so the projection function will always return the mapping stored in the value. The result follows from inverting the value typing judgment; whether it is proved with \textsc{VT-stored-ref} or \textsc{VT-new-ref} there is a premise that all values stored within have the expected type.

The expected cardinality follows from the relevant property of cardinality modes, given by this lemma:

\begin{lemma}
    If $l_i \in \cardmeta$ for $i = 1\ldots n$, and $n \in \cardmeta'$ then $l_1 + \ldots + l_n \in \cardmeta \times \cardmeta'$.
\end{lemma}
\begin{proof}
    Follows from definition of cardinality mode.
\end{proof}

\[
\inferrule[(E-insert)]
{
\schema(\tbn) = \rec{\tmap{\lbnx{1}}{\tmeta_1}{\cardmeta_1},\ldots,\tmap{\lbnx{n}}{\tmeta_n}{\cardmeta_n}}
\\
\deval{\vctxn}{\dbn_1}{\dbn_2}{\emeta_1}{\dsetmetax{1}}
\quad
\cdots
\quad
\deval{\vctxn}{\dbn_n}{\dbn_{n+1}}{\emeta_n}{\dsetmetax{n}}
\\
\dsetmetax{1} : \tmeta_1 \rhd \vsetmetax{1}
\quad
\cdots
\quad
\dsetmetax{n} : \tmeta_n \rhd \vsetmetax{n}
\\
\idn\#\dbn_\textit{init}
\\
\idn\#\dbn_{n+1}
\\
\dbn' = \dbn_{n+1} \cup \left\{\left(
  \idn, \tbn, \lkclosed\,,
   \rec{
     \maprec{\lbnx{1}}{\vsetmetax{1}},\ldots,
     \maprec{\lbnx{n}}{\vsetmetax{n}}
   }
\right)\right\}
}
{\deval{\vctxn}{\dbn_1}{\dbn'}{
\mathsf{insert}\,\tbn\,
\rec{\maprec{\lbnx{1}}{\emeta_1}\recsep\ldots\recsep\maprec{\lbnx{n}}{\emeta_n}}
}
{[
\refval{\idn}{
\imaprec{\lbnx{1}}{\dsetmetax{1}},\ldots,
\imaprec{\lbnx{n}}{\dsetmetax{n}}}
]}
}
\]
To show: $\mathstrut \vdash_{\schema, \dbn_\textit{init}, \dbn'}\strut \tmap{[
\refval{\idn}{
\imaprec{\lbnx{1}}{\dsetmetax{1}},\ldots,
\imaprec{\lbnx{n}}{\dsetmetax{n}}}
]}{\typmeta}{\cardmeta}$ and $\mathstrut \vdash_\schema \dbn'\,\textsf{ok}$.

By inversion on the typing judgment, the derivation must end with \textsc{T-insert} with $n$ sub-derivations that end with either \textsc{TS-simple} or \textsc{TS-reference}. We have \schema(\tbn) = \rec{\tmap{\lbnx{1}}{\tmeta_1}{\cardmeta_1},\ldots,\tmap{\lbnx{n}}{\tmeta_n}{\cardmeta_n}}, and for the $n$ other sub derivations $0 \leq i \leq n$, we have $\tctxn \vdash_{\schema} \emeta_i \div \typmeta_i/\tmeta_i \hash \cardmeta_i$.

For each of the $n$ sub-derivations, we call the I.H., use the extension lemmas as needed, and perform inversion on the resulting derivation to get $\mathstrut\vdash_{\schema, \dbn_{\textit{init}}, \dbn_{i+1}} \tmap{\dsetmeta_i}{\typmeta_i}{\cardmeta'_i}$ where $\cardmeta'_i \leq \cardmeta_i$ and $\mathstrut\vdash_\schema \dbn_{i+1}\,\mathsf{ok}$. That means each
$\vdash_{\schema, \dbn_{\textit{init}}, \dbn_{\textit{i+1}}} \tmap{\dsetmeta_i}{\typmeta_i}{\cardmeta_i}$ as well by the definition of value sequence typing.

The identifier $\idn$ does \textit{not} appear in $\dbn_\textit{init}$, but the value we return has all the labels that appear in the schema, so the only way to type the result is with \textsc{VT-new-ref} and \textsc{VT-set}.

We have $\vdash_\schema \dbn_{n+1}\,\mathsf{ok}$, and to show that $\vdash_\schema \dbn'\,\mathsf{ok}$ it suffices to show that $\vdash_{\schema, \dbn_{n+1}} \tmap{\vsetmetax{i}}{\tmeta_i}{\cardmeta_i}$. For this, we use a lemma:

\begin{lemma}[Preservation for storage stripping]\label{lemma-stripping-pres}
    If $\vdash_{\schema, \dbn_\textit{init}, \dbn} \tmap{\dsetmeta}{\typmeta}{\cardmeta}$ (computed value typing),
    $\tctxn \vdash_{\schema} \emeta \div \typmeta/\tmeta \hash \cardmeta$, and $\dsetmeta : \tmeta \rhd \vsetmeta$, then $\vdash_{\schema, \dbn} \tmap{\vsetmeta}{\tmeta}{\cardmeta}$ (stored value typing).
\end{lemma}
\begin{proof}
    Inspection of the possible derivations of $\tctxn \vdash_{\schema} \emeta \div \typmeta/\tmeta \hash \cardmeta$.
\end{proof}

\[
\inferrule[(E-update)]
{
\seval{\vctxn}{\dbn_0}{\dbn_1}{\emeta}{[\dmeta]}
\\
\dmeta = \refval{\idn}{\ldots}
\\\\
\seval{\vctxn, x \mapsto [\dmeta]}{\dbn_1}{\dbn_2}{\emeta_1}{\dsetmetax{1}}
\quad
\cdots
\quad
\seval{\vctxn, x \mapsto [\dmeta]}{\dbn_n}{\dbn_{n+1}}{\emeta_n}{\dsetmetax{n}}
\\
\dbn_{n+1} - \left\{\left(\idn, \tbn, \lkn, \rec{
\maprec{\lbnprimex{1}}{\vsetmetaprimex{1}}\recsep\ldots\recsep\maprec{\lbnprimex{n}}{\vsetmetaprimex{n}}
}\right)\right\} = \dbn'
\\
\dbn'\#\idn
\\\\
\schema(\tbn) = \rec{\tmap{\lbnx{1}}{\tmeta_1}{\cardmeta_1},\ldots,\tmap{\lbnx{n}}{\tmeta_n}{\cardmeta_n},\ldots}
\\
\dsetmetax{1} : \tmeta_1 \rhd \vsetmetax{1}
\quad
\cdots
\quad
\dsetmetax{n} : \tmeta_n \rhd \vsetmetax{n}
\\
\dbn'' = \dbn' \cup \left\{\left(
\idn, \tbn, \lkclosed\,, \rec{
\left\{
\maprec{\lbnprimex{i}}{\vsetmetaprimex{i}}
\mid \lbnprimex{i} \notin \{ \lbnx{1}, \ldots, \lbnx{n} \}
\right\}
\cup
\left\{
\maprec{\lbnx{1}}{\vsetmetax{1}},\ldots,\maprec{\lbnx{n}}{\vsetmetax{n}}
\right\}
}
\right)\right\}
\\\\
{\begin{cases}
\dbn_\textit{result} = \dbn_{n+1}, \dsetmetaprime = [] & \text{if }  \idn\#\dbn_{n+1} \text{ or if } \lkn = \lkclosed\\
\dbn_\textit{result} = \dbn'', \dsetmetaprime = [\refval{\idn}{\imaprec{\lbnx{1}}{\dsetmetax{1}},\ldots,\imaprec{\lbnx{n}}{\dsetmetax{n}}}] & \textit{if } \dbn_{n+1} \neq \dbn' \text{ and } \lkn = \lkopen
\end{cases}}
}
{\seval{\vctxn}{\dbn_0}{\dbn_\textit{result}}
{
\mathsf{update}^\dagger\,\emeta\,\mathsf{set}\,x.\rec{\maprec{\lbnx{1}}{\emeta_1}\recsep\ldots\recsep\maprec{\lbnx{n}}{\emeta_n}}
}
{
\dsetmetaprime
}
}
\]
By the same logic used for the \textsc{E-insert} case we can get  for $1 \leq i \leq n$,
$\tctxn \vdash_{\schema} \emeta_i \div \typmeta_i/\tmeta_i \hash \cardmeta_i$ and
$\mathstrut\vdash_{\schema, \dbn_{\textit{init}}, \dbn_{i+1}} \tmap{\dsetmetax{i}}{\typmeta_i}{\cardmeta_i}$ where $\cardmeta'_i \leq \cardmeta_i$ and $\mathstrut\vdash_\schema \dbn_{i+1}\,\mathsf{ok}$.
 The rest of the case depends on which of the two actions update takes:
\begin{itemize}
    \item $\text{if }  \idn\#\dbn_{n+1} \text{ or if } \lkn = \lkclosed$, then no values are returned and the result is immediate from \textsc{VS-set}.
    \item Otherwise, we are performing an actual update. By the properties of database extension, the tuple that we're updating in $\dbn_{n+1}$ must also appear in $\dbn_\textit{init}$, which means that the result follows from \textsc{VT-stored-ref} and \textsc{VT-set}. We also need to show $\mathstrut\vdash_\schema \dbn''\,\mathsf{ok}$, which follows from the same \cref{lemma-stripping-pres} used for the \textsc{E-insert} case.
\end{itemize}


\subsection{Proof of Theorem 4.4, Totality}

To restate, if we have:
\begin{itemize}
    \item $\dbn_1$ \textit{extends} $\dbn_\textit{init}$
    \item $\tctxn \vdash_\schema \tmap{\emeta}{\typmeta}{\cardmeta}$,
    \item $\vdash_\schema \dbn_\textit{init} \,\mathtt{ok}$,
    \item $\vdash_\schema \dbn_1 \,\mathtt{ok}$,
    \item $\strut\vdash_{\schema, \dbn_\textit{init}} \vctxn : \tctxn$,
\end{itemize}
then there exists some $\dsetmeta$ and $\dbn_2$ such that
$\deval{\vctxn}{\dbn_1}{\dbn_2}{\emeta }{\dsetmeta}$.

\medskip
\noindent
The proof is by induction on the typing derivation: we must construct an expression typing derivation in each case. Representative cases are shown:

\[
\inferrule[(T-var)]
{ \tmap{x}{\typmeta}{\cardmeta} \in \tctxn}
{\tctxn \vdash_{\schema} \tmap{x}{\typmeta}{\cardmeta}}
\]
By the definition of environment typing (page~\pageref{pos-envtypingdef} in \cref{sec-metatheory}), there must be $\dsetmeta$ such that $x \mapsto \dsetmeta \in \vctxn$, so the result follows by \textsc{E-var} with $\dbn_2 =\dbn_1$.

\[
\inferrule[(T-prim)]
{\vsmeta : \tsmeta}
{\tctxn \vdash_{\schema} \vsmeta : \tsmeta \hash [1,1]}
\qquad
\inferrule[(T-emp)]
{ }
{\tctxn \vdash_{\schema} \varnothing_\typmeta : \typmeta \hash [0,0]}
\]
In the first case by \textsc{E-prim} ($\dsetmeta = [\vsmeta]$ and $\dbn_2 = \dbn_1$), in the second case by \textsc{E-emp} ($\dsetmeta = []$ and $\dbn_2 = \dbn_1$).

\[
\inferrule[(T-union)]
{
\tctxn \vdash_{\schema} \tmap{\emeta_1}{\typmeta}{\cardmeta_1}
\\
\tctxn \vdash_{\schema} \emeta_2 : \typmeta \hash \cardmeta_2
}
{\tctxn \vdash_{\schema} \emeta_1 \cup \emeta_2 : \typmeta \hash \cardmeta_1 + \cardmeta_2}
\]
By the I.H. on the first premise we have a $\dsetmetax{1}$ and a $\dbn_b$ with $\deval{\vctxn}{\dbn_1}{\dbn_b}{\emeta_1}{\dsetmetax{1}}$. Evaluation preserves extension (\cref{label-eval-pres-extension}), so we can invoke the I.H. on the second premise to get $\dsetmetax{2}$ and a $\dbn_c$ with $\deval{\vctxn}{\dbn_b}{\dbn_c}{\emeta_2}{\dsetmetax{2}}$. It is always possible to compute a $\dsetmeta$ such that $\left[\dsetmetax{1}, \dsetmetax{2}\right] \approx \dsetmeta$ by concatenating the two sequences, so the result follows by \textsc{E-union} with $\dbn_2 = \dbn_c$.


\[
\inferrule[(T-proj-db)]
{
\tctxn \vdash_{\schema} \tmap{e}{\reft{\tbn}{\tmap{\lbnx{1}}{\typmeta_1}{\cardmeta_1},\ldots,\tmap{\lbnx{n}}{\typmeta_n}{\cardmeta_n}}}{\cardmeta'}
\\\\
\lbn \notin \{ \lbnx{1}, \ldots, \lbnx{n} \}
\\
\schema(\tbn) = \rec{\ldots\recsep\tmap{\lbn}{\tmeta}{\cardmeta}\recsep\ldots}
}
{\tctxn \vdash_{\schema} \tmap{e \cdot \lbn}{\tmeta}{\cardmeta \times \cardmeta'}}
\]
By the I.H. on the first premise we have a $\dsetmetaprime$ and a $\dbn_b$ with $\deval{\vctxn}{\dbn_1}{\dbn_b}{\emeta}{\dsetmetaprime}$, and by preservation
$\vdash_{\schema, \dbn_\textit{init}, \dbn_b} \tmap{\dsetmetaprime}{\reft{\tbn}{\tmap{\lbnx{1}}{\typmeta_1}{\cardmeta_1},\ldots,\tmap{\lbnx{n}}{\typmeta_n}{\cardmeta_n}}}{\cardmeta'}$. Given that  $\lbn$ is in the schema for $\tbn$ and is not among the $\lbnx{i}$, the only possible way the typing of $\dsetmetaprime$ could conclude is with \textsc{VT-stored-ref}.

From this, we can conclude that each value in $\dsetmetaprime$ is a reference in the database, and given that $\dbn_\textit{init}$ is well typed, this allows us to compute $\textsf{project}(\dbn_\textit{init}, \lbn, \dmeta)$ for each $\dmeta \in \dsetmetaprime$, always taking the second branch of the definition of projection that looks up a value from the database. That allows us to construct a $\dsetmetax{b}$ such that $\left[ \mathsf{project}(\dbn_\textit{init}, \lbn, \dmeta) \mid \dmeta \in \dsetmetaprime \right] \approx \dsetmetax{b}$ and conclude with \textsc{E-proj}.

\[
\inferrule[(T-proj-ext)]
{
\tctxn \vdash_{\schema} \tmap{e}{\reft{\tbn}{\ldots\recsep\tmap{\lbn}{\typmeta}{\cardmeta}\recsep\ldots}}{\cardmeta'}
}
{\tctxn \vdash_{\schema} \tmap{e \cdot \lbn}{\typmeta}{\cardmeta \times \cardmeta'}}
\]
By the I.H. on the premise we have a $\dsetmetaprime$ and a $\dbn_b$ with $\deval{\vctxn}{\dbn_1}{\dbn_b}{\emeta}{\dsetmetaprime}$, and by preservation
$\vdash_{\schema, \dbn_\textit{init}, \dbn_b} \tmap{\dsetmetaprime}{\reft{\tbn}{\ldots\recsep\tmap{\lbn}{\typmeta}{\cardmeta}\recsep\ldots}}{\cardmeta'}$. Regardless of which rule is used to establish the type the values in $\dsetmetaprime$, typing inversion ensures that it is an object value with a mapping for $\lbn$, so the projection function will always take the first branch, allowing us to construct the result with \textsc{E-proj}.

%

\[
\inferrule[(T-function)]
{
f : \left( \typmeta_1 \hash p_1, \ldots \typmeta_n \hash p_n \right) \rightarrow \typmeta' \hash \cardmeta'
\\\\
\Gamma \vdash_\schema \tmap{\emeta_1}{\typmeta_1}{\cardmeta_1}
\quad
\cdots
\quad
\Gamma \vdash_\schema \tmap{\emeta_n}{\typmeta_n}{\cardmeta_n}
\\
\cardmeta_1 \leq \llbracket p_1 \rrbracket
\quad
\cdots
\quad
\cardmeta_n \leq \llbracket p_n \rrbracket
}
{\tctxn \vdash_{\schema} \tmap{f^\dagger(\emeta_1, \ldots, \emeta_n)}{\typmeta'}{\cardmeta'}}
\]
By $n$ repeated iterations of the induction hypothesis (supported by $n-1$ appeals to  extension lemmas), we have that
$\deval{\vctxn}{\dbn_i}{\dbn_{i+1}}{\emeta_i}{\dsetmetax{i}}$ for $1 \leq i \leq n$. By preservation, the $\dsetmetax{i}$ are within the domain of $\llbracket f \rrbracket$, so we can let $\dsetmeta = \llbracket f \rrbracket (\dsetmetax{1},\ldots,\dsetmetax{n})$
and derive the result with \textsc{E-function}.

%
%
%

\[
\inferrule[(T-insert)]
{
\schema(\tbn) = \rec{\tmap{\lbnx{1}}{\tmeta_1}{\cardmeta_1},\ldots,\tmap{\lbnx{n}}{\tmeta_n}{\cardmeta_n}}
\\
\tctxn \vdash_{\schema} \emeta_1 \div \typmeta_1 / \tmeta_1 \hash \cardmeta_1
\quad
\cdots
\quad
\tctxn \vdash_{\schema} \emeta_n \div \typmeta_n / \tmeta_n \hash \cardmeta_n
}
{\tctxn \vdash_{\schema} \tmap{\mathsf{insert}\,\tbn\,
\rec{\maprec{\lbnx{1}}{\emeta_1}\recsep\ldots\recsep\maprec{\lbnx{n}}{\emeta_n}}
}{\reft{\tbn}{
\tmap{\lbnx{1}}{\typmeta_1}{\cardmeta_1},\ldots,\tmap{\lbnx{n}}{\typmeta_n}{\cardmeta_n}
}}{\card{1}{1}}}
\]
We can further invert the $n$ premises $\tctxn \vdash_{\schema} \emeta_i \div \typmeta_i / \tmeta_i \hash \cardmeta_i$ to get smaller derivations of
$\tctxn \vdash_{\schema} \tmap{\emeta_i}{\typmeta_i}{\cardmeta'_i}$ with $\cardmeta'_i \leq \cardmeta_i$, and thus invoke the induction hypothesis $n$ times (supported by $n-1$ appeals to evaluation preserving extension, \cref{label-eval-pres-extension}) to get $\deval{\vctxn}{\dbn_i}{\dbn_{i+1}}{\emeta_i}{\dsetmetax{i}}$ for $1 \leq i \leq n$, and by preservation we have $\vdash_{\schema, \dbn_\textit{init}} \tmap{\dsetmetax{i}}{\typmeta}{\cardmeta_i}$.

It is always possible to generate an identifier $\idn$ that is fresh relative to $\dbn_\textit{init}$ and $\dbn_{n+1}$, so to conclude with \textsc{E-insert} we must only obtain a derivation of $\dsetmetax{i} : \tmeta_i \rhd \vsetmetax{i}$ for $1 \leq i \leq n$. For this, we need a lemma:

\begin{lemma}[Progress for store coercion]
    If $\tctxn \vdash_{\schema} \emeta_i \div \typmeta_i / \tmeta_i \hash \cardmeta_i$ and $\vdash_{\schema, \dbn_\textit{init}} \tmap{\dsetmetax{i}}{\typmeta_i}{\cardmeta'_i}$, then there exists some $\vsetmetax{i}$ such that we can construct $\dsetmetax{i} : \tmeta_i \rhd \vsetmetax{i}$.
\end{lemma}
\begin{proof}
    Case analysis on the given derivations.
\end{proof}

\[
\inferrule[(T-update)]
{
\tctxn \vdash_{\schema} \tmap{\emeta}{\typmeta}{\card{1}{1}}
\\
\typmeta = \reft{\tbn}{\ldots}
\\
\schema(\tbn) = \rec{\tmap{\lbnx{1}}{\tmeta_1}{\cardmeta_1},\ldots,\tmap{\lbnx{n}}{\tmeta_n}{\cardmeta_n}, \ldots}
\\
\tctxn, \tmap{x}{\typmeta}{\card{1}{1}} \vdash_{\schema} \emeta_1 \div \typmeta_1/\tmeta_1 \hash \cardmeta_1
\quad
\cdots
\quad
\tctxn, \tmap{x}{\typmeta}{\card{1}{1}} \vdash_{\schema} \emeta_n \div \typmeta_n/\tmeta_n \hash \cardmeta_n
}
{\tctxn \vdash_{\schema} \tmap{\mathsf{update}^\dagger\,\emeta\,\mathsf{set}\,
x.\rec{\maprec{\lbnx{1}}{\emeta_1}\recsep\ldots\recsep\maprec{\lbnx{n}}{\emeta_n}}
}{\reft{\tbn}{
\tmap{\lbnx{1}}{\typmeta_1}{\cardmeta_1},\ldots,\tmap{\lbnx{n}}{\typmeta_n}{\cardmeta_n}
}}{\card{0}{1}}}
\]
By the induction hypothesis we have $\deval{\vctxn}{\dbn_1}{\dbn_0}{e}{\dsetmetax{0}}$, and by preservation we learn that $\vdash_{\schema, \dbn_\textit{init}} \tmap{\dsetmetax{0}}{\reft{\tbn}{\ldots}}{\card{1}{1}}$, so by inversion $\dsetmetax{0} = [\dmeta]$ with $\dmeta = \refval{\idn}{\ldots}$, because evaluation preserves extension (\cref{label-eval-pres-extension}), $\dbn_0$ extends $\dbn_\textit{init}$.

By the same logic from the \textsc{T-insert} case, we can then get
$\deval{\vctxn}{\dbn_i}{\dbn_{i+1}}{\emeta_i}{\dsetmetax{i}}$
and
$\vdash_{\schema, \dbn_\textit{init}} \tmap{\dsetmetax{i}}{\typmeta}{\cardmeta_i}$
and
$\dsetmetax{i} : \tmeta_i \rhd \vsetmetax{i}$
for $1 \leq i \leq n$.

There are then three cases to deal with:
\begin{itemize}
    \item In the event that there is no tuple with first component $\idn$, then we let $\dbn' = \dbn_{n+1}$, and we have $\dbn_{n+1} - \left\{ (\idn, \tbn, \lkclosed\,, \rec{}) \right\} = \dbn'$ with $\dbn' \# \idn$. We can construct the $\dbn''$ required by \textsc{E-update}, but we're just going to ignore it, because we take the first case of the disjunctive rule for \textsc{E-update}. The derivation returns the value sequence $[]$ and the database store $\dbn_{n+1}$.
    \item In the event that there is a tuple $\left(\idn, \tbn, \lkclosed\,, \rec{
\maprec{\lbnprimex{1}}{\vsetmetaprimex{1}}\recsep\ldots\recsep\maprec{\lbnprimex{n}}{\vsetmetaprimex{n}}
}\right) \in \dbn_{n+1}$, then we construct $\dbn'$ as $\dbn_{n+1}$ without that tuple, and we can construct $\dbn''$ as required by \textsc{E-update}. We then ignore it, because the derivation returns the value sequence $[]$ and the database store $\dbn_{n+1}$.
    \item In the event that there is a tuple $\left(\idn, \tbn, \lkopen, \rec{
\maprec{\lbnprimex{1}}{\vsetmetaprimex{1}}\recsep\ldots\recsep\maprec{\lbnprimex{n}}{\vsetmetaprimex{n}}
}\right) \in \dbn_{n+1}$, then we construct $\dbn'$ as $\dbn_{n+1}$ without that tuple, and we can construct $\dbn''$ as required by \textsc{E-update}. The derivation returns the value sequence $[\refval{\idn}{\imaprec{\lbnx{1}}{\dsetmetax{1}},\ldots,\imaprec{\lbnx{n}}{\dsetmetax{n}}}]$ and the constructed updated database store $\dbn''$.
\end{itemize}

\section{Updated IMDBench Performance}
\label{sec-perf}

In \cref{sec-eval}, we described the IMDBench benchmark designed for Gel \cite{whyorms}, and presented \cref{fig:perf} showing a summary of the benchmark's throughput performance across nine implementations of the benchmark: four JavaScript ORMs, two Python ORMs, the Gel client library for Python and JavaScript, and an implementation that queried PostgreSQL directly with the asyncpg client library for PostgreSQL. Summarization was done by capturing the geometric mean across three different benchmarks.

In this appendix, we show the three benchmarks separately, and also include latency information. This section also presents data for three additional implementations of the benchmark that use PostgreSQL queries directly: a Python implementation using the psycopg2 client library, a Node.js implementation using the pg client library, and a Go implementation using the pgx client library. 
These there implementations were excluded from the paper because we're not certain they are fair comparisons: when the benchmark was created in 2022, these libraries did not have adequate support for PostgreSQL arrays, so the implementations use multiple calls to the database, each of which imposed a one-millisecond delay. As of 2025, at least one of the three, the psycopg2 client library for Python, has adequate support for PostgreSQL arrays, so the measurements are no longer representative of the best-possible implementation that uses psychopg2 and hand-tuned PostgreSQL queries.

\Cref{fig:post-perf,fig:movie-perf,fig:user-perf} each show one of the three benchmarks. Gel, which is called ``AnonDB'' in the figures, consistently outperforms all the ORM implementations and does slightly worse than the asyncpg PostgreSQL client library for Python. 
The latency graphs use a box-and-whiskers plot that captures the minimum, 25th-percentile, 50th-percentile (median), 75th-percentile, and 99th percentile latency across the benchmark runs; maximum latency frequently had significant outliers for all implementations.

\begin{figure}
    \centering
    \includegraphics[width=0.48\linewidth]{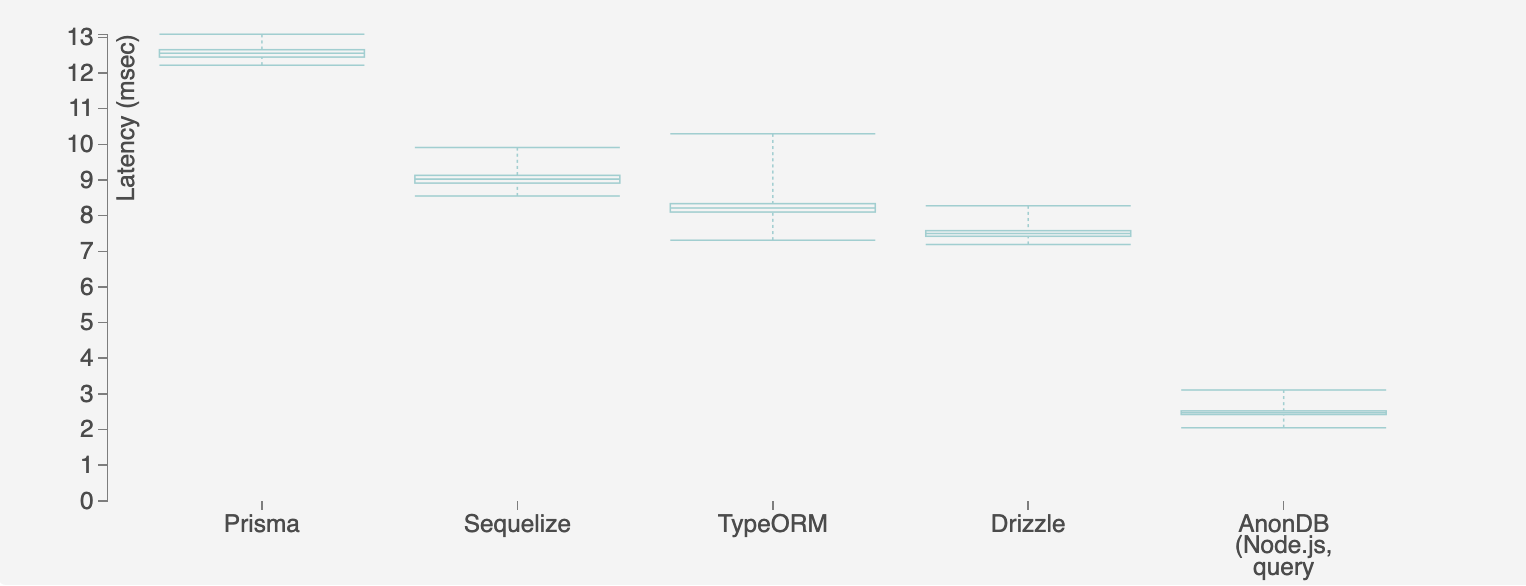}
    \includegraphics[width=0.48\linewidth]{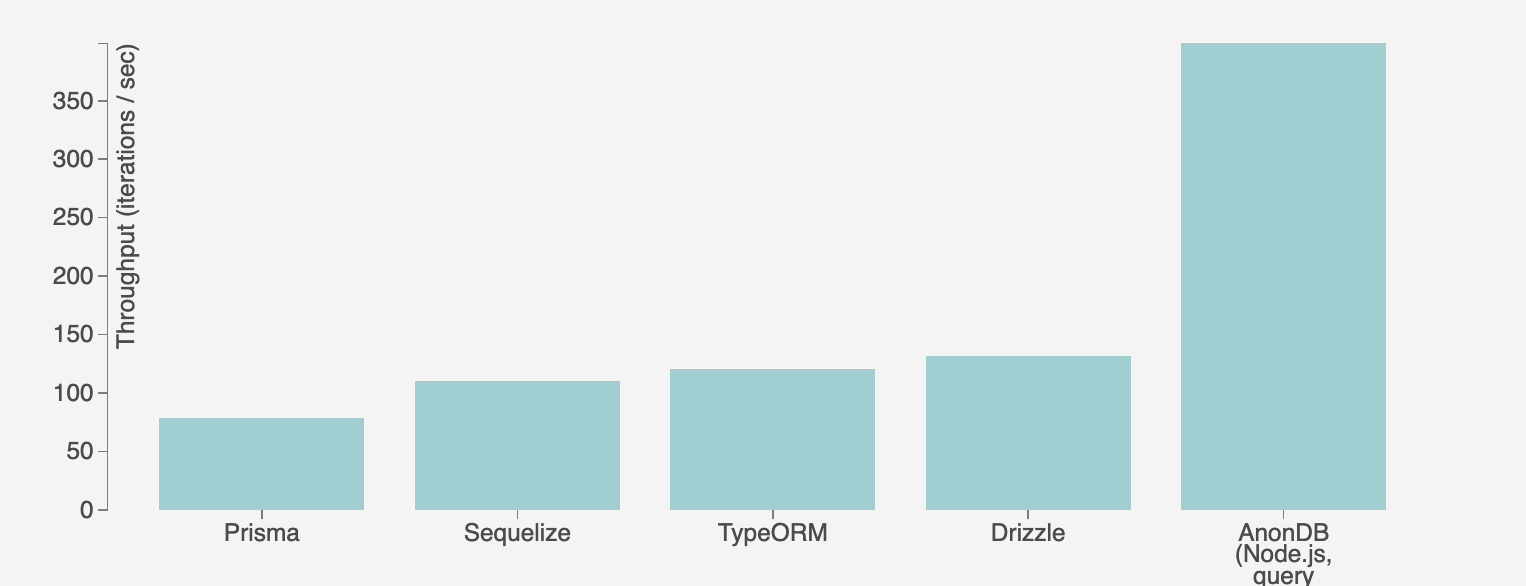}
    \includegraphics[width=0.48\linewidth]{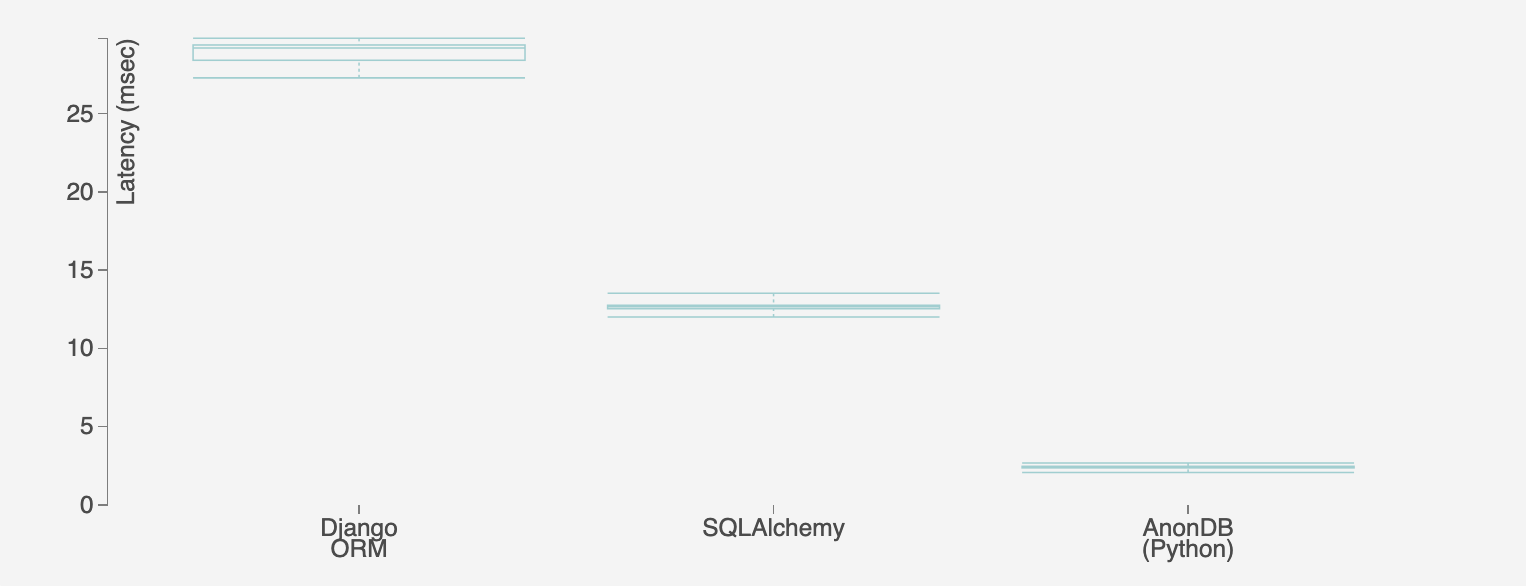}
    \includegraphics[width=0.48\linewidth]{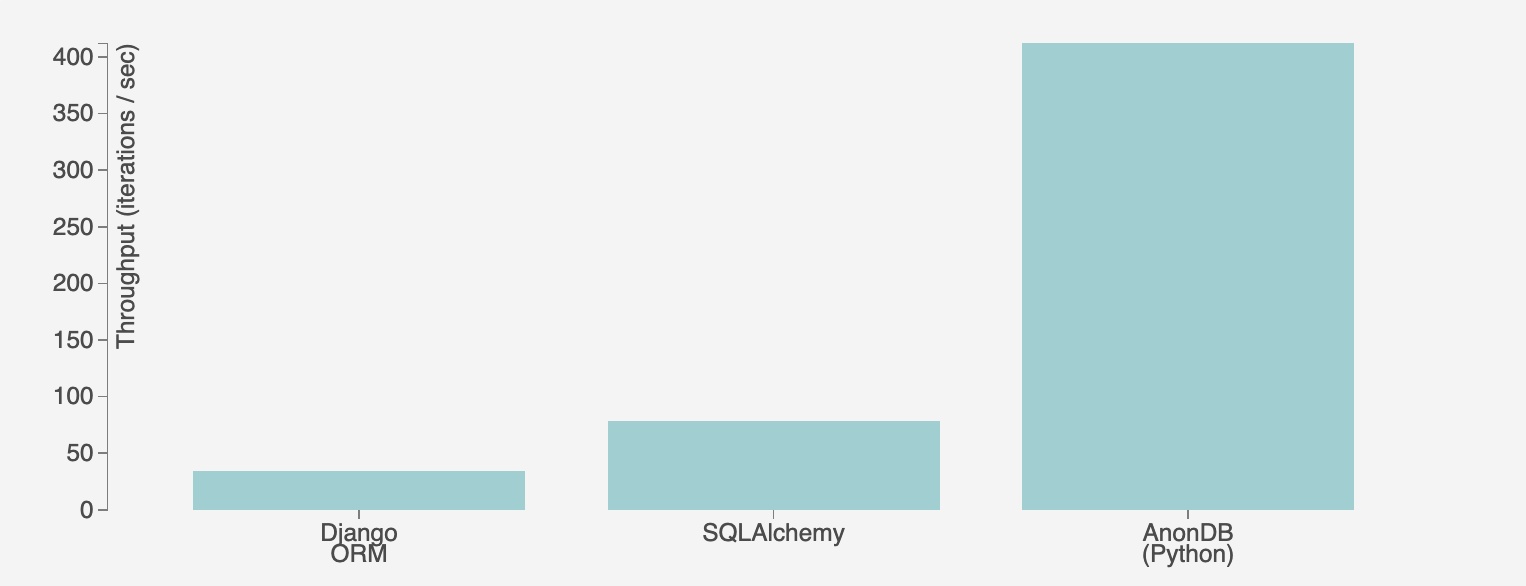}
    \includegraphics[width=0.48\linewidth]{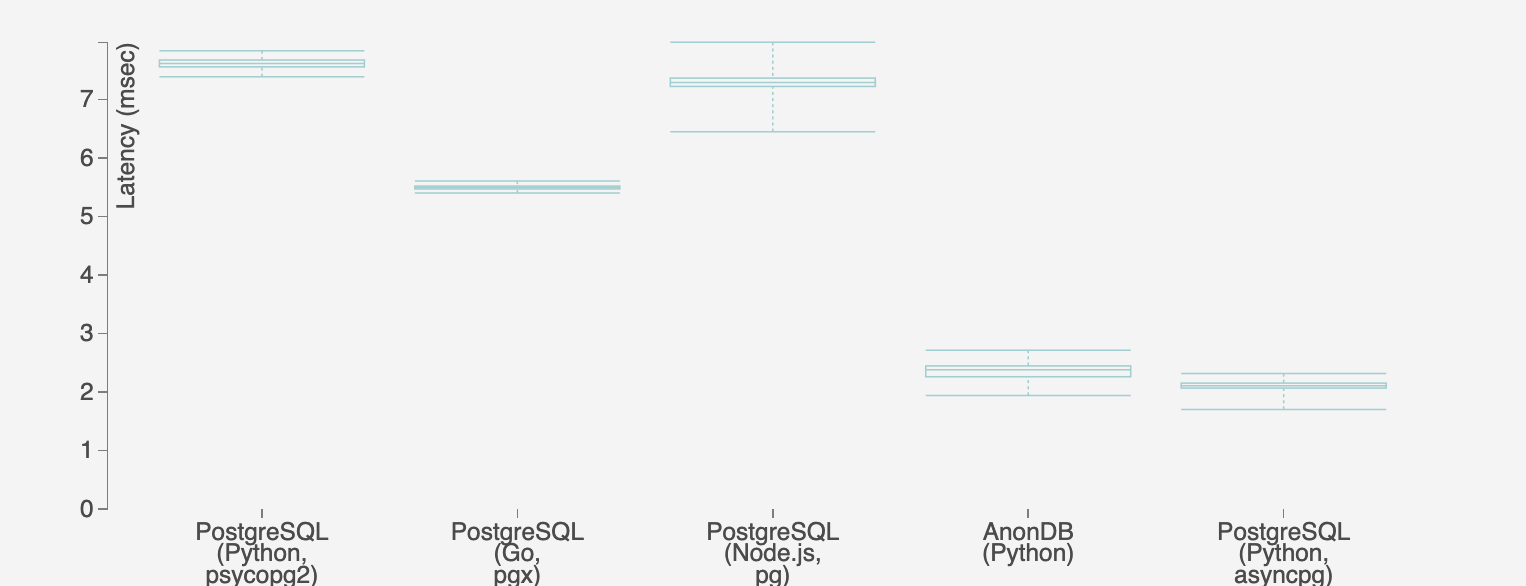}
    \includegraphics[width=0.48\linewidth]{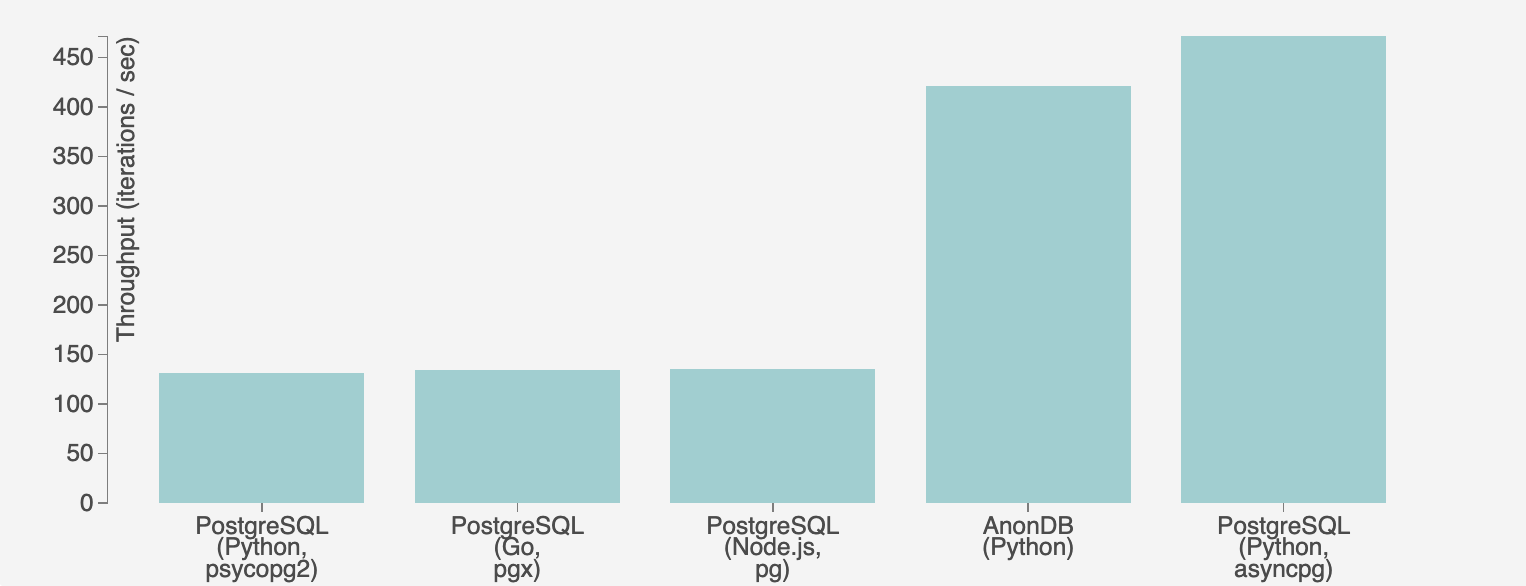}
    \caption{First benchmark: a POST query that inserts a new movie}
    \label{fig:post-perf}
\end{figure}

\begin{figure}
    \centering
    \includegraphics[width=0.48\linewidth]{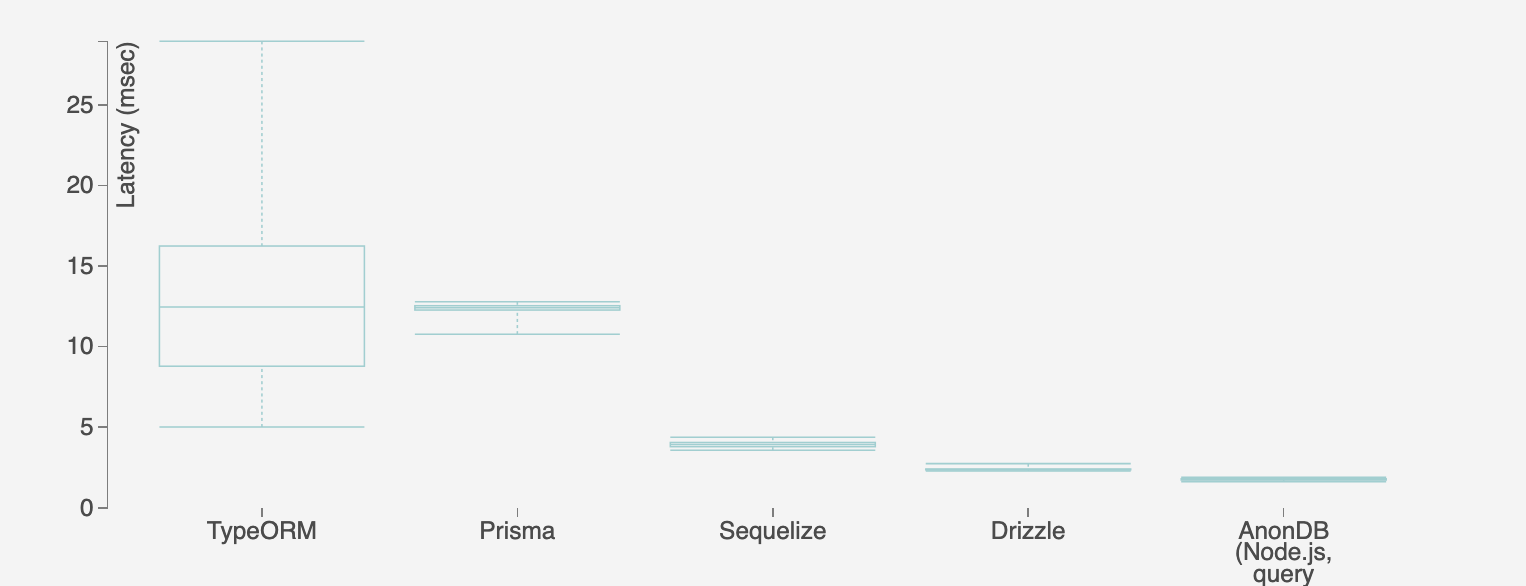}
    \includegraphics[width=0.48\linewidth]{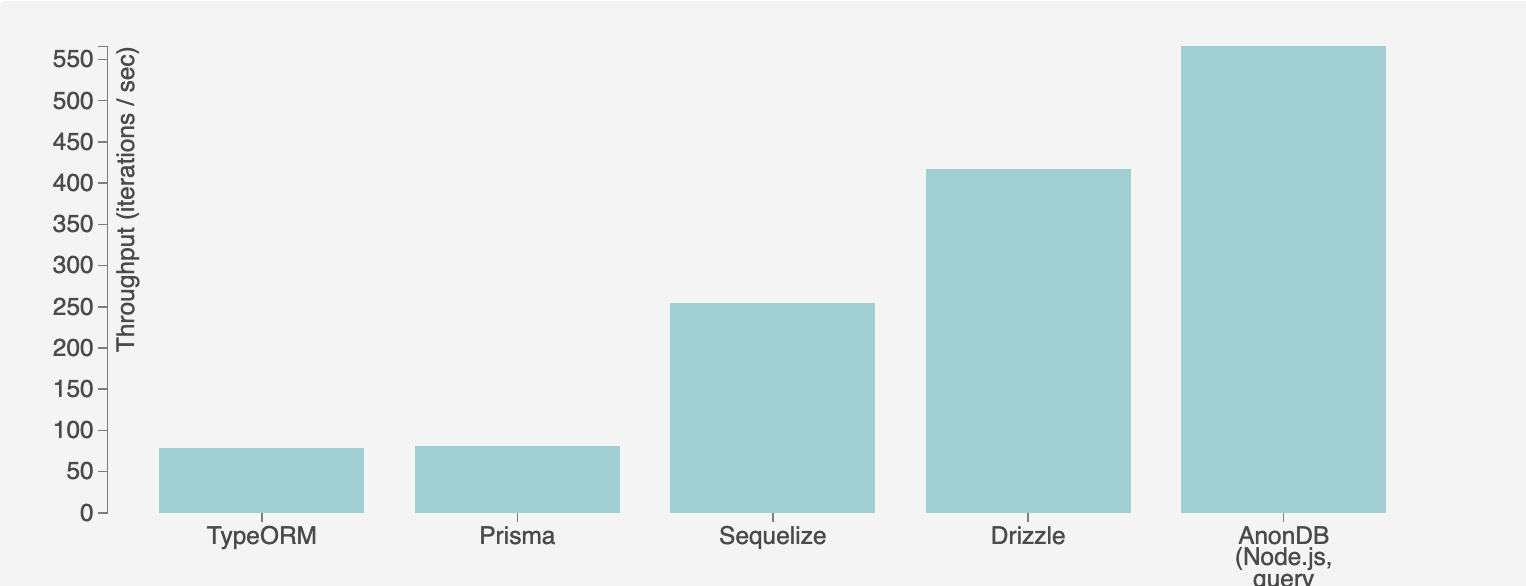}
    \includegraphics[width=0.48\linewidth]{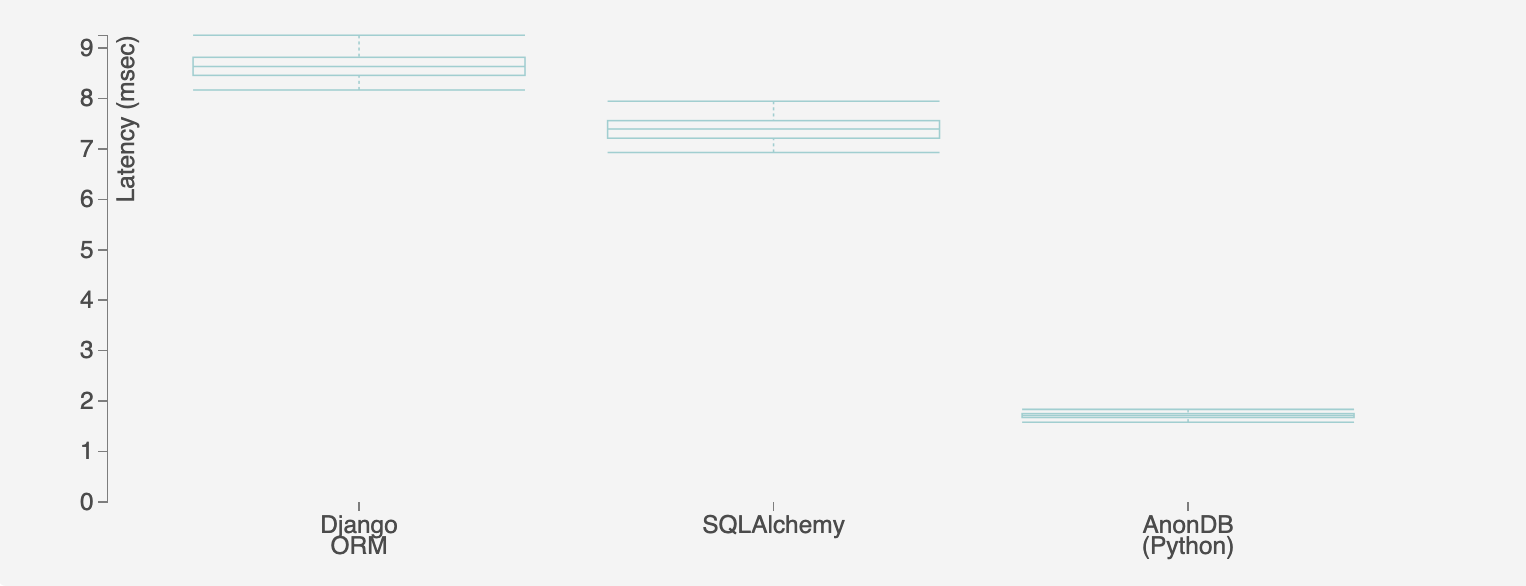}
    \includegraphics[width=0.48\linewidth]{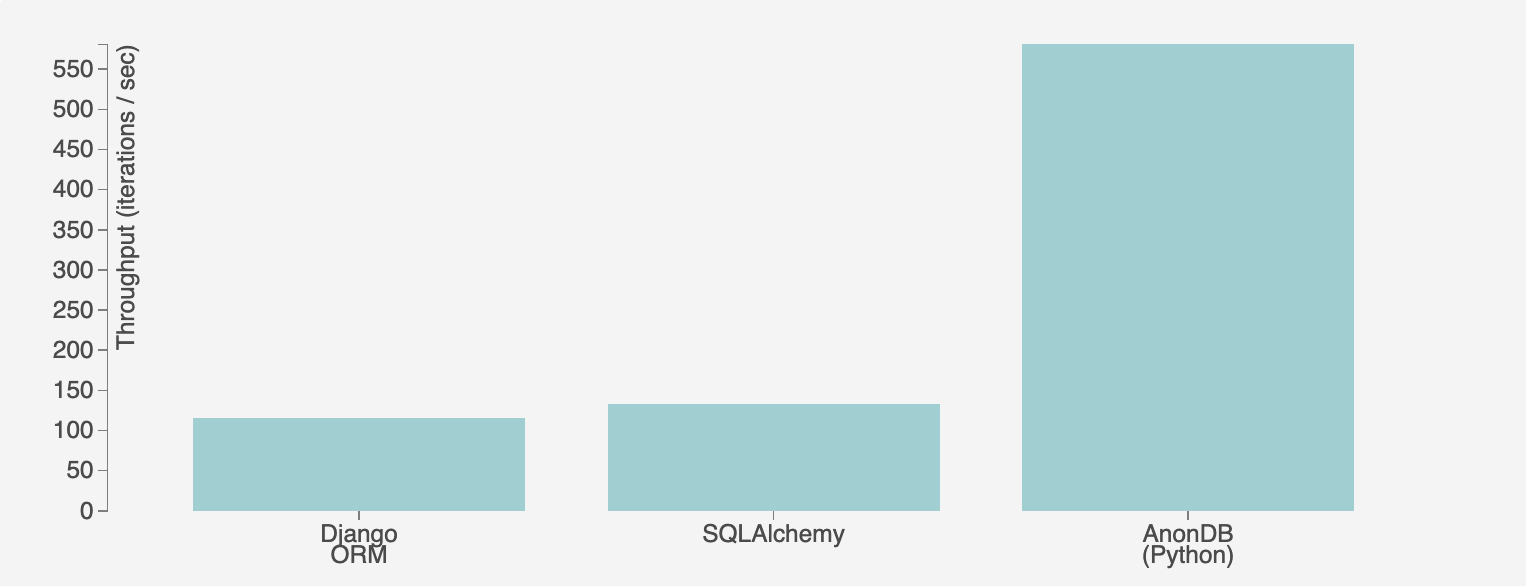}
    \includegraphics[width=0.48\linewidth]{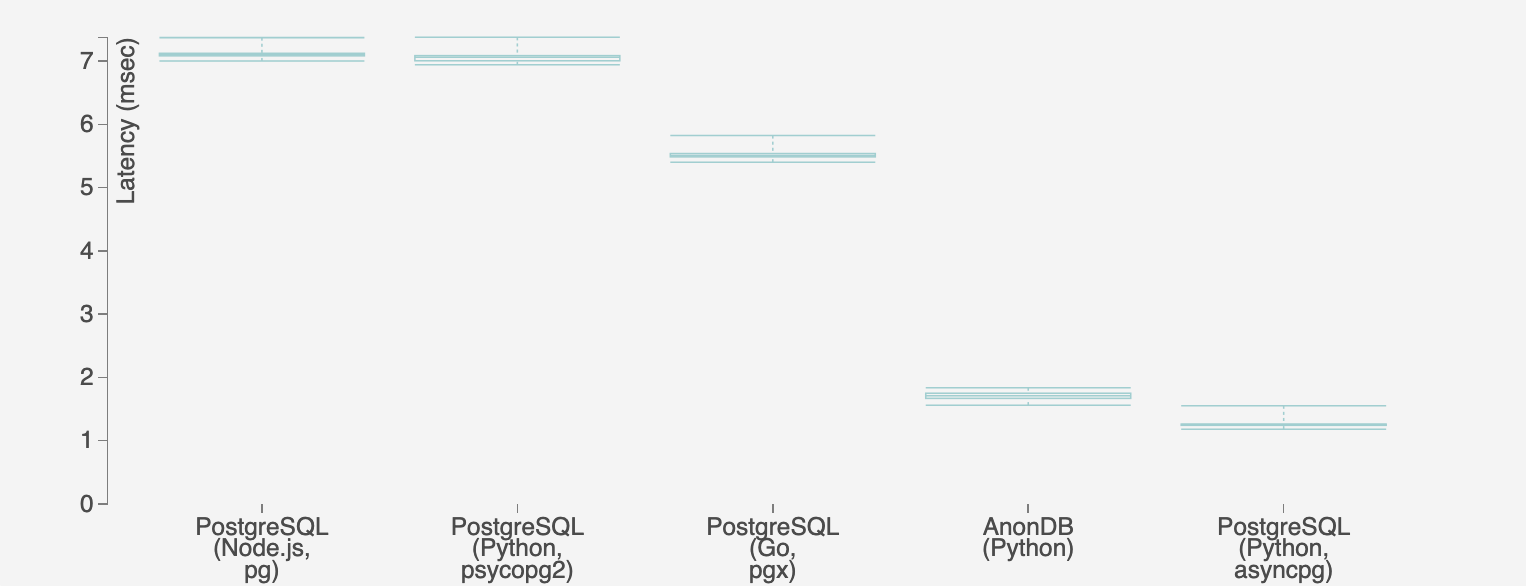}
    \includegraphics[width=0.48\linewidth]{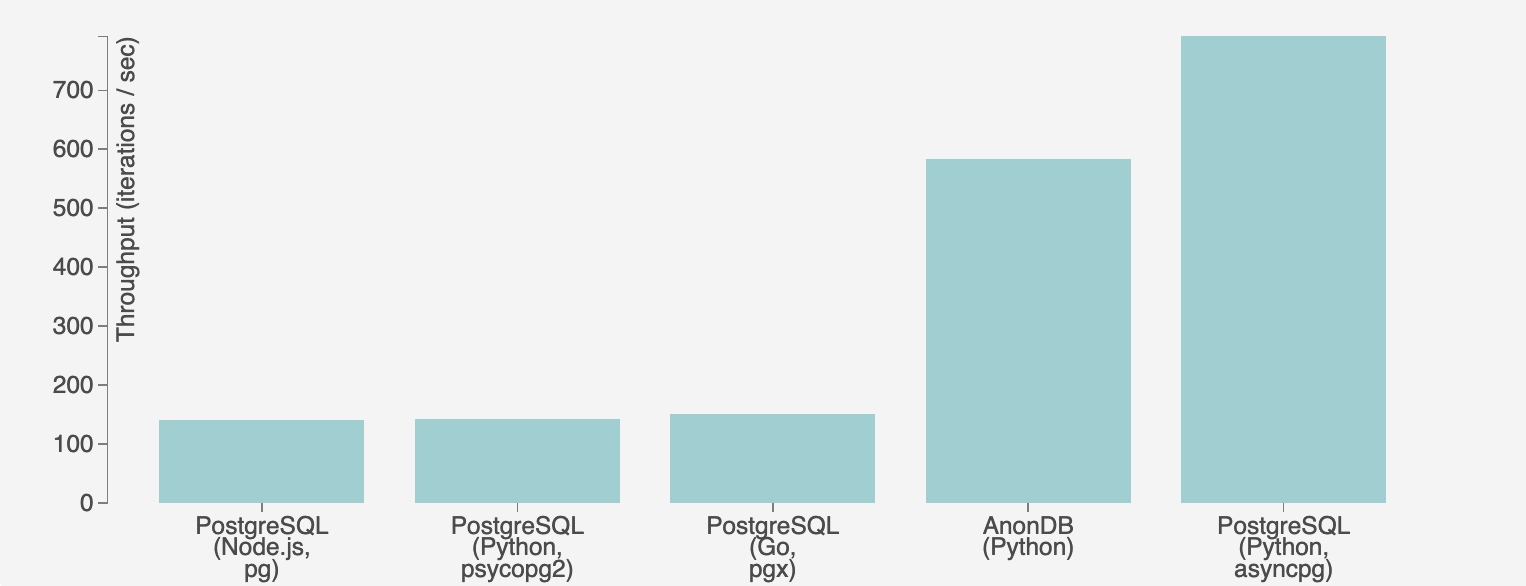}
    \caption{Second benchmark: a GET query that requests data on a movie}
    \label{fig:movie-perf}
\end{figure}

\begin{figure}
    \centering
    \includegraphics[width=0.48\linewidth]{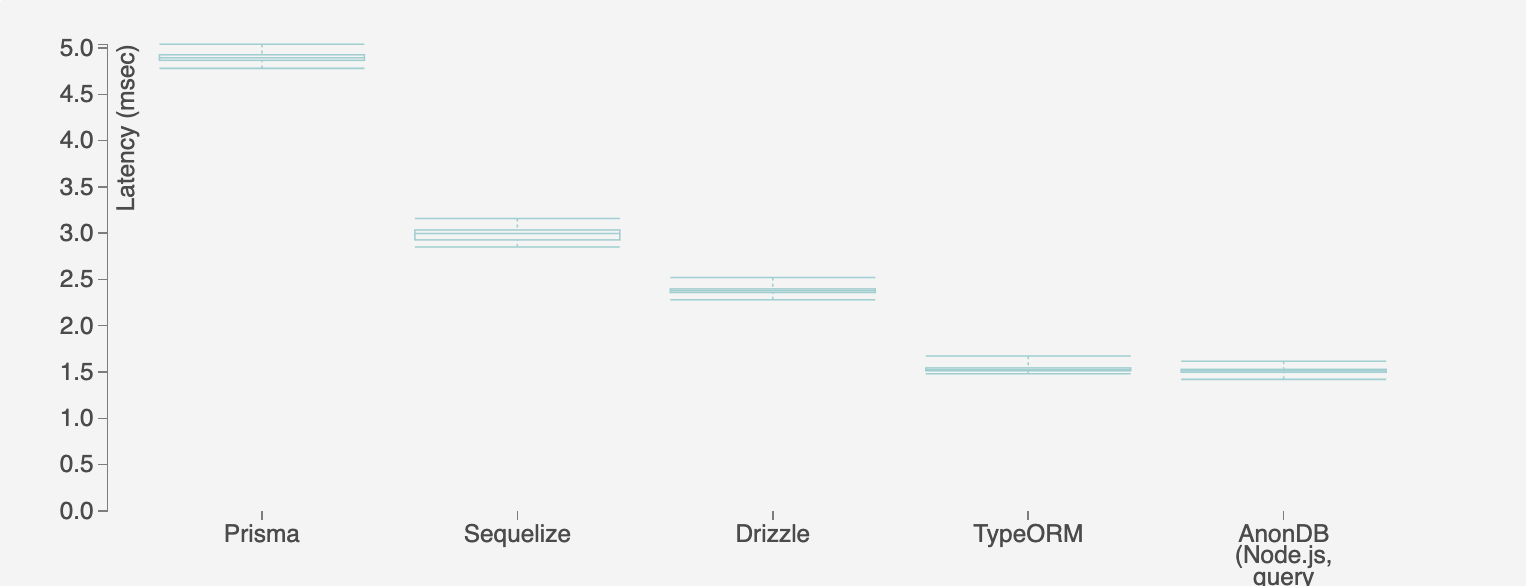}
    \includegraphics[width=0.48\linewidth]{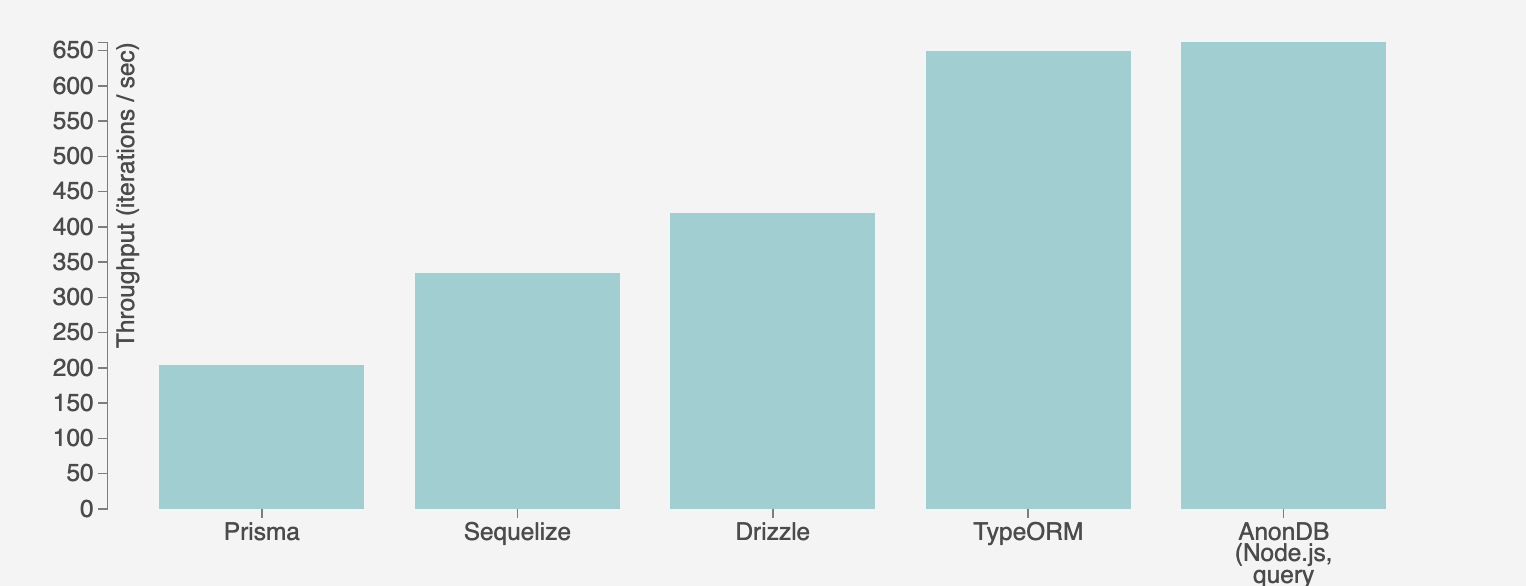}
    \includegraphics[width=0.48\linewidth]{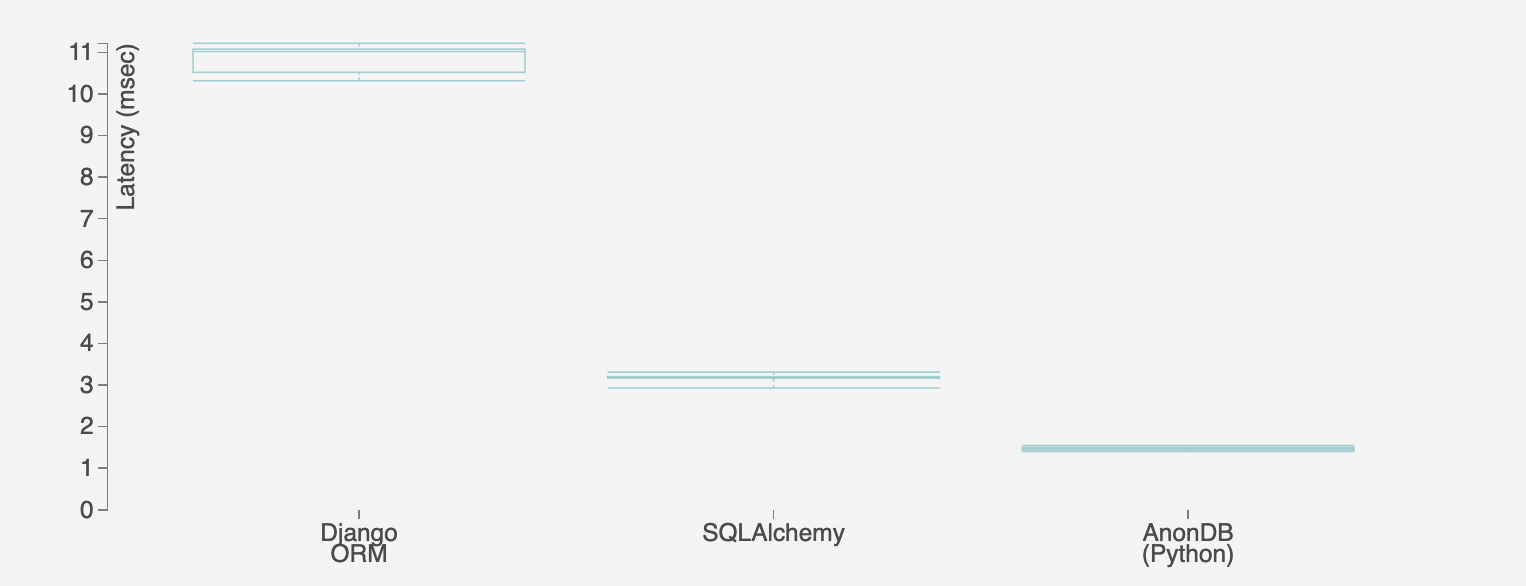}
    \includegraphics[width=0.48\linewidth]{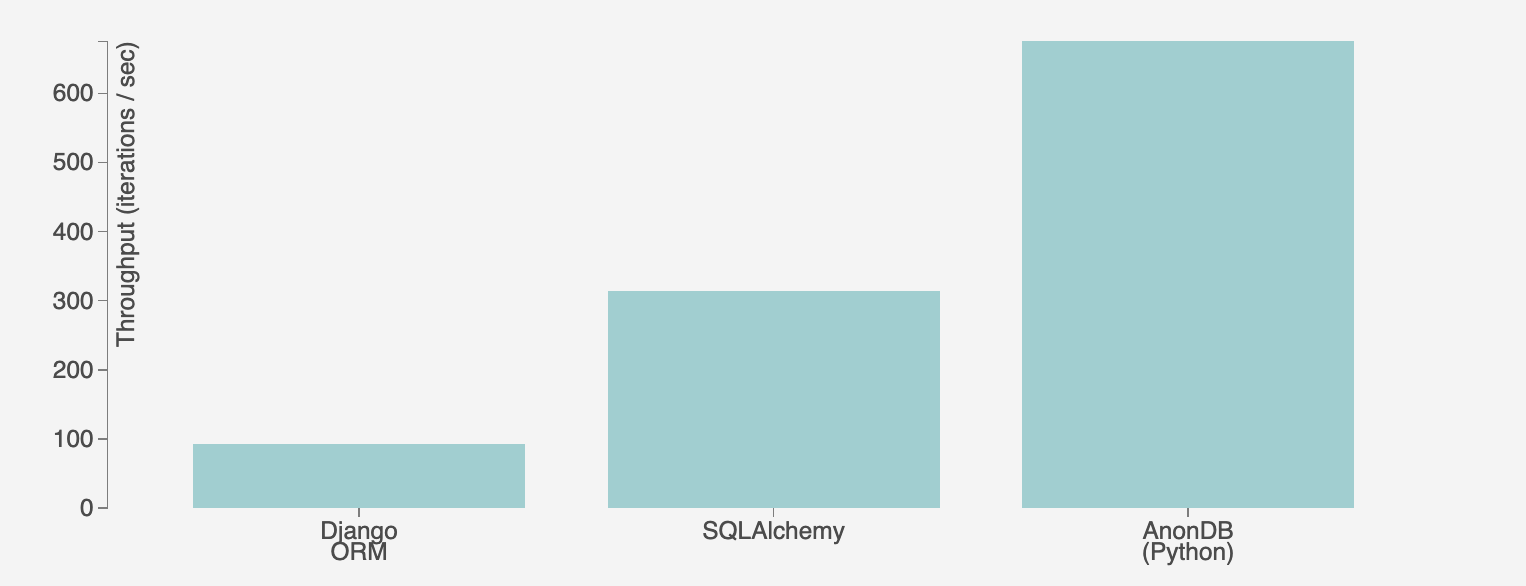}
    \includegraphics[width=0.48\linewidth]{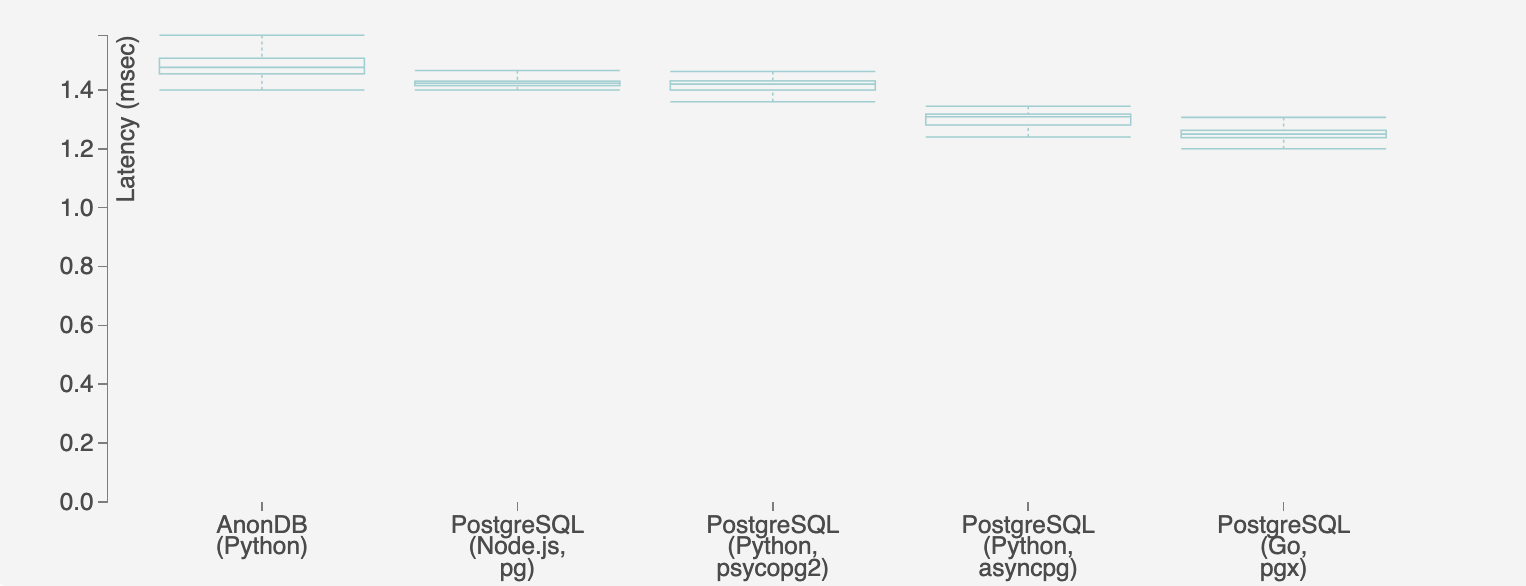}
    \includegraphics[width=0.48\linewidth]{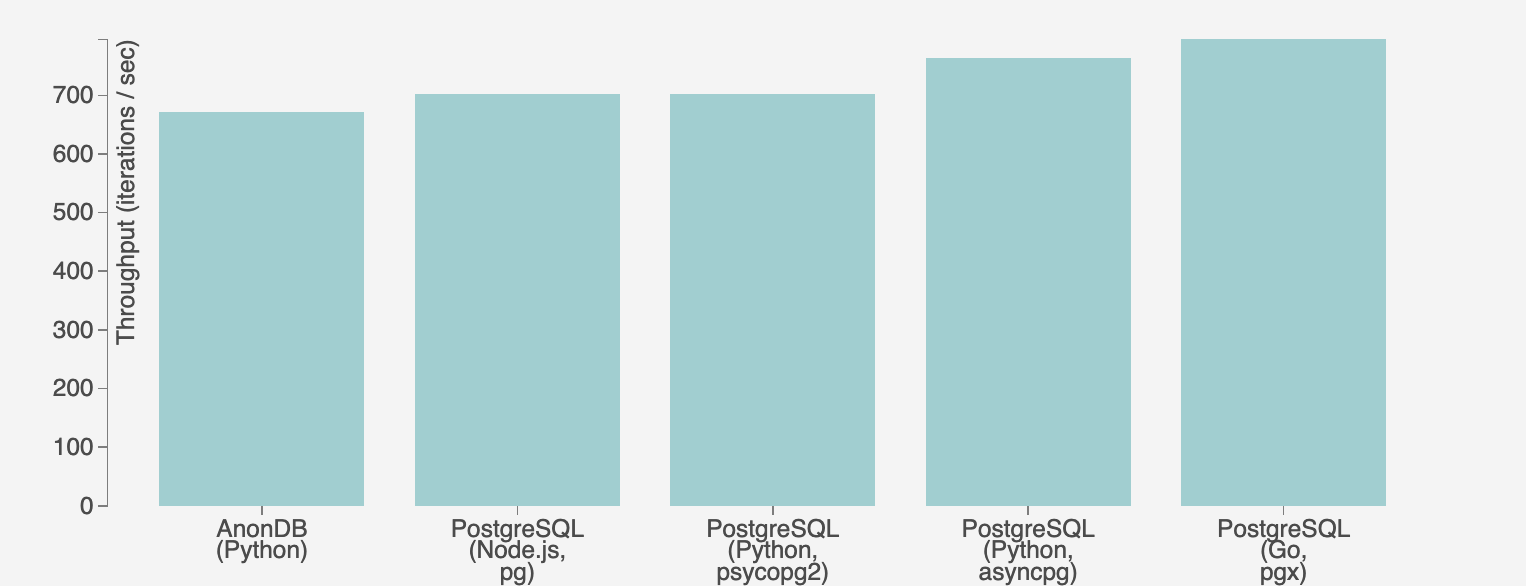}
    \caption{Third benchmark: a GET query that requests information about a user}
    \label{fig:user-perf}
\end{figure}


\end{document}